%% file: main.tex
\documentclass[11pt,a4paper]{article}
\usepackage[lmargin=1.0in,rmargin=1.0in,bottom=1.0in,top=1.0in,twoside=False]{geometry}
\usepackage{hchang}
\usepackage{mathtools}
\usepackage{graphicx}%
\usepackage{enumerate}
\usepackage{tikz}
\usetikzlibrary{calc}
\usetikzlibrary{shapes}
\usetikzlibrary{decorations.pathmorphing}
\usetikzlibrary{decorations.pathreplacing, calligraphy}
\tikzset{snake it/.style={decorate, decoration=snake}}

\usepackage{caption}    %
\usepackage{subcaption} %
\usepackage{thmtools} %
\usepackage{float}
\usepackage[T1]{fontenc}

\usepackage{enumitem}

\usepackage{microtype}
\usepackage{comment}
\usepackage{cleveref}
\usepackage[normalem]{ulem}

\usepackage{stmaryrd}
\usepackage{marvosym}
\usepackage{booktabs}
\usepackage{multirow}
\usepackage{makecell}

\newif\ifanonymous

\ifanonymous%
\else%
\nolinenumbers%
\fi%

\newcommand{\cutgraph}{\textrm{\CutLeft}}

\newtheorem{lemma}{Lemma}[section]
\newtheorem{theorem}[lemma]{Theorem}
\newtheorem{corollary}[lemma]{Corollary}
\newtheorem{definition}[lemma]{Definition}
\newtheorem{claim}[lemma]{Claim}

\newtheorem{remark}{Remark}

\newcommand{\Oh}{\mathcal{O}}
\newcommand{\dist}{\mathsf{dist}}
\newcommand{\prof}[2]{\mathsf{prof}_{#1}[#2]}
\newcommand{\ppath}[1]{\llbracket #1 \rrbracket}
\newcommand{\eps}{\varepsilon}
\newcommand{\real}{\mathbb{R}}

\newcommand{\cut}{\cutgraph}

\newcommand{\DL}{\mathcal{D}}

\newcommand{\dltri}[1]{(p_{#1}, \mathfrak{b}_{#1}, \mathfrak{t}_{#1})}
\newcommand{\dlpt}[1]{\mathfrak{#1}}

\newcommand{\LDL}{L^{\text{\#}}}

\setuptodonotes{size=\normalsize}

\newcommand{\Mpath}[2]{\llbracket #1#2 \rrbracket}
\newcommand{\inst}{\mathcal{I}}
\newcommand{\instind}[1]{\inst^{#1}}
\newcommand{\instexcl}[1]{\instind{\neg #1}}
\newcommand{\Gind}[1]{G^{#1}}
\newcommand{\Gexcl}[1]{\Gind{\neg #1}}
\newcommand{\Udist}{\widehat{\dist}}
\newcommand{\Uprof}[2]{\widehat{\mathsf{prof}}_{#1}[#2]}

\begin{document}

\ifanonymous
\author{Anonymous}
\else
\author{
  Kacper Kluk%
  \thanks{Institute of Informatics, University of Warsaw, Poland. \texttt{k.kluk@uw.edu.pl}}
  \and%
  Hung Le%
  \thanks{University of Massachusetts, Amherst, USA. \texttt{hungle@cs.umass.edu}}
  \and%
  Wojciech Nadara%
  \thanks{Institute of Informatics, University of Warsaw, Poland. \texttt{w.nadara@uw.edu.pl}}
  \and%
  Marcin Pilipczuk%
  \thanks{Institute of Informatics, University of Warsaw, Poland. \texttt{m.pilipczuk@uw.edu.pl}}
  \and%
  Hector Tierno%
  \thanks{University of Massachusetts, Amherst, USA. \texttt{htierno@umass.edu}}
  \and%
  Vinayak%
  \thanks{University of Massachusetts, Amherst, USA. \texttt{vvinayak@umass.edu}}
}
\fi
\title{A Polynomial Coreset for Furthest Neighbor in Planar Metrics%
\ifanonymous%
\else%
\thanks{K.K. and Ma.P. are supported by Polish National Science Centre SONATA BIS-12 grant number 2022/46/E/ST6/00143. W.N. is supported by European Union’s Horizon 2020 research
and innovation programme, grant agreement No. 948057 — BOBR. H.L., Vinayak, and H.T. are  supported by NSF grant CCF-2517033 and NSF CAREER Award CCF-2237288.}
\fi%
}

\date{}

\maketitle

\begin{abstract}

A furthest neighbor data structure on a metric space $(V,\dist)$ and a set $P \subseteq V$ answers the following query: given $v \in V$, output $p \in P$ maximizing $\dist(v,p)$; in the approximate version, it is allowed to report any $p \in P$ with $\dist(v,p) \geq (1-\eps)\max_{p' \in P} \dist(v,p')$ for an accuracy parameter $\eps \in (0,1)$. 
A particular type of approximate furthest neighbor data structure is an $\eps$-coreset: a small subset $Q \subseteq P$ such that for every query $v \in V$ there is a feasible answer
$p \in Q$. 

Our main result is that in planar metrics there always exists an $\eps$-coreset
for furthest neighbors of size bounded polynomially in $(1/\eps)$. 
This improves upon an exponential bound of Bourneuf and Pilipczuk~\cite{BP25}
and resolves an open problem of de Berg and Theocharous~\cite{BT24}
for the case of polygons with holes.

On the technical side, we develop a connection between $\eps$-coreset for furthest neighbors and an invariant of a metric space that we call an \EMPH{$\eps$-comatching index} --- a sibling of \text{$\eps$-(semi-)ladder} index, a.k.a,  $\eps$-scatter dimension, as defined by Abbasi et al~\cite{ABBCGKMSS23}. While the $\eps$-(semi-)ladder index of planar metrics admits an \EMPH{exponential} lower bound, we show that the $\eps$-comatching index of planar metrics is \EMPH{polynomial}, all in $1/\eps$. The exponential separation between $\eps$-(semi-)ladder and $\eps$-comatching is rather surprising, and the proof is the main technical contribution of our work.

Building on our new technical ideas, we provide two extensions:
\begin{enumerate}
    \item We show an $\eps$-coreset for $k$-center in planar metrics of size polynomial in $(1/\eps)$ for \EMPH{fixed $k$} (furthest neighbor coreset is a special case of $k$-center when $k=1$). We complement this result with an exponential coreset lower bound of $2^{\Omega(1/\eps)}$ for $k$-center when $k = 1/\eps$. In the regime of fixed $k$, our result substantially improves the recent $\eps$-coreset for $k$-center of Bourneuf and Pilipczuk~\cite{BP25} who showed a coreset bound $O(k)^{\poly(1/\eps)}$.

    \item We extend all our results to metrics induced by graphs of bounded Euler genus, with a polynomial dependency on the genus in the bound on the size of the coreset. We complement these upper bounds with a lower bound construction showing that, going beyond bounded genus graphs, e.g., minor-free graphs,  the comatching index and $\eps$-coreset for furthest neighbors must be exponential in $1/\eps$. 
  
\end{enumerate}
\end{abstract}

\thispagestyle{empty}

\newpage
\thispagestyle{empty}
\tableofcontents
\thispagestyle{empty}
\newpage

\clearpage
\setcounter{page}{1}

\input{introduction}

\input{prelim}

\input{comatching2coreset}

\input{comatching}
\input{reductions_bounds}
\input{kcentre}

\input{OtherApps}

\input{genus}
\input{lower-bound}

\subsection*{AI Disclosure}
The example of \Cref{thm:kcenter-lb2}, presented in \Cref{ss:lb2}, has been found by ChatGPT 5.4 Pro. We have found the example of \Cref{thm:kcenter-lb} by hand and asked ChatGPT to look for an improvement in the regime where $k$ and $d$ differ significantly. ChatGPT autonomously found the ``local gadget'' used in the construction. 
We have verified the correctness of the construction; the write-up is also done by the authors. 

\ifanonymous%
\else%
\paragraph{Acknowledgments.} H.L. would like to thank Lazar Milenković, Arnold Filtser, and Omrit Filtser for helpful discussions in the early stage of this work.
\fi%
\bibliographystyle{plain}
\bibliography{ref}

\end{document}

%% file: introduction.tex
\section{Introduction}

\paragraph{Furthest neighbor.}  Nearest neighbor and furthest neighbor searches are basic and well-studied proximity problems in computational geometry.  In the furthest neighbor problem, the goal is to preprocess a given set $P$ of
points in a metric space $(V,\dist)$ to quickly answer \emph{furthest neighbor queries}: for a point $v \in V$, return a point $p \in P$ that maximizes $\dist(v,p)$~\cite{AFW88,MMR01,CGBook08,PSSS17,BT24}.
In the approximate version, an accuracy parameter $\eps > 0$ is fixed 
and we want to return any $p \in P$ with $\dist(v,p) \geq (1-\eps)\max_{p' \in P} \dist(v,p')$. The furthest neighbor problem is closely related to computing the diameter of point sets.  In Euclidean spaces, computing approximate furthest neighbors can be efficiently reduced to computing approximate nearest neighbors~\cite{GIV01}.  In the graph setting, the distance to the furthest neighbor of a vertex is called the \EMPH{eccentricity } of the vertex. Computing the diameter and eccentricities is a fundamental graph problem that has been studied for decades.

A natural approach to the furthest neighbor problem is to identify,
in the preprocessing phase, a small set $Q \subseteq P$ that correctly represents $P$:
for any $v \in V$, there is a valid answer $p \in Q$. Then, answering a query boils down to a search in $Q$; if $Q$ is small, even an exhaustive search can be efficient. 
Such a set $Q$ is a \EMPH{coreset} for the furthest neighbor problem. 

\begin{definition}[Furthest Neighbor $\eps$-Coreset]\label{def:furthest}
  Given a metric space $(V,\dist)$, a set $P \subseteq V$ and a~real $\varepsilon > 0$, a \EMPH{furthest neighbor $\eps$-coreset} is a subset $Q \subseteq P$ such that
  for every $v \in V$ there exists $q \in Q$ such that: 
  \[ \dist(v,q) \geq (1-\varepsilon) \max \{\dist(v,p)~|~p \in P\}. \]
\end{definition}

In general, a \EMPH{coreset} of a dataset is a small subset such that a solution of a problem on the coreset gives a good (approximate) solution for the problem on the original dataset. Coreset has been a very well-studied topic since it offers a simple approach for solving problems in massive datasets: simply construct a coreset of the (big) input and solve the problems on the coreset. A particularly productive line of research is on coresets for \emph{clustering problems}, specifically $k$-means, $k$-median, and $(k,z)$-clustering, for all well-studied metrics, e.g.,  general metrics~\cite{FL11,CSS21,uniform22}, Euclidean metrics~\cite{FL11,FSS20,CSS21,uniform22,CDRSS25}, doubling metrics~\cite{FL11,Chen09,HJLW18}, and planar/minor-free metrics~\cite{BBHJKW20,BJKW21,CDRSS25}. For these problems, good coresets (of polylogarithmic size or smaller) are known to exist, and known bounds are (nearly) optimal~\cite{BBHJKW20}. 

Unfortunately, without any additional assumptions, the question for a coreset for the furthest neighbor problem is uninteresting:
There exists a very simple set of $n$ points such that any coreset for furthest neighbor of the point set has size $\Omega(n)$~\cite{BP25}. 
For Euclidean and doubling metrics, a simple packing bound gives a coreset of size $\eps^{-O(d)}$, which is optimal up to a small constant factor in the exponent~\cite{CGJK25}. Therefore, unlike other clustering counterparts, coresets for furthest neighbor suffer from the curse of dimensionality: the size of the coresets grows exponentially with the dimension.

In planar metrics, which are the shortest path metrics of edge-weighted planar graphs, Bourneuf and Pilipczuk~\cite{BP25} construct an $\eps$-coreset for furthest neighbors\footnote{Their coreset results are for $k$-center in which the furthest neighbor coreset is a special case, specifically when $k = 1$.} with an \EMPH{exponential size} of $2^{\poly(1/\eps)}$ where $\poly(\cdot)$ is a polynomial function. Despite its huge size, this result shows that it is possible to construct an  $\eps$-coreset for furthest neighbors whose size is independent of $n$.  As we will discuss in more detail below, there is a crucial bottleneck in their technique, and the exponential dependency on $1/\eps$ is the barrier. In this work, we overcome this barrier and construct the first  $\eps$-coreset of size polynomial in $1/\eps$ for planar metrics.

\begin{theorem}\label{thm:1center} Given a point set $P$ in a planar metric $(V,\dist)$, and a parameter $\eps\in(0,1)$, there exists a furthest neighbor $\eps$-coreset of size $\poly(1/\eps)$  for $P$ that can be constructed in polynomial time.
\end{theorem}

Theorem~\ref{thm:1center} implies similar results for 
many geometric settings in the plane.
Recently, de Berg and Theocharous~\cite{BT24} studied $\eps$-coreset for simple polygon, and constructed an $\eps$-coreset of size $O(1/\eps^2)$. They asked if such a coreset exists for \EMPH{polygon with holes}, and suggested that a different approach is needed. 
Theorem~\ref{thm:1center} resolves their question in the
affirmative.
(We discuss formally how \Cref{thm:1center} yields an $\eps$-coreset for furthest neighbor
in polygon with holes in \Cref{subsec:other}.)

\paragraph{Structural invariants of metric spaces.}
Consider the following very simple greedy algorithm for constructing a~furthest neighbor $\eps$-coreset: as long as there exists a point $v\in V$ that has no $(1-\eps)$-approximate furthest neighbor in $Q$, then add the furthest neighbor $q$ of $v$ in $P$ to $Q$; $v$ is called the \EMPH{witness} for (the addition of) $q$. To bound the size of $Q$, let $\ell = |Q|$ and $q_1, \ldots, q_{\ell}$ be the points in $Q$ ordered by the greedy algorithm, i.e., $q_{j}$ is added after $q_{i}$ whenever $i < j$. Let  $L = \{(q_1,v_1), \ldots, (q_{\ell}, v_{\ell})\}$ be the sequence of pairs where $v_{i}$ is witness of $q_i$, $i \in [\ell]$. 
By a simple bucketing trick, one could assume that the distances $\dist(q_i,v_i)$
are roughly the same, for simplicity, say equal to $R$. (This trick incurs a loss of $1/\eps$ factor in the size of the final $\eps$-coreset, which is considered negligible here.) As $q_j$ is added after $q_1,\ldots, q_{j-1}$, we have $\dist(q_i,v_j) < (1-\eps)R$ for every $i < j$. 

Abbasi et al.~\cite{ABBCGKMSS23} identified the maximum length of such a structure $L$
as a critical invariant of a metric space that governs the complexity of a natural algorithmic
approach to a wide family of clustering problems and called it an \EMPH{$\eps$-scatter}. Bourneuf and Pilipczuk~\cite{BP25} observed that the structure $L$ resembles the notion of a \EMPH{semi-ladder} in 
the theory of structurally sparse graphs~\cite{FPST18,Sokolowski21,AA14}.

The main focus of this theory are classes of graphs
that can be defined using first-order logic
in classes of sparse graphs~\cite{DBLP:conf/icalp/BonnetBEGMPPT25,DBLP:journals/lmcs/Dreier23,DBLP:conf/focs/DreierEMMPT24,DBLP:conf/lics/DreierGKPT22,DBLP:journals/corr/abs-2601-14906,DBLP:conf/stoc/DreierMS23,DBLP:conf/icalp/DreierMST23,DBLP:conf/stoc/DreierMT24,DBLP:conf/stoc/DreierT25,DBLP:journals/tocl/GajarskyKNMPST20,DBLP:conf/icalp/GajarskyMMOPPSS23}, such as planar graphs,
graphs excluding a fixed minor, or, more generally,
bounded expansion graphs and nowhere dense graphs~\cite{sparsity}.
This area experienced rapid growth in the last decade,
with the main goal being a full understanding of the limits
of efficient first-order model checking algorithms. 
We refer to a recent survey~\cite{DBLP:journals/corr/abs-2501-04166} for a broader introduction.

Let us make a small detour into the theory of structurally sparse graphs.
Let $H$ be a simple unweighted graph. Let $L = (p_i,q_i)_{i=1}^\ell$ be a sequence
of vertices in $H$. We say that $L$ is
\begin{description}
\item[a semi-ladder] if $p_iq_i \notin E(H)$ for every $i \in [\ell]$, but $p_iq_j \in E(H)$
for every $1 \leq i < j \leq \ell$;
\item[a ladder] if for every $i,j \in [\ell]$ it holds that $p_iq_j \in E(H)$ if and only if $i < j$;
\item[a comatching] if for every $i,j \in [\ell]$ it holds that $p_iq_j \in E(H)$ if and only if $i \neq j$. 
\end{description}
Note that a semi-ladder does not specify adjacency between $p_i$ and $q_j$ for $i > j$
and none of the above definitions cares about the adjacencies between the $p$-vertices and between the $q$-vertices. A simple Ramsey argument shows that a huge semi-ladder contains a large ladder or a large comatching.

In the setting of structurally sparse graphs, the graph $H$ is first-order interpreted
in some other graph $G$, that is, its edges represent pairs $(u,v)$ of vertices of $G$ satisfying
a fixed first-order formula $\phi(u,v)$.
A classic ``benchmark'' formula is ``the distance between $u$ and $v$ is at most $d$''
for a constant $d$. Following this analogy, 
in our ``metric'' setting, the graph $H$ represents
pairs of vertices that are ``close'': whose distance is at most $(1-\eps)R$. 
This gives the following definition (see Figure~\ref{fig:comatching-ladder}).

\begin{figure}[ht]
    \centering
    \begin{subfigure}{0.32\textwidth}
        \centering
        \begin{tikzpicture}[scale=1, every node/.style={circle, draw, fill=white, inner sep=1.5pt}]
            \foreach \i in {1,...,5}{
                \node (p\i) at (\i*1,3) {$p_{\i}$};
                \node (q\i) at (\i*1,0) {$q_{\i}$};
            }       
            \foreach \i in {1,...,5}{
                \draw[dashed, thick] (p\i) -- (q\i);
            }       
            \foreach \i in {1,...,5}{
                \foreach \j in {1,...,5}{
                    \ifnum\j<\i
                        \draw[] (p\i) -- (q\j);
                    \fi
                }
            }
        \end{tikzpicture}
        \caption{$\eps$-semi-ladder}
        \label{fig:semi-ladder}
    \end{subfigure}
    \begin{subfigure}{0.32\textwidth}
        \centering
        \begin{tikzpicture}[scale=1, every node/.style={circle, draw, fill=white, inner sep=1.5pt}]
            \foreach \i in {1,...,5}{
                \node (p\i) at (\i*1,3) {$p_{\i}$};
                \node (q\i) at (\i*1,0) {$q_{\i}$};
            }       
            \foreach \i in {1,...,5}{
                \draw[dashed, thick] (p\i) -- (q\i);
            }       
            \foreach \i in {1,...,5}{
                \foreach \j in {1,...,5}{
                    \ifnum\j<\i
                        \draw[] (p\i) -- (q\j);
                    \fi
                    \ifnum\j>\i
                        \draw[dashed,thick] (p\i) -- (q\j);
                    \fi
                }
            }
        \end{tikzpicture}
        \caption{$\eps$-ladder}
        \label{fig:ladder}
    \end{subfigure}
    \begin{subfigure}{0.32\textwidth}
        \centering
        \begin{tikzpicture}[scale=1, every node/.style={circle, draw, fill=white, inner sep=1.5pt}]
            \foreach \i in {1,...,5}{
                \node (p\i) at (\i*1,3) {$p_{\i}$};
                \node (q\i) at (\i*1,0) {$q_{\i}$};
            }
            \foreach \i in {1,...,5}{
                \draw[dashed, thick] (p\i) -- (q\i);
            }
            \foreach \i in {1,...,5}{
                \foreach \j in {1,...,5}{
                    \ifnum\i=\j\relax
                    \else
                        \draw[] (p\i) -- (q\j);
                    \fi
                }
            }
        \end{tikzpicture}
        \caption{$\eps$-comatching}
        \label{fig:comatching}
    \end{subfigure}
    \caption{The dotted line indicates distance $\geq R$, and the solid line indicates distance $\leq(1-\eps)R$.}
    \label{fig:comatching-ladder}
\end{figure}
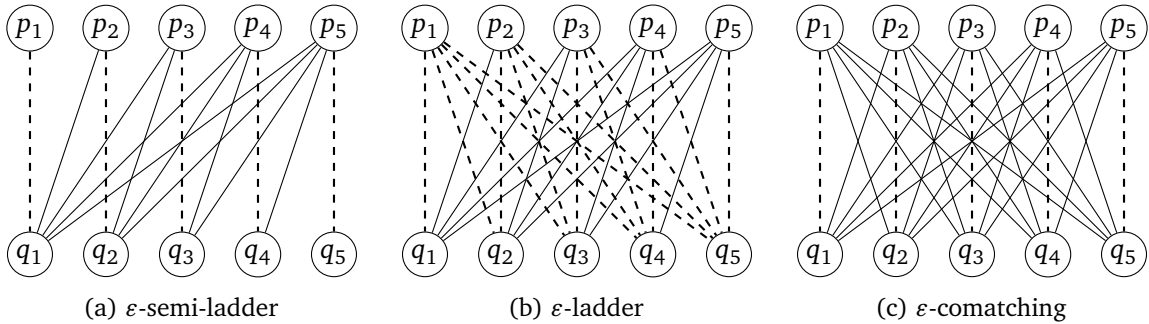

\begin{definition}[metric ladders]\label{def:metric-ladder}
Let $\eps > 0$ be an accuracy parameter and let $(V,\dist)$ be a metric space.
A sequence $L = (p_i,q_i)_{i=1}^\ell$ of points in $V$ is 
\begin{description}
    \item[an $\eps$-semi-ladder] if for every $i \in [\ell]$ we have $\dist(p_i,q_i) > R$
    while for every $1 \leq i < j \leq \ell$ we have $\dist(p_i,q_i) \leq (1-\eps)R$ for some $R > 0$;
    \item[an $\eps$-ladder] if for every $i,j \in [\ell]$ we have $\dist(p_i,q_j) \leq (1-\eps)R$
    if $i < j$ and $\dist(p_i,q_j) > R$ otherwise, for some $R > 0$;
    \item[an $\eps$-comatching] if for every $i \in [\ell]$ we have $\dist(p_i,q_i) > R$
    while for every $i,j \in [\ell]$, $i \neq j$ we have $\dist(p_i,q_i) \leq (1-\eps)R$ for some $R > 0$.
\end{description}
\end{definition}
The notion of an $\eps$-semi-ladder concides with the notion of an $\eps$-scatter
of~\cite{ABBCGKMSS23}, but we prefer the former name as the name ``semi-ladder'' is well established in the theory of structurally sparse graphs.

Following the nomenclature of structurally sparse graphs, we call the maximum length of an $\eps$-semi-ladder, $\eps$-ladder, or $\eps$-comatching of a metric space $(V,\dist)$ the
\EMPH{$\eps$-semi-ladder index}, \EMPH{$\eps$-ladder index}, or \EMPH{$\eps$-comatching index}, respectively.
(The $\eps$-semi-ladder index is called an $\eps$-scatter dimension in~\cite{ABBCGKMSS23}.)

Bourneuf and Pilipczuk~\cite{BP25} were first to observe the connection
between the notions of $\eps$-scatter of~\cite{ABBCGKMSS23} and the aforementioned concepts from the theory of structurally sparse graphs. They adopted the crucial tools from the latter to the metric settings and proved that, for any proper minor-closed graph class $\mathcal{G}$, the class of metrics induced by edge-weighted graphs from $\mathcal{G}$ has $\eps$-semi-ladder index bounded by\footnote{The subscript $\mathcal{G}$ in $\mathrm{poly}_{\mathcal{G}}$ indicates the dependency on the size of the excluded minors of the graph family $\mathcal{G}$.} $2^{\mathrm{poly}_{\mathcal{G}}(1/\eps)}$. Note that this result applies to the classes of planar graphs, graphs embeddable on a fixed surface, or graphs of bounded treewidth, since they are all minor-closed. It also immediately implies a $2^{\mathrm{poly}(1/\eps)}$ bound on the
size of the furthest neighbor coreset in planar metrics constructed by the aforementioned simple greedy algorithm. 
Unfortunately, the exponential size is the best possible with this approach: there exists an $\eps$-ladder in planar metrics of size $2^{\Omega(1/\eps)}$ as constructed by Sokołowski~\cite{Sokolowski21}. Hence, to bypass the exponential barrier, we need to depart from the greedy approach, which is closely tied to $\eps$-semi-ladder index.

To motivate our new approach, consider now the following special case: the metric $(V,\dist)$ is given by a simple unweighted graph and for every $v \in V$ it holds
that $\max_{p \in P} \dist(v,p) = R$ for some integer $R$. Furthermore, $\eps < \frac{1}{R}$, that is, $(1-\eps)R > R-1$. In other words, a furthest neighbor $\eps$-coreset $Q$
need to include, for every $v \in v$, at least one $p \in P$ with $\dist(v,p) = R$.

Let $Q$ be an inclusion-wise minimal $\eps$-coreset in this special case. By the minimality,
for every $p \in Q$ there exists $v_p \in V$ with $p$ being the only point in $Q$ at distance
exactly $R$. We observe that $(v_p,p)_{p \in Q}$ is an $\eps$-comatching in $V$!
So, planar metrics may have exponential $\eps$-ladder index~\cite{Sokolowski21},
but maybe a better coreset for furthest neighbor can be obtained by looking at comatchings? If one wishes to follow this approach, there are two major obstacles: 
\begin{itemize}
    \item \textbf{Obstacle (i).~} showing that $\eps$-comatching index and $\eps$-coreset are polynomially related. 
    \item  \textbf{Obstacle (ii).~} bounding the size of $\eps$-comatching index as a polynomial function of $1/\eps$.
\end{itemize}
Obstacle (ii) is clearly the most difficult: $\eps$-comatching and $\eps$-semi-ladder are very similar (see \Cref{def:metric-ladder}), and there is currently no evidence to support that an exponential separation between them exists.  Nevertheless, we take exactly this approach and resolve both obstacles along the way.

\subsection{Our main results}\label{subsec:mainresult}

Our first result is a formal connection between the $\eps$-comatching index and a coreset for furthest neighbor, overcoming obstacle (i). Unfortunately, the above minimality argument stops working if we consider weights.  Instead, we devise a new reduction based on linear programming. Specifically, we round a set-cover-like linear program for metrics where the ball set system has a bounded VC-dimension. And luckily, the ball set system of planar metrics has VC-dimension $4$. Thus, our reduction from $\eps$-coreset to  $\eps$-comatching exploits planarity in a crucial way via VC-dimension. This contrasts with the reduction from $\eps$-coreset to $\eps$-semi-ladder in the prior work~\cite{BP25}, which holds for any metric space. 

Let $B(v,r) = \{u\in V| \dist(u,v)\leq r\}$ be a ball of radius $r\geq 0$ centered at $v$. Let $\mathcal{B}_V =  (V, \{B(v,r)\}_{v\in V, r\in \real^+})$ be the set system with ground set $V$, and all the balls as the family of sets. 

\begin{theorem}[Coreset to Comatching]\label{thm:comatching2coreset}
  Let $(V,\dist)$ be a metric space,  $P \subseteq V$ be a set of points, and a parameter $\varepsilon \in (0,1)$. Let $L(\eps)$ be the $\eps$-comatching index of $(V,\dist)$.
  Let $d$ be the VC-dimension of $\mathcal{B}_V$. Then there exists a furthest neighbor $\varepsilon$-coreset for $P$ of size $O(d\eps^{-2}\cdot L(\eps^2/4)\log L(\eps^2/4))$ that is computable in polynomial time. 
\end{theorem}

If $(V,\dist)$ is a planar metrics, it is known that $\mathcal{B}_V$ has VC-dimension at most $4$; more generally, for $K_h$-minor-free metrics, the VC-dimension is at most $h-1$~\cite{CEV2007,LW24}.   Thus, Theorem~\ref{thm:comatching2coreset} reduces the algorithmic question of finding a coreset to a purely combinatorial question
of the $\eps$-comatching index of planar metrics. 

Our next result shows that the $\eps$-comatching index of planar metrics is bounded polynomially in $1/\eps$, fully resolving obstacle (ii). We found this result rather surprising, given the exponential lower bound for the $\eps$-semi-ladder index by Soko{\l}owski~\cite{Sokolowski21}. %

\begin{restatable}[Comatching in Planar Metrics]{theorem}{CoMatchingTheorem}\label{thm:comatching-bound}
There exists a polynomial $L_{\mathrm{planar}}$ such that
for every $\eps \in (0,1)$ and every planar metric $(V,\dist)$,
the $\eps$-comatching index of $(V,\dist)$ is at most $L_{\mathrm{planar}}(1/\eps)$.
\end{restatable}

Theorems~\ref{thm:comatching2coreset} and~\ref{thm:comatching-bound}
give immediately the main result of this work, namely Theorem~\ref{thm:1center}.
\begin{proof}[Proof of \Cref{thm:1center}] By \Cref{thm:comatching-bound}, we have $L(\eps) = \poly(1/\eps)$ and hence $L(\eps^2/4) = \poly(1/\eps)$. As VC dimension of $\mathcal{B}_V$  is at most $4$ in planar metrics, by \Cref{thm:comatching2coreset}, we can construct a furthest neighbor $\eps$-coreset of size $O(\eps^{-2}L(\eps^2/4)\log L(\eps^2/4)) = \poly(1/\eps)$. 
\end{proof}

We give an overview of the proof of Theorem~\ref{thm:comatching-bound} in Section~\ref{ss:over}.

\begin{table}[htbp]
    \centering
    \setlength{\tabcolsep}{1.5em} 
    \renewcommand{\arraystretch}{1.5} 
    
    \begin{tabular}{llll}
    \toprule
    \textbf{Problem} & \textbf{Graph class} & \textbf{Bounds} (UB/LB) & \textbf{Reference}\\
    \midrule
    
    \multirow{6}{*}{\makecell[l]{$\eps$-comatching / \\ Furthest neighbor \\ $\eps$-coreset}} 
    & Simple polygons & $O(1/\eps^2)$ (UB) & \cite{BT24} \\
    & Planar & $\poly(1/\eps)$ (UB) & Thm.~\ref{thm:1center}  \\
    & Bounded genus $g$ & $\poly(g, 1/\eps)$ (UB)& Thm.~\ref{thm:genus} \\
    & Treewidth $\leq 5$ & $2^{\Omega(1/\eps)}$ (LB) & Cor.~\ref{cor:minor-lb-FNcoreset} \\
    & $K_5$-minor free & $2^{\Omega(1/\eps)}$ (LB) & Cor.~\ref{cor:minor-lb-FNcoreset} \\
    & $H$-minor-free & $2^{\mathrm{poly}_H(1/\eps)}$ (UB) & \cite{BP25} \\
    
    \addlinespace[1.5ex]
    \midrule
    \addlinespace[1.5ex]
    
    \multirow{5}{*}{\makecell[l]{$\eps$-coreset for \\ $k$-Center}} 
    & Trees & $\Omega(2^{1/\eps})$ when $k \sim 1/\eps$ (LB) & Thm.~\ref{thm:kcenter-lower-bound}\\\cmidrule{2-4}
    & \multirow{2}{*}{Planar} & $\left({1/\eps}\right)^{2^{O(k)}}$ for fixed $k$ (UB) &Thm.~\ref{thm:kcenter} \\
    & & $\Omega(4^{1/\eps})$ when $k \sim 1/\eps$ (LB) & Thm.~\ref{thm:kcenter-lower-bound} \\ \cmidrule{2-4}
    & Bounded genus $g$ & $\poly(g) \cdot (1/\eps)^{2^{O(k)}}$ (UB) & Thm.~\ref{thm:genus2} \\
    \cmidrule{2-4}
    & $H$-minor-free & $k^{\mathrm{poly}_H(1/\eps)}$ (UB) & \cite{BP25} \\
    \bottomrule
    \end{tabular}
    \caption{Summary of bounds of $\eps$-coreset for furthest neighbor and $k$-center in various graph classes. (UB) marks upper bounds and (LB) marks lower bounds.}
    \label{tab:coreset_bounds}
\end{table}

\subsection{Other results}

We extend our main results in \Cref{subsec:mainresult} in two different directions: %
\begin{enumerate}
    \item We investigate other minor-closed graph families. We show that bounded-genus metrics have polynomial $\eps$-comatching index. We also establish \EMPH{exponential lower bounds} on $\eps$-comatching index for other families, including bounded-treewidth metrics (and therefore, general minor-free metrics).
    \item We extend the coreset for furthest neighbors to $k$-center, establishing an $\eps$-coreset of size $\poly(1/\eps)$ \EMPH{for a fixed $k$} in planar and bounded-genus metrics, and provide an exponential lower bound of $2^{\Omega(1/\eps)}$  when $k = 1/\eps$. To get the upper bound, we introduce a variant of $\eps$-comatching called $(k,\eps)$-comatching. The bulk of the technical details is to bound the size of   $(k,\eps)$-comatching. For $\eps$-comatching, there is a clear symmetry between $p$-points and $q$-points. However,  $(k,\eps)$-comatching does not have the same symmetry, and hence bounding its size is much more demanding.
    
\end{enumerate}

Our results are summarized in \Cref{tab:coreset_bounds}. We start with bounded genus graphs.

\subsubsection{Graphs of bounded genus}

We observe that our results for planar metrics generalize to metrics induced by graphs
embeddable on any fixed surface. 

\begin{theorem}\label{thm:genus}
For any integer $g \geq 0$, $\eps \in (0,1)$, and a metric space $(V,\dist)$
induced by a graph embedded in a surface of Euler genus at most $g$,
the $\eps$-comatching index of $(V,\dist)$ is bounded by
$\Oh(g^4 \eps^{-8}) \cdot L_{\mathrm{planar}}(1/\eps)$,
where $L_{\mathrm{planar}}$ comes from Theorem~\ref{thm:comatching-bound}.
\end{theorem}

The proof of Theorem~\ref{thm:genus} follows the standard way of cutting a surface
along $\Oh(g)$ shortest paths and individual edges.

Theorem~\ref{thm:genus}, together with Theorem~\ref{thm:comatching2coreset}, generalize Theorem~\ref{thm:1center} to metrics induced by graphs embeddable on a fixed surface; the size of the coreset now depends polynomially on $1/\eps$ and the Euler genus of the surface.

\paragraph{Beyond surface-embedded graphs: lower bounds.}
The discussed results so far uncover a surprising separation between $\eps$-comatching
and $\eps$-ladder indices in planar and, more generally, bounded genus metrics. 
It is natural to ask how wide this phenomenon is, in particular, whether
it applies to any proper minor-closed graph class. We show that actually our results
almost completely exhaust the plateau of positive results. 

As all lower bounds we present in this paper are for unweighted graphs, let us make
a handy definition: for an (unweighted) graph $G$
and an integer $d \geq 1$, a \EMPH{$d$-comatching} is a family $\mathcal{M}$
of pairs of vertices of $G$ such that
\begin{enumerate}
    \item for every $(p,q) \in \mathcal{M}$, we have $\dist(p,q) > d$;
    \item for every distinct $(p,q),(p',q') \in \mathcal{M}$, we have $\dist(p,q') \leq d$.
\end{enumerate}
Clearly, a $d$-comatching is an $\eps$-comatching for any $\eps < \frac{1}{d+1}$. First, we give the following exponential lower bound on the size of $d$-comatching.  Our construction is a modification of the construction of Sokołowski~\cite{Sokolowski21}, with a detailed analysis of its properties.

\begin{restatable}[Lower bound construction for furthest neighbor and comatchings]{theorem}{SokoLowerBound}
\label{thm:comatching-lb}
For every integer $k \geq 3$ there exists a graph $G_k$ with the following properties:
\begin{enumerate}
    \item \label{it:com-lb1} $G_k$ contains a $(2k-1)$-comatching of size $2^k$.
    \item \label{it:com-lb2} $G_k$ admits a tree decomposition, where every bag is of size at most $6$ and every adhesion is of size at most $2$ (hence, in particular, it is of constant treewidth);
    \item \label{it:com-lb3} $G_k$ is $K_{3,4}$-minor-free;
    \item \label{it:com-lb4} $G_k$ is $K_5$-minor-free;
    \item \label{it:com-lb5} the Euler genus of $G_k$ is between $2^{k-3}$ and $2^k$.
\end{enumerate}
\end{restatable}

Observe that if $\mathcal{M}$ is a $d$-comatching in $G$ and $P = \{p~|~(p,q) \in \mathcal{M}\}$, then for every $\eps < \frac{1}{d+1}$, any $\eps$-coreset for furthest neighbor needs to store the entire $P$, as $p$ is the only valid answer for the query $q$ for any $(p,q) \in \mathcal{M}$. Thus, we obtain the following direct corollary of \Cref{thm:comatching-lb}:

\begin{corollary}\label{cor:minor-lb-FNcoreset}
There exists a constant $c > 0$ such that for every $\eps \in (0,1)$
there exists a  $K_{5}$-minor-free graph of treewidth at most $5$ where any $\eps$-coreset for furthest neighbor must have size  at least $2^{c/\eps}$.
\end{corollary}

Note that \Cref{cor:minor-lb-FNcoreset} gives little space left outside 
graph classes of bounded genus.

\subsubsection{Coreset for $k$-center}

For a set $X\subseteq (V,\dist)$, we define $\dist(v,X) = \min_{x\in X}\dist(v,x)$.  In the $k$-center problem, the coreset preserves the maximum distance from any set of $k$ centers to a point in $P$, as formally defined below.

\begin{definition}[$k$-Center $\varepsilon$-Coreset] \label{def:coreset-kcenter}
    Given a point set $P$ in a metric $(V,\dist)$, a $k$-center $\varepsilon$-coreset for $P$ is a subset $Q\subseteq P$ such that for every set $X\subseteq V$ of at most $k$ potential centers,  there exists a~point $q\in Q$ such that
    \[d(q,X) \geq (1-\varepsilon) \max \{\dist(p,X)~|~p \in P\} .\]
\end{definition}
Observe that for $k=1$, the $k$-center $\eps$-coreset is exactly the $\eps$-coreset for furthest neighbor. 

Bourneuf and Pilipczuk~\cite{BP25} showed that $k$-center in planar and minor-free metrics admits an $\eps$-coreset of size $O(k)^{\poly(1/\eps)}$; the dependency on $1/\eps$ is exponential even when $k = 1$. For a fixed $k$,  we show that one can get a polynomial dependency on $1/\eps$. To complement our upper bound, we show an exponential lower bound when $k \sim 1/\eps$; our lower bound holds even for \EMPH{tree metrics}.

\begin{theorem}\label{thm:kcenter} Given a point set $P$ in a planar metric $(V,\dist)$, a number of centers $k\in \mathbb{N}$, and a parameter $\eps\in(0,1)$, there exists a $k$-center $\eps$-coreset of size $1/\eps^{2^{O(k)}}$  for $P$.
\end{theorem}
\begin{theorem}\label{thm:kcenter-lower-bound}
For every $\eps \in (0,1/2)$, for $k \coloneqq \lceil \frac{1}{\eps} \rceil - 2$,
there exists
\begin{enumerate}[itemsep=0px]
    \item a tree metric such that any $k$-center $\eps$-coreset must have size at least $\frac{2^{1/\eps}}{4}$;
    \item a planar metric such that any $k$-center $\eps$-coreset must have size at least $\frac{4^{1/\eps}}{4}$.
\end{enumerate}
\end{theorem}

We prove Theorems~\ref{thm:kcenter} and~\ref{thm:kcenter-lower-bound} by introducing variants of $\eps$-comatching and $d$-matching, called \EMPH{$(k,\eps)$-comatching} and \EMPH{$(k,d)$-comatching}, respectively. In a $(k,\eps)$-comatching, one side, instead of being points, contains \EMPH{sets of points} of size $k$. More precisely, each (ordered) pair in a $(k,\eps)$-comatching $\mathcal{M}$ is of the form $(p,X)$ where $X$ is a set of $k$ points. A $(k,d)$-comatching is similar to a $d$-comatching earlier for a positive integer $d$ where the distance threshold is $d$. Note that a $(k,d)$-comatching is a $(k,\eps$)-comatching for any $\eps < \frac{1}{d+1}$. 

\paragraph{The lower bound: $(k,d)$-comatching.} Our lower bound is established via $(k,d)$-comatching in unweighted planar graphs. Specifically, we show two constructions. The first construction is a simple construction that gives a weaker bound, but constructs a tree; the second one gives a planar graph with a stronger lower bound. Both constructions rule out fully polynomial dependency on $k$ and $d$ .

\begin{restatable}{theorem}{KKComatching}\label{thm:kcenter-lb}
For every integer $k \geq 1$ there exists a tree $G_k$
with a  $(k,k)$-comatching of size $2^k$.
\end{restatable}

\begin{restatable}{theorem}{SoullessTheorem}\label{thm:kcenter-lb2}
For every integers $k,d \geq 1$ there exists a planar graph $G_{k,d}$
with a $(k,d)$-comatching of size at least
    \[ \left(2\left\lfloor\frac{d}{k}\right\rfloor+2\right)^k \quad\mathrm{if\ }k \leq d;\qquad\qquad
        \left(3\left\lfloor\frac{k}{d}\right\rfloor+1\right)^d \quad\mathrm{if\ }d \leq k. \]
    (Note that both bounds give $4^k$ for $k=d$.)
\end{restatable}

Note that if $\mathcal{M}$ is a $(k,d)$-comatching, then for
$P = \{p~|~(p,X) \in \mathcal{M}\}$,  and any $\eps < \frac{1}{d+1}$, any $\eps$-coreset for $k$-center needs
to keep the entire $P$, as $p$ is the only valid answer for the query $X$ for $(p,X) \in \mathcal{M}$. Thus, by setting $k = d = \lceil \frac{1}{\eps} \rceil - 2$, we have $\eps < \frac{1}{k+1}$, so our lower bound for $(k,d)$-comaching in \Cref{thm:kcenter-lb} implies the exponential lower bound $\frac{2^{1/\eps}}{4}$ for tree metrics claimed in \Cref{thm:kcenter-lower-bound}. For planar metrics, \Cref{thm:kcenter-lb2} implies a slightly better lower bound, of $\frac{4^{1/\eps}}{4}$, when $k = d = \lceil \frac{1}{\eps} \rceil -2$ (so again $\eps < \frac{1}{k+1}$).

\paragraph{The upper bound: $(k,\eps)$-comatching.} The upper bound proof of Theorem~\ref{thm:kcenter} follows the same general structure as the proof of Theorem~\ref{thm:1center}.  By a similar reduction as in the case $k=1$, based on linear programming and VC-dimension, we could reduce the $k$-center coreset construction to bounding the maximum size of a $(k,\eps^2/4)$-comatching. 

\begin{restatable}{theorem}{KCenterComatching}\label{thm:kcenter-comatching}
Let $(V,\dist)$ be a metric space,  $P \subseteq V$ be a set of points, and a parameter $\varepsilon \in (0,1)$. Let $L(k,\eps)$ be the length of the maximum $(k,\eps)$-comatching in $(V,\dist)$. Let $d$ be the VC-dimension of $\mathcal{B}_V$. Then there exists a $k$-center $\varepsilon$-coreset for $P$ of size $O(dk\log k\eps^{-2}\cdot L(k,\eps^2/4)\log L(k,\eps^2/4))$. %
\end{restatable}
 Working with $(k,\eps)$-comatchings is more challenging than with $\eps$-comatchings
as in Theorem~\ref{thm:comatching-bound}.
The problem here is that the sets on one side of the $(k,\eps)$-comatching introduce a lot of difficulty due to lack of symmetry, if we were to follow the same approach for bounding the size of $\eps$-comatching. Specifically, the $\eps$-comatching proof technique crucially exploits the symmetry between pairs: for each $(p, q)\in \mathcal{M}$, the point $q$ is ``close" to every other $p_i$, which is essential for establishing a packing-type contradiction. However, in a $(k,\eps)$-comatching $\mathcal{M}'$, this symmetry disappears, i.e., for some $(p, X)\in \mathcal{M}'$, different points in $X$ can serve as witnesses of closeness to different $p_i$s, and may be far from all other points. This lack of symmetry makes it difficult to restrict the point-set pairs within a~topological boundary, and the tools established for comatching no longer apply here.

We re-establish the symmetry using Ramsey-type arguments. 
Luckily, we are working with set systems of bounded VC-dimension
and we can use the recent polynomial Ramsey bounds for graphs of bounded VC-dimension by Nguyen, Scott, and Seymour~\cite{NSS24}. We establish that a huge $(k,\eps)$-comatching
contains either a large (regular) $(\eps/2)$-comatching or a large \EMPH{$(\eps/2)$-double-ladder}
defined as follows (see \Cref{fig:double-ladder} for an illustration).

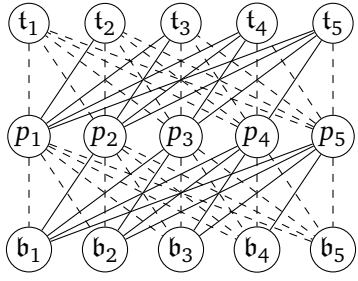
\begin{figure}[ht]
        \centering
        \begin{tikzpicture}[scale=1, every node/.style={circle, draw, fill=white, inner sep=1pt}]
            \foreach \i in {1,...,5}{
                \node (t\i) at (\i*1,3) {$\mathfrak{t}_{\i}$};
            }        
            \foreach \i in {1,...,5}{
                \node (p\i) at (\i*1,1.5) {$p_{\i}$};
            }       
            \foreach \i in {1,...,5}{
                \node (b\i) at (\i*1,0) {$\mathfrak{b}_{\i}$};
            }        
            \foreach \j in {1,...,5}{
                \foreach \i in {1,...,5}{
                    \ifnum\i>\j\relax
                    \else
                        \draw[dashed] (p\j) -- (t\i);
                    \fi
                }
            }       
            \foreach \i in {1,...,5}{
                \foreach \j in {1,...,5}{
                    \ifnum\i>\j\relax
                    \else
                        \draw[dashed] (p\i) -- (b\j);
                    \fi
                }
            }        
            \foreach \i in {1,...,5}{
                \foreach \j in {1,...,5}{
                    \ifnum\i<\j
                        \draw[] (p\i) -- (t\j);
                    \fi
                }
            }      
            \foreach \j in {1,...,5}{
                \foreach \i in {1,...,5}{
                    \ifnum\i<\j
                        \draw[] (p\j) -- (b\i);
                    \fi
                }
            }
        \end{tikzpicture}
        
    \caption{$\eps$-double ladder. Again, the dotted line indicates distance $\geq R$, and the solid line indicates distance $\leq(1-\eps)R$.}
    \label{fig:double-ladder}
\end{figure}

\begin{restatable}[$\varepsilon$-double ladder]{definition}{DoubleLadderDef}\label{def:double-ladder}
Let $\DL = \{\dltri{1}, \dltri{2},\ldots, \dltri{\ell}\}\subseteq V^3$  be a sequence of $\ell$ \emph{ordered triplets} of points in a metric $(V,\dist)$, and $\eps\in (0,1)$ be a parameter.  We say that $\DL$ is an \EMPH{$\eps$-double ladder} if there exists $R > 0$ such that:
    \begin{enumerate}
        \item for every $1 \leq i \leq j \leq \ell$,
        \[ \dist(p_j, \dlpt{t}_i) > R \quad \mathrm{and} \quad \dist(p_i, \dlpt{b}_j) > R;\]
        \item for every $1 \leq i < j \leq \ell$,
        \[ \dist(p_i, \dlpt{t}_j) \leq (1-\varepsilon)R \quad \mathrm{and} \quad \dist(p_j, \dlpt{b}_i) \leq (1-\varepsilon)R.\]
    \end{enumerate}
The sets of points $\{\dlpt{t}_1,\ldots, \dlpt{t}_{\ell}\}$ and $\{\dlpt{b}_1,\ldots, \dlpt{b}_{\ell}\}$ are  called the \EMPH{top} and  \EMPH{bottom} of the double ladder\footnote{We will be using \texttt{mathfrak} font for the top and bottom points of a double ladder}, respectively.
\end{restatable}

Observe that an $\eps$-double ladder $\mathcal{D} = (p_i,\dlpt{t}_i,\dlpt{b}_i)_{i=1}^\ell$
is an $(k=2,\eps)$-comatching $(p_i,X_i)_{i=1}^\ell$
with $X_i = \{\dlpt{t}_i, \dlpt{b}_i\}$.
The vertex $\dlpt{t}_i$ serves as the ``close to $p_j$'' 
vertex of $X_i$ for $j < i$ while $\dlpt{b}_i$ serves
this role for $j > i$. 

\begin{restatable}{theorem}{KComatchingToDoubleLadder}\label{thm:comatching-to-ladder} Let $(V,\dist)$ be a metric space such that the ball set system $\mathcal{B}_V$ has VC-dimension at most $d$. Let $\eps \in (0,1)$ be a parameter. There exists a universal constant $c$ and a polynomial $f_{k,d}(x) = x^{(c\cdot d)^{4k}}$ such that if $(V,\dist)$ contains a $(k,\varepsilon)$-comatching  of size at least $f_{k,d}(\ell)$, then it contains  either an $(\varepsilon/2)$-comatching or an $(\varepsilon/2)$-double ladder  of size $\ell$. 
\end{restatable}

Using a similar general approach as for Theorem~\ref{thm:comatching-bound}
(but with significantly different arguments in the last phase of the proof) we establish 
a polynomial bound on a maximum size of double ladder in a planar graph

\begin{theorem}[Double Ladders in Planar Metrics]\label{thm:doubleladder-bound}
There exists a polynomial $L'_{\mathrm{planar}}$ such that 
for any $\eps \in (0,1)$, any $\eps$-double ladder in planar metrics has size at most $L'_{\mathrm{planar}}(1/\eps)$. 
\end{theorem}

\noindent Now \Cref{thm:kcenter} is a simple consequence of \Cref{thm:comatching-bound,thm:kcenter-comatching,thm:comatching-to-ladder,thm:doubleladder-bound}.

\begin{proof}[Proof of \Cref{thm:kcenter}] By \Cref{thm:comatching-bound,thm:comatching-to-ladder,thm:doubleladder-bound}, the maximum size of a $(k,\eps)$-comatching is at most $f_{k, 4}(\poly(1/\eps)) = (1/\eps)^{2^{O(k)}}$. Then by \Cref{thm:kcenter-comatching}, we can construct an $\eps$-coreset for $k$-center of size:
\begin{equation*}
    O(\eps^{-2} f_{k,4}(\poly(1/\eps)) \cdot \log  f_{k,4}(\poly(1/\eps))   =  (1/\eps)^{2^{O(k)}}~,
\end{equation*}
as claimed.
\end{proof}

Furthermore, we note that the technique we used to generalize from planar metrics
to metrics of bounded genus works here as well.
\begin{theorem}\label{thm:genus2}
For any integer $g \geq 0$, $\eps \in (0,1)$, and a metric space $(V,\dist)$
induced by a graph embedded in a surface of Euler genus at most $g$,
the maximum size of an $\eps$-double ladder in $(V,\dist)$ is bounded by
$\Oh(g^4 \varepsilon^{-8}) \cdot L'_{\mathrm{planar}}(1/\eps)$,
where $L'_{\mathrm{planar}}$ comes from Theorem~\ref{thm:doubleladder-bound}.
\end{theorem}
Similarly as for furthest neighbor, Theorem~\ref{thm:genus2} generalizes
the $k$-center coreset of Theorem~\ref{thm:kcenter} to graphs of bounded genus.

\subsection{Organization of the paper}

In~\Cref{ss:over} we give a short overview of the proof of Theorem~\ref{thm:comatching-bound}.
After preliminaries in \Cref{sec:prelims},
we focus on \Cref{thm:1center}: 
\Cref{sec:comatchingtocoreset} proves \Cref{thm:comatching2coreset}
and \Cref{sec:comatching-bound} proves \Cref{thm:comatching-bound}.
Then, we move to $k$-center:
\Cref{sec:kcenter} proves \Cref{thm:kcenter-comatching} and \Cref{thm:comatching-to-ladder}
while \Cref{sec:doubleladder} proves \Cref{thm:doubleladder-bound}.
In \Cref{subsec:other}, we briefly discuss how to cast polygons with holes
onto planar metrics.
Extension to bounded genus graphs is discussed in \Cref{sec:genus}.
Finally, Section~\ref{sec:lb} proves our lower bounds, \Cref{thm:comatching-lb},
 \Cref{thm:kcenter-lb}, and \Cref{thm:kcenter-lb2}.

\section{Overview of the Proof of Theorem~\ref{thm:comatching-bound}}\label{ss:over}

In this section, we give the high-level ideas of our key technical result, which bounds the size of $\eps$-matching in planar metrics by a polynomial function of $1/\eps$. Our proof is rather involved, and some steps rely heavily on the planar embedding of the graph to narrow down the possible configurations of shortest paths. Since the phenomenon of polynomial bounds for $\eps$-comatching index
is restricted to planar and bounded genus metrics and is not present in, say, bounded
treewidth graphs, the very topological nature of our proof is most likely inevitable. 

Let $L = (p_i,q_i)_{i=1}^\ell$ be an $\eps$-comatching in a planar graph $G$
and let $R$ be the witnessing distance, i.e., $\dist(p_i,q_i) > R$ 
but $\dist(p_i,q_j) \leq (1-\eps)R$ for $i \neq j$. 
For every $i \neq j$, fix a shortest path $\Mpath{p_i}{q_j}$
between $p_i$ and $q_j$ in a consistent way, that is, any two shortest chosen shortest
paths are either disjoint or intersect in a subpath.
Without loss of generality, assume that every edge of $G$ lies on some path $\Mpath{p_i}{q_j}$.
It is easy to see that, as long as $\ell \geq 3$, the diameter of $G$ is bounded
by $3R$. 

\paragraph{Approximate distance profiles.}
An important tool in the analysis are so-called \emph{approximate distance profiles}.
Denote $\delta = \frac{\eps}{100} R$; we think of $\delta$ as a unit of distance, so that
an error of a few $\delta$s does not matter. 
For vertices $u,v \in V(G)$, let
\[ \Udist(u,v) = \left\lceil \dist(u,v)/\delta \right\rceil. \]
That is, $\Udist(u,v)$ is the distance from $u$ to $v$ measured in the number of $\delta$s,
rounded up. Note that, as the diameter of $G$ is bounded by $3R$,
we have $\Udist(u,v) = \Oh(1/\eps)$. 

Let $Q$ be a path of length $\Oh(R)$ in $G$. Let $P$ be a set of portals on $Q$, placed
equidistantly at distance $\delta$ on $Q$; note that $|P| = \Oh(1/\eps)$. 
The \EMPH{approximate distance profile} of a vertex $v \in V(G)$ is
the function mapping $P \ni p \mapsto \Udist(v,p) \in \mathbb{N}$.
There are two important observations about distance profiles. First, due to the placement of $P$ along $Q$, an approximate distance profile of $v$ allows us to deduce, up to an error of $2\delta$,
the distance from $v$ to any vertex on $Q$. Second, it follows from the VC-dimension bound
on the ball system of planar graphs that the number of different distance profiles is
bounded polynomially in $1/\eps$ (note that the naive bound is exponential in $|P| = \Oh(1/\eps)$). 

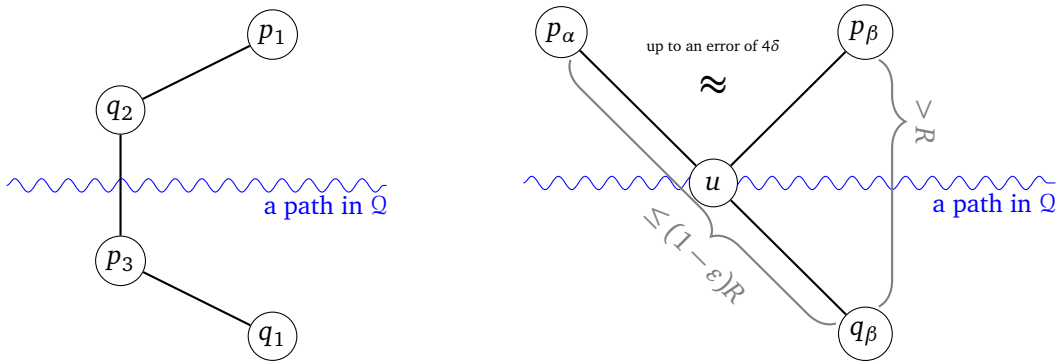
\begin{figure}[!htb]
\centering
\begin{subfigure}{0.49\textwidth}
        \centering
        \begin{tikzpicture}[scale=1]
        \tikzstyle{vertex} = [circle, fill=black, minimum size=4pt, inner sep=0pt]
      \tikzstyle{fvertex} = [circle, draw=black, fill=white, minimum size=18pt, inner sep=2pt]

    \path[draw=blue,snake it] (-2.5, 0) -- (2.5,0);
    \draw[blue] (1.7, -0.3) node {\small a path in $\mathcal{Q}$};

    \node[fvertex] (p1) at (1, 2) {$p_1$};
    \node[fvertex] (q2) at (-1, 1) {$q_2$};
    \node[fvertex] (p3) at (-1, -1) {$p_3$};
    \node[fvertex] (q1) at (1, -2) {$q_1$};
    \draw[thick] (p1) -- (q2) -- (p3) -- (q1);
        \end{tikzpicture}
        \caption{Proof that $\mathcal{Q}$ needs to intersect
        some $\Mpath{p_\alpha}{q_\beta}$.}\label{fig:over:sep1}
    \end{subfigure}%
    \begin{subfigure}{0.49\textwidth}
        \centering
        \begin{tikzpicture}[scale=1]
    \tikzstyle{vertex} = [circle, fill=black, minimum size=4pt, inner sep=0pt]
      \tikzstyle{fvertex} = [circle, draw=black, fill=white, minimum size=18pt, inner sep=2pt]

    \path[draw=blue,snake it] (-2.5, 0) -- (4.5,0);
    \draw[blue] (3.7, -0.3) node {\small a path in $\mathcal{Q}$};

    \node[fvertex] (pa) at (-2, 2) {$p_\alpha$};
    \node[fvertex] (qb) at (2, -2) {$q_\beta$};
    \node[fvertex] (pb) at (2, 2) {$p_\beta$};
    \node[fvertex] (u) at (0,0) {$u$};
    \draw[thick] (pa) -- (u);
    \draw[thick] (pb) -- (u);
    \draw[thick] (qb) -- (u);

    \draw [thick,decorate,gray,decoration={brace,raise=5pt,amplitude=10pt,aspect=0.25}] (pb) -- (qb) 
    node [pos=0.25, anchor=south, sloped, yshift=15pt] {$> R$};

    \draw [thick,decorate,gray,decoration={brace,mirror,raise=5pt,amplitude=10pt,aspect=0.6}] (pa) -- (qb) 
    node [pos=0.6, anchor=north, sloped, yshift=-15pt] {$\leq (1-\eps)R$};

    \draw (0, 1.3) node {\Large $\approx$};
    \draw (0,1.8) node {\tiny up to an error of $4\delta$};
\end{tikzpicture}
\caption{The final contradiction.}\label{fig:over:sep2}
\end{subfigure}
\caption{Illustration for the proof of the separation lemma.}\label{fig:over:sep}
\end{figure}

\paragraph{Separation lemma.}
Let $\mathcal{Q}$ be a family of $m$ paths in $G$, each of length $\Oh(R)$, 
and let $I \subseteq [\ell]$ be the set of those indices $i \in [\ell]$
such that $p_i$ and $q_i$ are separated by $V(\mathcal{Q})$,
i.e., do not lie in the same connected component
of $G-V(\mathcal{Q})$, where $V(\mathcal{Q})$ denotes the vertices lying on the paths $\mathcal{Q}$. 
Assume there exists three indices $i_1,i_2,i_3 \in I$ such that $p_{i_1}$, $p_{i_2}$,
and $p_{i_3}$ have the same approximate distance profile to each path in $\mathcal{Q}$. 
As $V(\mathcal{Q})$ separates $p_{i_1}$ from $q_{i_1}$,
$V(\mathcal{Q})$ needs to intersect somewhere the walk
being the concatenation of $\Mpath{p_{i_1}}{q_{i_2}}$, $\Mpath{q_{i_2}}{p_{i_3}}$, and $\Mpath{p_{i_3}}{q_{i_1}}$
(see Figure~\ref{fig:over:sep1}).
That is, for some $\alpha,\beta \in [3]$, $\alpha \neq \beta$,
the path $\Mpath{p_{i_\alpha}}{q_{i_\beta}}$ necessarily intersects $V(\mathcal{Q})$;
let $u$ be a vertex of the intersection. We obtain a contradiction as follows (see also Figure~\ref{fig:over:sep2}):
\begin{align*} R &< \dist(p_{i_\beta}, q_{i_\beta}) \leq \dist(p_{i_\beta}, u) + \dist(u, q_{i_\beta}) \\&\leq 
\dist(p_{i_\alpha}, u) + 4\delta + \dist(u, q_{i_\beta}) \\
&= \dist(p_{i_\alpha}, q_{i_\beta}) + 4\delta \leq (1-\eps)R + 4\delta < R  \quad \text{(as $\delta = \frac{\eps}{100} R$)}\end{align*}
Here, the third inequality comes from the assumption that $p_{i_\alpha}$ and $p_{i_\beta}$
have the same approximate distance profile to the paths of $\mathcal{Q}$. 

Hence, the size of $|I|$ is bounded by twice the number of possible approximate distance profiles to the paths in $\mathcal{Q}$, which is bounded polynomially in $m/\eps$ by the same VC-dimension based argument as for the individual path.

\begin{figure}[!h]
    \centering
    \begin{subfigure}{0.45\textwidth}
        \centering        \includegraphics[width=\linewidth]{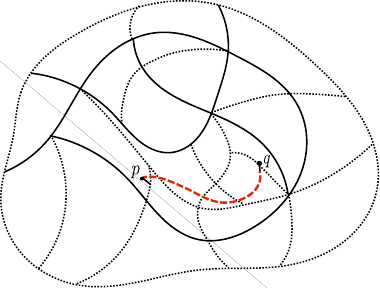}
        \caption{Local instance}
    \end{subfigure}
    \hfill
    \begin{subfigure}{0.45\textwidth}
        \centering        \includegraphics[width=\linewidth]{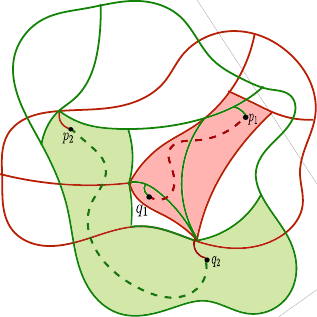}
        \caption{Linked non-crossing instance}
    \end{subfigure}
    \caption{The plane drawing in (a) depicts a local instance $G$, and the local property is illustrated by the pair $(p, q)$. Here the edges of $G$ which belong to paths $\Mpath{p_{j_1}}{q_{j_2}}$ for $p_{j_1}, q_{j_2} \neq p, q$ are drawn as solid lines and any other edges as dotted lines. We can draw a link (the dashed red curve) in the plane connecting $p$ and $q$ without intersecting any of the solid lines, so $(p, q)$ does indeed have the local property. The linked non-crossing instance is shown in (b) with respect to pairs $(p_1, q_1)$ and $(p_2, q_2)$. For each pair, the link represented by the dashed line does not intersect the solid lines of the same color to show the local property analogously to a). Furthermore the links cannot intersect each other (while maintaining the local property) since the regions in which links can lie are disjoint, as represented, which guarantees the non-crossing property.} 
    \label{fig:local_noncrossing_intro}
\end{figure}

\paragraph{Phase one: reduction to a local instance.}
The separation lemma motivates the following reduction. For $I \subseteq [\ell]$ and $i \in [\ell] \setminus I$, we say that $I$ \EMPH{separates} $i$
if $p_i$ and $q_i$ are separated by paths $\{ \Mpath{p_{j_1}}{q_{j_2}}~|~j_1,j_2 \in I, j_1 \neq j_2 \}$. The separation lemma implies the number of indices $I$ can separate
is bounded polynomially in $|I|/\eps$. 
A charging argument implies that if $\ell$ is huge, there is a large subset
$J \subseteq [\ell]$ such that for every $i \in J$, $J \setminus \{i\}$ does not separate
$i$.

We henceforth restrict to only pairs in $J$.
That is, we assume that for every $i \in [\ell]$, the set $[\ell] \setminus \{i\}$ does
not separate $i$. In other words, we can draw in a plane a link connecting 
$p_i$ with $q_i$ in a way that it does not intersect any edge of $G$
that belongs to a path $\Mpath{p_{j_1}}{q_{j_2}}$ for $j_1,j_2 \neq i$. 
Let us call such an instance a \EMPH{local} instance (see Figure~\ref{fig:local_noncrossing_intro}).
Our goal is now to prove that in a local instance, $\ell$ is bounded polynomially
in $1/\eps$.

\paragraph{Phase two: reduction to a linked noncrossing instance.}
In the second phase, we would like to comb our local instance so that the links do not cross
each other. To this end, we employ again the separation lemma in the following setting.

A well-known fact is that every connected planar graph admits a tree decomposition,
whose every bag is a union of three shortest path with a common endpoint. 
This is actually a special case of a more general statement: for every rooted spanning
tree $T$ of a planar graph $G$, there exists a tree decomposition of $G$
whose every bag is a union of three upwards paths in $T$~\cite{DBLP:journals/siamcomp/LiptonT80}. 
(The first statement is just the application of the second statement to any shortest-path
tree in $G$.) 
Instead of a shortest path tree, we construct a spanning tree $T$ of $G$ where each
upward path can be covered by $\Oh(\log \ell)$ paths $\Mpath{p_i}{q_j}$. 
In other words, if we obtain a tree decomposition of $G$ from $T$,
then  only very few --- $\Oh(\log \ell)$ --- links can intersect
a single bag of the decomposition.
By recursively separating the graph $G$ with bags of the obtained tree decomposition
in a balanced manner, we obtain a large number of indices $i$ whose links are separated
by the chosen bags --- and thus cannot intersect. 

Hence, a huge local instance needs to necessarily contain a large set $I \subseteq [\ell]$
whose links pairwise do not intersect. In what follows, we focus only on $I$:
that is, we assume that all links are pairwise disjoint, call such an instance
a \EMPH{linked noncrossing instance} (see Figure~\ref{fig:local_noncrossing_intro}), and proceed with bounding $\ell$ in such an instance.

\begin{figure}[!h]
\centering
\begin{tikzpicture}[scale=1]
    \tikzstyle{vertex} = [circle, fill=black, minimum size=4pt, inner sep=0pt]
      \tikzstyle{fvertex} = [circle, draw=black, fill=white, minimum size=18pt, inner sep=2pt]

    \node[fvertex] (pi) at (-2, 2) {$p_i$};
    \node[fvertex] (qi) at (-2, -2) {$q_i$};
    \node[fvertex] (qj) at (2, -2) {$q_j$};
    \node[fvertex] (pj) at (2, 2) {$p_j$};
    \node[fvertex] (u) at (0,0) {$u$};
    \draw[thick] (pi) -- (u);
    \draw[thick] (pj) -- (u);
    \draw[thick] (qj) -- (u);
    \draw[thick] (qi) -- (u);

    \draw [thick,decorate,gray,decoration={brace,mirror,raise=5pt,amplitude=10pt,aspect=0.5}] (pi) -- (qi) 
    node [pos=0.5, anchor=north, sloped, yshift=-15pt] {$> R$};

    \draw [thick,decorate,gray,decoration={brace,raise=5pt,amplitude=10pt,aspect=0.5}] (pj) -- (qj) 
    node [pos=0.5, anchor=south, sloped, yshift=15pt] {$> R$};

    \draw [thick,decorate,gray,decoration={brace,mirror,raise=5pt,amplitude=10pt,aspect=0.75}] (pi) -- (qj) 
    node [pos=0.75, anchor=north, sloped, yshift=-15pt] {$\leq (1-\eps)R$};

    \draw [thick,decorate,gray,decoration={brace,mirror,raise=5pt,amplitude=10pt,aspect=0.25}] (pj) -- (qi) 
    node [pos=0.25, anchor=south, sloped, yshift=15pt] {$\leq (1-\eps)R$};

\end{tikzpicture}
\caption{The contradiction in the proof of Lemma~\ref{lem:over:ex}.}\label{fig:over:ex}
\end{figure}
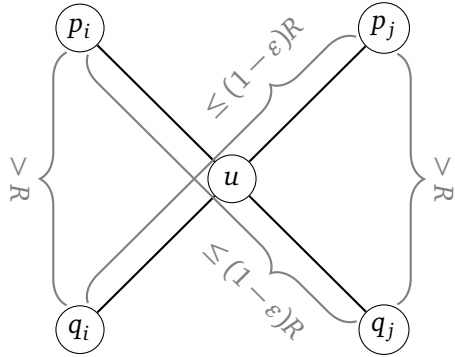

\paragraph{Phase three: topological study of a linked noncrossing instance.}
In a linked noncrossing instance, we finally have enough topological constraints
to employ a detailed case-by-case analysis how such a graph needs to look like in the plane. 
Beside the hard-earned constraint that links are pairwise disjoint, we make a number of
observations in the following flavor (see Figure~\ref{fig:over:ex}):
\begin{lemma}\label{lem:over:ex}
    For every $i \neq j$, the paths $\Mpath{p_i}{q_j}$ and
    $\Mpath{p_j}{q_i}$ are disjoint.
\end{lemma}
\begin{proof}
    By contradiction, assume that a vertex $u$ lies on both paths. Then
    we obtain a contradiction as follows:
    \begin{align*}
        2(1-\eps)R &\geq \dist(p_i,q_j) + \dist(p_j, q_i) \\
        &  = \left(\dist(p_i, u) + \dist(u,q_j)\right) + \left(\dist(p_j, u) + \dist(u, q_i)\right) \\
        &  = \left(\dist(p_i, u) + \dist(u,q_i)\right) + \left(\dist(p_j, u) + \dist(u, q_j)\right) \\
        & \geq \dist(p_i, q_i) + \dist(p_j, q_j) > 2R.
    \end{align*}
\end{proof} 
It turns out that the only way to have a somewhat large number of pairs in a comatching
in a linked noncrossing instance is to arrange them as in Figure~\ref{fig:over:final}.
\begin{figure}[htb]
\centering
\begin{tikzpicture}
    \tikzstyle{vertex} = [circle, fill=black, minimum size=4pt, inner sep=0pt]
    \tikzstyle{fvertex} = [circle, draw=black, fill=white, minimum size=18pt, inner sep=2pt]

    \node[vertex] (q) at (1, 5) {};
    \node[vertex] (p) at (14, 5) {};
    \foreach \n in {1,2,3,4,5,6,7} {
      \node[fvertex] (p\n) at (\n*2-1, 0) {$p_\n$};
      \node[fvertex] (q\n) at (\n*2, 0) {$q_\n$};
      \draw[thick,dashed] (p\n) -- (q\n);
      \draw[thick] (p\n) -- (q);
      \draw[thick] (q\n) -- (p);
    }
\end{tikzpicture}
\caption{The final picture in the proof of Theorem~\ref{thm:comatching-lb}: how
a large linked noncrossing instance needs to look like.}\label{fig:over:final}
\end{figure}
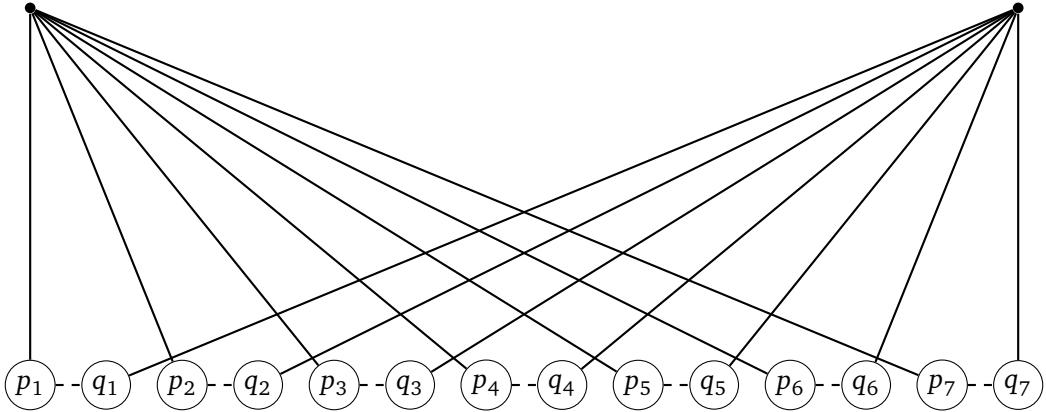

However, it can be shown that such an arrangement can fit only $\ell = \Oh(1/\eps)$. 
This concludes the overview of the proof of Theorem~\ref{thm:comatching-bound}.

\paragraph{A remark on double-ladders.}
Our bound on the size of a double-ladder, i.e., Theorem~\ref{thm:doubleladder-bound},
follows the same general outline. The first two phases are analogous, and differ only in
minor details.
The third phase, however, while also can be summarized as a detailed case-by-case analysis
using sufficient amount of topological constraints, is different --- the particular
arguments and case analysis are very different.

\paragraph{A remark on extending to bounded genus graphs.}
For Theorems~\ref{thm:genus} and~\ref{thm:genus2}, 
we use a standard tool of cutting open a graph of Euler genus $g$ into a planar graphs
by cutting $\Oh(g)$ shortest paths and single edges.
Observe that the argumentation
used in the proof of the separation lemma is also useful in this setting: 
if we partition the pairs of vertices of a comatching (or triples of vertices of a double ladder) into buckets according
to the approximate distance profile
to the cutset, then within a single bucket, all shortest paths $\Mpath{p_i}{q_j}$ cannot 
touch the cutset (or we get a similar contradiction as in the separation lemma), and thus essentially induce a comatching (or a double ladder) in a planar graph.

%% file: prelim.tex
\section{Preliminaries}\label{sec:prelims}

\subsection{VC dimension}
For a universe $U$ and a set system $\mathcal{A}$ of subsets of $U$
the \emph{VC-dimension} of $\mathcal{A}$ is the maximum size of a subset $X \subseteq U$
that is \emph{shattered} by $\mathcal{A}$; a set $X$ is shattered by $\mathcal{A}$
if for every $Y \subseteq X$ there exists $A \in \mathcal{A}$ with $A \cap X = Y$. 
One of the basic properties of set systems of bounded VC dimension is the following.
\begin{lemma}[Sauer-Shelah lemma]
  The number of sets in
  a set system $\mathcal{A}$ of VC-dimension at most $d$ over a universe $U$ 
  is at most
  \[ \sum_{i=0}^d \binom{|U|}{i} = \Oh(|U|^d). \]
\end{lemma}

Let $(V,\dist)$ be a metric space.
For $v \in V$ and $r > 0$, we define the ball $B(v,r) = \{u \in V~|~\dist_G(v,u) \leq r\}$. 
The set system $\{B(v,r)~|~v \in V(G), r > 0\}$ of subsets of $V$
is called \emph{the ball set system of $(V,\dist)$}. 

Graphs with positive edge-weights naturally induce a metric space on their set of vertices
via shortest path distances. 
In such metric spaces, if the underlying graph excludes
a small minor, the ball set system has small VC-dimension.

\begin{theorem}[\cite{CEV2007,DBLP:journals/dm/BousquetT15}]
  Let $G$ be an edge weighted graph with no $K_h$ as a minor. 
  Then, the ball set system $\{B(v,r)~|~v \in V(G), r > 0\}$ has VC-dimension
  at most $h-1$. 
\end{theorem}

For a metric space $(V,\dist)$ and $Z \subseteq V(G)$, for every $v \in V$ we define
\emph{the distance profile of $v$ to $Z$} as a function $\prof{Z}{v} : Z \to \mathbb{R}$
defined as $\prof{Z}{v}(z) = \dist(v,z)$. 
With a distance profile $\prof{Z}{v}$ we associate a set
\[ \mathcal{Z}_Z[v] = \{(z,x) \in Z \times \mathbb{R}~|~\prof{Z}{v}(z) \leq x\}. \]
An important property of the set system $\{ \mathcal{Z}_Z[v]~|~v \in V\}$ is that it has small
VC-dimension in graphs excluding a fixed minor. 
\begin{theorem}[Theorem 3~\cite{LW24}]\label{thm:prof-vc}
  Let $G$ be an edge weighted graph with no $K_h$ as a minor  and let $Z \subseteq V(G)$.   Then, the set system $\{\mathcal{Z}_Z[v]~|~v \in V(G)\}$ has VC-dimension
  at most $h-1$. 
\end{theorem}

Let $\mathcal{A}$ be a set. We define its \EMPH{$k$-fold set intersection} by $\mathcal{A}_{\sqcap}^{k} = \{A_1\cap \ldots\cap A_k~|~ A_i\in \mathcal{A} ~\forall i\in [k]\}$. We use the following bound on the VC dimension of $k$-fold set intersection: 

\begin{lemma}[Lemma 3.2.3.~\cite{DBLP:journals/jacm/BlumerEHW89}]\label{lm:vc-facts} Let $(U,\mathcal{A})$ be a set system of VC dimension at most $d$, then the VC dimension of $(U,\mathcal{A}_{\sqcap}^{k})$ is $O(d k\log k)$
\end{lemma}

\subsection{Tree decompositions of planar graphs}

The following theorem encapsulates the main construction leading to the 
celebrated Lipton-Tarjan separator theorem for planar graphs.
\begin{theorem}[see~\cite{DBLP:journals/siamcomp/LiptonT80}] \label{thm:tree-decomp}
 Let $G$ be a plane triangulation
 (i.e., a simple connected graph embedded in the plane where every face is a triangle)
 and let $H$ be a spanning tree of $G$ with a designated root $r$.
 Let $T$ be a subgraph of the dual of $G$ with $V(T)$ being the set of all faces of $G$
 and $E(T)$ being the duals of all edges of $E(G) \setminus E(H)$.
 Furthermore, for every $t \in V(T)$, let $\beta(t) \subseteq V(G)$ consist
 of the vertex sets of the three upward paths in $H$ from the vertices of the face $t$ to
 the root $r$.
 
 Then, $T$ is a tree and $(T,\beta)$ is a tree decomposition of $G$. 
 In particular, for every edge $t_1t_2 \in E(T)$, the adhesion $\beta(t_1) \cap \beta(t_2)$
 is the vertex set of the two upward paths in $H$ from the endpoints of the common edge
 of the faces $t_1$ and $t_2$ to the root $r$. 
\end{theorem}

We will need also the following standard argument.

\begin{lemma}\label{lem:tree-decomp-split}
    Let $k, n > 0$. Let $G$ be a planar graph on $n$ vertices and let $(T, \beta)$ be a tree decomposition of $G$ such that the degree of every vertex in $T$ is at most $3$. Then, for any weight function $\chi : V(G) \to \{0, 1\}$, one of the following holds:
    \begin{itemize}
        \item there is a bag $b \in V(T)$ such that $\chi(\beta(b)) \geq \frac{1}{4k} \chi(V(G))$, or
        \item there exists a subset $F \subseteq E(T)$ of $k$ edges such that the union of bags of each connected component of $T - F$ has weight at least $\frac{1}{4k} \chi(V(G))$.
    \end{itemize}
\end{lemma}

\begin{proof}
    Let $W=\chi(V(G))$. Now, by induction, we assume that there is no bag $b\in V(T)$ such that $\chi(\beta(b))\geq W/4k$ and we have picked $i-1$ edges such that each connected component of $T-F$ has weight at least $W/4k$. We can always find a connected component $C$ of weight at least $W/i\geq W/k$.
    
    We now pick the $i^{th}$ edge in $C$ as follows: start with an arbitrary edge $(u, v)$ and check if it separates $C$ into two components of size at least $W/4k$ each. If not, let $H$ be the component with higher weight attached to $v$ (w.l.o.g.). The vertex $v$ cannot have degree $1$, since then $H$ is a single vertex with weight at least $3W/4k$. If $v$ has degree $2$, then take the potential cut edge to be $(v, w)$ (considering $w$ is the other neighbor of $v$). If $\chi(V(C- H)\cup \{\beta(v)\})\geq W/4k$, we are done. Otherwise, continue with the next edge towards $w$. If $v$ has degree $3$, let $v_1, v_2$ be the other neighbors of $v$, and let $T_1, T_2$ resp. be the subtrees corresponding to them. Since $\chi(\beta(v))<W/4k$, we have $\chi(V(T_1)\cup V(T_2))\geq W/2k$. This means one of $T_1$ or $T_2$ has weight at least $W/4k$. Assume $\chi(V(T_1))\geq W/4k$. If $\chi(C-T_1)\geq W/4k$, we cut the $(v, v_1)$ edge. If not, we continue towards $v_1$ and its other neighbors. Note that in this whole procedure, the weight of the lighter component is strictly increasing. Therefore, there must exist a cut edge with separates $C$ into components having weight at least $W/4k$ each. We perform this operation $k$ times to obtain a set of $k$ edges as desired.
\end{proof}

\subsection{The routing lemma}

We will also need the following routing lemma.
\begin{lemma}\label{lem:log-cover-spanner}
Let $G$ be a graph on $n$ vertices and let $T_1, \dots, T_k$ be a family of trees that are subgraphs of $G$ such that:
\begin{itemize}
    \item $V(T_1) \cap V(T_i) \neq \emptyset$ for every $i \in [k]$,
    \item $\bigcup_{i \in [k]} V(T_i) = V(G)$.
\end{itemize}
Then, there exists a spanning tree $H$ of $G$ rooted at some vertex $r$, such that for every $v \in V(G)$, the path from $v$ to $r$ can be partitioned into $\mathcal{O}(\log k \cdot \log |V(G)|)$ parts, each being a subpath of one of the trees $T_1, T_2, ..., T_k$.
\end{lemma}

\begin{proof}
    
    By our assumptions, $\cup_iT_i$ spans $V(G)$, and each $T_i$ has a non-empty intersection with $T_1$. We start by showing a randomized procedure that generates a spanning tree $H$ which achieves the $\Oh(\log k)$ bound in expectation.

    We choose $r$ as an arbitrary vertex of $V(T_1)$ and root each tree $T_i$ in an arbitrary vertex of $V(T_1) \cap V(T_i)$ for every $i \in \{2, \dots, k\}$.
    Let $\pi$ be a randomly and uniformly chosen permutation $\pi = (\pi_2, \pi_3, ..., \pi_k)$ of the set $\{2, \dots, k\}$. For the sake of simplicity, we put $\pi_1 = 1$. Let $<_\pi$ be the linear ordering on $\{T_1, \dots, T_k\}$ induced by $\pi$, i.e., we have $T_i <_\pi T_j$ iff $\pi^{-1}(i) < \pi^{-1}(j)$.

    For each $v \in V(G)$, we put $T_v$ to be the $<_\pi$-smallest tree $T_i$ which contains $v$. For $v \neq r$, we put $e_v$ to be the edge connecting $v$ to its parent in $T_v$. Note that such edge always exists as the roots of all $T_i$ belong to $V(T_1)$ and $T_1$ is $<_\pi$-smallest. We define $H$ as the union of all edges $e_v$ for all $v \in V(G) - \{r\}$. It's easy to see that $H$ is indeed a spanning tree of $G$.

    Finally, for the proof of the desired property. Pick any $v \in V(G) - \{r\}$ and let $P_v$ be the path from $v$ to $r$ in $H$.
    (We think of $P_v$ as oriented from $v$ to $r$.)
    Let $\mathcal{T}_v$ be the set of trees $T_i$ that intersect $P_v$
    Order $\mathcal{T}_v$ as $T_{i_1}, \ldots, T_{i_\ell}$ according to the
    earliest vertex of the intersection with $P_v$, breaking ties arbitrarily.
    We say that an index $a$ is \emph{important} if 
    $T_{i_a}$ is $<_\pi$-minimum among $\{T_{i_1}, \ldots, T_{i_a}\}$
    and \emph{used} if $T_{i_a} = T_u$ for some $u \in V(P_v)$. 
    Observe that a used index is necessarily important.
    
    Consider the statement $A_{a, v}$:
    the index $a$ is important for the vertex $v$.
    The crucial observation is that 
    \[ \Pr(A_{a,v}) \leq \frac{1}{a}. \]
    Indeed, the path $P_v$ up to the intersection with $T_{i_a}$
    does not depend on the position of the tree $T_{i_a}$ 
    in the order $\pi$, so 
    all $a$ positions of $T_{i_a}$ with respect
    to $\{T_{i_1}, \ldots, T_{i_{a-1}}\}$ in the order $<_\pi$
    are equally probable, and only one makes $a$ important for $v$.
    
    Now, let $X_v$ be the random variable for the number of trees $T_i$ needed to cover $P_v$.
    Observe that $X_v$ is bounded by the number of used indices for $v$, 
    as for every $u \in V(P_v) \setminus \{r\}$ we can cover the edge
    going upwards from $u$ with $T_u$. 
    This quantity is, in turn, bounded by the number of important indices for $v$.
    Hence, in expectation $E_{\pi}[X_v] \leq \Sigma_{a=1}^\ell \Pr(A_{a, v}) \leq \Sigma_{a=1}^{\ell} \frac{1}{a}$, which is the harmonic series; the finite harmonic series up to $L$ is known to be bounded by $\mathcal{O}(\log L)$.

    Now, we proceed with the deterministic construction. For a set
    $\emptyset \neq W \subseteq V(G)$, let $X_W$ be equal to $X_v$ assuming we sample $v$ uniformly randomly from $W$.
    Again, $E[X_W] = \Oh(\log k)$.
    By the Markov inequality, $P(X_W \geq 2E[X_W]) \leq \frac{1}{2}$, hence for at least half of all pairs $(\pi, v)$, the value of $X_v$ assuming $\pi$ is at most $2E[X_W]$. In particular, there exists a permutation $\pi_0$ s.t. $\Pr(X_W \leq 2E[X_W] \mid \pi = \pi_0) \geq \frac{1}{2}$.
    Thus, fixing $\pi = \pi_0$, the $v$-to-$r$ paths for half
    of the vertices $v \in W$ will be covered by fewer than $2E[X_W] = \mathcal{O}(\log k)$ trees. 

    We perform the above procedure iteratively. We first put $W = V(G)$, find $\pi_{0, V(G)}$ and corresponding $H$ (according to our randomized procedure), and fix the parent vertex according to $H$ for every $v$ for which we can cover the $v$-to-$r$ path in $H$ using at most $2E[X_W]$ trees. Afterwards, we restrict $W$ to vertices with parent not yet fixed and repeat this step.

    At every step, the vertices with parent already fixed form a connected subtree containing $r$. Every step fixes at least half of the vertices not yet fixed, hence the total number of steps will be at most $\Oh(\log |V(G)|)$. %
    Any $v$-to-$r$ path in the final tree obtained this way can be decomposed into $\Oh(\log |V(G)|)$ subpaths, one for each step, and each of those is coverable with $\Oh(\log k)$ trees, which gives us the desired $\Oh(\log k \cdot \log |V(G)|)$ bound.
\end{proof}

%% file: comatching2coreset.tex
\section{From Comatching to Coreset: Proof of~\Cref{thm:comatching2coreset}}\label{sec:comatchingtocoreset}

This section is devoted to the proof of~\Cref{thm:comatching2coreset}. 
We will need a classic tool: an LP rounding approximation algorithm for \textsc{Hitting Set}
in set systems of bounded VC dimension.

\subsection{\textsc{Hitting Set} LP}\label{subsec: hitLP}

For a family $\mathcal{A}$ of subsets of a universe $U$, a set $X \subseteq U$
is called a \emph{hitting set} of $\mathcal{A}$ if $X \cap A \neq \emptyset$ for every $A \in \mathcal{A}$.
The \textsc{Hitting Set} problem, given $U$ and $\mathcal{A}$, asks for a hitting set
of minimum possible size.

Below is the LP relaxation for the natural formulation of the \textsc{Hitting Set} problem
and its dual (which is a set packing problem). 

\begin{equation*}
  \begin{aligned}
  \mathrm{min} & \sum_{u \in U} x_u & \\
  \mathrm{s.t.} & \sum_{u \in A} x_u \geq 1 & \forall_{A \in \mathcal{A}} \\
  & x_u \geq 0 & \forall_{u \in U} \\
  \end{aligned}\qquad\qquad\qquad\qquad%
  \begin{aligned}
  \mathrm{max} & \sum_{A \in \mathcal{A}} y_A & \\
  \mathrm{s.t.} & \sum_{A \in \mathcal{A}~|~u \in A} y_A \leq 1 & \forall_{u \in U} \\
  & y_A \geq 0 & \forall_{A \in \mathcal{A}} \\
  \end{aligned}
\end{equation*}

We need the following classic result about rounding the \textsc{Hitting Set} LP 
in set systems of bouned VC-dimension.

\begin{theorem}[\cite{PachAgarwal,Matousek}]\label{thm:hs-round}
  There exists a universal constant $c$ such that every \textsc{Hitting Set}
  instance $(U,\mathcal{A})$ admits a solution of size at most $c \cdot d \cdot \tau^\ast \cdot \log(\tau^\ast)$,
  where $d$ is the VC-dimension of $\mathcal{A}$ and $\tau^\ast$ is the optimum value
  of the linear programming relaxation. Furthermore, such a solution can be found in polynomial time.
\end{theorem}

\subsection{The proof}

We are now ready to prove~\Cref{thm:comatching2coreset}. Fix $ \varepsilon > 0$.
As for $\varepsilon \geq 1$ the task of finding an $\varepsilon$-coreset for furthest neighbor
is trivial (just take any vertex), we assume $\varepsilon < 1$.
Let $(V,\dist)$ be a metric space, let $P \subseteq V(G)$, 
and let $d$ be the VC-dimension of the ball set system of $(V,\dist)$.
We assume $|P| > 1$, as otherwise we can just return $Q = P$.
Let
\[\Delta := \max_{u,v \in P} \dist(u,v) \]
be the diameter of $P$ in $(V,\dist)$ and let
\[ \delta := \frac{\varepsilon}{4} \cdot \Delta. \]
For every $v \in V$, pick $z_v \in P$ to be a vertex maximizing the distance
from $v$, breaking ties arbitrarily.

We first deal with vertices of $V$ that are very far from $P$.
Pick $p_0 \in P$ and let $V_\mathrm{far} = \{v \in V(G)~|~\dist(v, p_0) > \varepsilon^{-1} \Delta\}$. Note that for every $v \in V_\mathrm{far}$:
\[ \dist(v, z_v) - \dist(v, p_0) \leq \dist(z_v, p_0) \leq \Delta \leq \varepsilon \dist(v, p_0) \leq \varepsilon \dist(v, z_v). \]
Hence, to satisfy $V_\mathrm{far}$, it suffices to just add $p_0$ to the constructed coreset.

We partition $V \setminus V_\mathrm{far}$ into sets $V_i$ for nonnegative integers
$i$; a vertex $v$ belongs to $V_i$ if
\[ i \delta \leq \dist(v, z_v) < (i+1) \delta. \]
Note that, as $\Delta/2 \leq \dist(v,z_v) \leq \varepsilon^{-1}\Delta$, the set $V_i$ is nonempty
only for $ \lfloor 2 \varepsilon^{-1} \rfloor \leq i \leq \lfloor 4 \varepsilon^{-2} \rfloor$. 

Fix $i$ for which $V_i \neq \emptyset$. For $v \in V_i$ define
\[ S_i^v = \{u \in P \mid \dist(v,u) \geq (i-1) \delta \}. \]

Clearly, $z_v \in S_i^v$.
We consider set system $\mathcal{A}_i = \{S_i^v \mid v \in V_i\}$
of subsets of $P$.
First, we observe that it suffices to find a hitting set of $\mathcal{A}_i$.
\begin{claim}\label{cl:hs2coreset}
  Any hitting set of $\mathcal{A}_i$ is also an $\varepsilon$-coreset for furthest neighbor
  of $V_i$.
\end{claim}
\begin{proof}
  Let $X$ be a hitting set of $\mathcal{A}_i$.
  Let $v \in V_i$ and let $x \in X \cap S_i^v$, which exists as $X$ is a hitting set. 
  Then, by the definition of $S_i^v$, 
  \[ \dist(v, x) \geq (i-1) \delta. \]
  Hence,
  \[ \dist(v,z_v) - \dist(v, x) \leq 2\delta. \]
  On the other hand, $\dist(v,z_v) \geq \Delta/2$. Thus, 
  \[ \dist(v,z_v) - \dist(v,x) \leq 2\delta \leq 4 \cdot \frac{\delta}{\Delta} \cdot \dist(v,z_v) \leq \varepsilon \cdot \dist(v,z_v). \]
  Thus, indeed $X$ is an $\varepsilon$-coreset for furthest neighbor of $V_i$. 
\end{proof}

By~\Cref{thm:hs-round}, to bound the size of a hitting set of $\mathcal{A}_i$
it suffices to bound the optimum value of its LP relaxation;
let $\tau_i^\ast$ be this optimum value.
We connect $\tau_i^\ast$ to the bound on the size of a~comathing in $(V,\dist)$
via the following simple rounding.

\begin{claim}\label{cl:lp2comatching}
There exists an $(\varepsilon^2/4)$-comatching in $(V,\dist)$ of size
$\lfloor \tau_i^\ast/4 \rfloor$.
\end{claim}
\begin{proof}
Consider an optimum solution to the dual LP to \textsc{Hitting Set} on $\mathcal{A}_i$;
for clarity of notation, we use $(y_v)_{v \in V_i}$ for the dual variables
instead of $(y_{S_i^v})_{v \in V_i}$. 
Clearly, $\sum_{v \in V_i} y_v = \tau_i^\ast$.
Let $\mu$ be a probability distribution over $V_i$, where $v \in V_i$ is sampled 
with probability $y_v / \tau_i^\ast$.
For $u,v \in V_i$, we say that \emph{$u$ threatens $v$} if $z_v \in S_i^u$. 
Note that any $v \in V_i$ threatens itself. 
Let $v \in V_i$ be fixed and let $u$ be sampled from $V_i$ according to $\mu$.
Then,
\begin{equation}\label{eq:threaten}
 \mathrm{Prob}(v \text{~is threatened by}~u) \leq  \sum_{u: z_v\in S^u_i} \frac{y_u}{\tau^*} \leq \frac{1}{\tau_i^\ast}. 
\end{equation}

Indeed, \eqref{eq:threaten} follows directly from the dual LP constraint for the 
vertex $z_v$:
\[ \sum_{w \in V_i \mid z_v \in S_i^w} y_w \leq 1. \]

Let $K := \lfloor \tau_i^\ast / 4 \rfloor$ and let $v_1,v_2,\ldots,v_{2K}$ be vertices
sampled independently from $V_i$ according to the distribution $\mu$.
Then, by~\eqref{eq:threaten},
the expected number of pairs $(i,j)$, $1\leq i,j \leq 2K$, $i \neq j$ such that
$v_i$ threatens $v_j$ is at most 
 \[ 4K^2 \cdot \frac{1}{\tau_i^\ast} \leq K. \]
Therefore, there exists a choice of $v_1,v_2,\ldots,v_{2K}$ such that the number of such pairs
$(i,j)$ is at most $K$. 

Pick such a choice and, for every pair $(i,j)$, $1 \leq i,j \leq 2K$, $i \neq j$,
 such that $v_i$ threatens $v_j$, discard $v_i$. We are left with a set $W$
 of at least $K$ vertices of $V_i$ such that for every $u,v \in W$, $u \neq v$, 
 $u$ does \emph{not} threaten $v$. 

 We claim that $\{(v,z_v)~|~v \in W\}$ is the desired $(\varepsilon^2/4)$-comatching in $G$ for $R=i \delta$;
 it clearly has size at least $K$.
 Since $W \subseteq V_i$, we have for every $v \in W$
  \[ \dist(v,z_v) \geq i \delta. \]
On the other hand, for every $u,v \in W$, $u \neq v$, as $u$ does not threaten $v$,
we have $z_v \notin S_i^u$ so 
 \[ \dist(u, z_v) \leq (i-1)\delta. \]
As $i \leq \lfloor 4\varepsilon^{-2} \rfloor$, we have
\[ \frac{i-1}{i} \leq 1 - \frac{1}{\lfloor 4\varepsilon^{-2} \rfloor} \leq 
1 - \frac{1}{4\varepsilon^{-2}} = 
1 - \frac{\varepsilon^2}{4}. \]
\end{proof}

By~\Cref{cl:lp2comatching}, $\tau_i^\ast \leq 4L+3$, where $L$ is the maximum
size of an $(\varepsilon/4)$-comatching in $(V,\dist)$.
Every element of $\mathcal{A}_i$ is a complement of a ball, and hence 
$\mathcal{A}_i$ has VC-dimension at most $d$.
By~\Cref{thm:hs-round}, there exists a hitting set $X_i$ of $\mathcal{A}_i$
of size $\Oh(d L \log L)$ and we can find such a hitting set in polynomial time. \Cref{cl:hs2coreset} asserts that $X_i$ is an $\varepsilon$-coreset
for furthest neighbor of $V_i$. 
Hence, for the $\varepsilon$-coreset of the whole $V$, it suffices to take the union of $\{p_0\}$ and 
of $\Oh(\varepsilon^{-2})$ sets $X_i$ for which $V_i$ is nonempty. 
This finishes the proof of~\Cref{thm:comatching2coreset}.

%% file: comatching.tex
\section{Comatching Bound in Planar Metrics: Proof of~\Cref{thm:comatching-bound}}\label{sec:comatching-bound}

This section is devoted to the proof of~\Cref{thm:comatching-bound}, which bounds the size of $\varepsilon$-comatching. We begin by repeatedly reducing $\varepsilon$-comatching, each time paying a polynomial-in-$\eps$ cost to add new constraints on the topology, in order to isolate pairs $(p, q)$, where $p, q$ are, in a certain sense, close to each other. Then we pick some landmarks from the remaining pairs which, after one more step of reduction where we only keep pairs at the same distances from the landmarks, we will use to define a linear order. At this point we have established the ultimate regime. Here we will identify two symmetric cases on the web of intersections of shortest paths between the pairs. Finally, rearranging the path segments between those intersections, in each case we obtain a combinatorial contradiction on the sum of minimum and maximum lengths of shortest paths between pairs as per the comatching definition.

Fix $\varepsilon > 0$. 
We need some notation. A pair $(G,\mathcal{M})$ is an \emph{instance}
if $G$ is a plane edge-weighted graph and $\mathcal{M}$ is an $\varepsilon$-comatching in $G$. 
That is, we assume that $G$ comes with a fixed embedding in a~plane. The vertices $p$ and $q$ for $(p,q) \in \mathcal{M}$ are called \EMPH{terminals}.

We assume that the lengths of all simple paths are pairwise different.
This can be realized by assuming an arbitrarily chosen linear order $\preceq$ on edges of $G$ and breaking ties between paths of same total weight by taking the one with the set of edges lexicographically $\preceq$-smaller than the other.
With this assumption, the shortest path between any pair of terminals $p_1, q_2$ for any $(p_1, q_1), (p_2, q_2) \in \mathcal{M}$ is unique. We will denote this path by $\ppath{p_1 q_2}$ or $\ppath{q_2 p_1}$. Moreover, for any triple of vertices $u, p, q$, the shortest paths from $u$ to $p$ and from $u$ to $q$ can share only an initial segment: after branching out for the first time, they remain disjoint.

This allows us to define for $\mathcal{N} \subseteq \mathcal{M}$ an instance
$\instind{\mathcal{N}} = (\Gind{\mathcal{N}}, \mathcal{N})$,
where $\Gind{\mathcal{N}}$ is the subgraph of $G$ consisting only of paths $\Mpath{p_1}{q_2}$ for 
$(p_1,q_1), (p_2,q_2) \in \mathcal{N}$, $(p_1,q_1) \neq (p_2,q_2)$. 
For a pair $\sigma \in \mathcal{M}$, we use $\Gexcl{\sigma}$ and $\instexcl{\sigma}$ as shorthands for
 $\Gind{\mathcal{M} \setminus \{\sigma\}}$ and $\instind{\mathcal{M} \setminus \{\sigma\}}$.

 It will be convenient not to have any other terminals on a path $\Mpath{p_1}{q_2}$ except for the endpoints.
 To this end, we can perform the following modification: for every $(p,q) \in \mathcal{M}$, 
 create two fresh vertices $p'$ and $q'$, each of degree one connected by edges $pp'$ and $qq'$
 of infinitesimal small weight, and replace $(p,q)$ with $(p',q')$ in $\mathcal{M}$. 
 The resulting instance will be a $(\varepsilon')$-comatching for any $\varepsilon' < \varepsilon$
 (formally, for any $\varepsilon' < \varepsilon$, we can choose sufficiently small weight of
 the newly created edges so that we obtain an $(\varepsilon')$-comatching).

 Consider the following contraction operation: pick any vertex $u$ of degree 2 connected to vertices $x_1, x_2$ with edges of weight respectively $w_1$ and $w_2$; replace it with a direct edge of weight $w_1 + w_2$ between $x_1$ and $x_2$. Such operation preserves the distances between any pair of remaining vertices. In particular, if our instance is an $\varepsilon$-comatching and has all terminals of degree 1, contracting all vertices of degree $2$ in this way results in an instance again being an $\varepsilon$-comatching. This modification will allow us to bound the size of our graph polynomially in the size of $\mathcal{M}$, which will be needed in a future section.

 These observations allow us to introduce the following definition: an instance $\inst = (G,\mathcal{M})$ 
 is \EMPH{clean} if 
 \begin{itemize}
 \item the terminals are pairwise distinct and of degree $1$ each;
 \item $G$ contains no vertices of degree $2$;
 \item $G = \Gind{\mathcal{M}}$, that is, every edge of $G$ lies on a path $\Mpath{p_1}{q_2}$ for some distinct
 $(p_1,q_1),(p_2,q_2) \in \mathcal{M}$.
 \end{itemize}
 It suffices to prove the bound on the size of $\mathcal{M}$ of~\Cref{thm:comatching-bound} only for clean
 instances.

 Observe that two paths $\ppath{p_a q_b}$ and $\ppath{p_c q_d}$ with pairwise
 distinct endpoints may either (a) be completely disjoint, 
 (b) share a subpath, but do not cross (i.e., $\ppath{p_c q_d}$ leaves
 $\ppath{p_a q_b}$ twice on the same side); in this case we say
 that they \EMPH{touch};
 (c) share a subpath, and cross (i.e., $\ppath{p_c q_d}$ leaves
 $\ppath{p_a q_b}$ on the other sides); in this case we say that
 they \emph{intersect transversally}. More formally, if $\ppath{p_a q_b}$ and $\ppath{p_c q_d}$ intersect, let $\sigma$ be any closed loop containing $\ppath{p_a q_b}$ and not intersecting $\ppath{p_c q_d} \setminus \ppath{p_a q_b}$. We say that $\ppath{p_a q_b}$ and $\ppath{p_c q_d}$ intersect transversally if and only if exactly one of $p_c$ and $q_d$ is inside the interior of $\sigma$. We note as $p_c$ and $q_d$ are of degree one, they cannot lie on $\sigma$ and that suitable $\sigma$ always exists and that this definition does not depend on the choice of $\sigma$ as $p_a$ and $q_b$ are of degree one.

 Let $\inst = (G,\mathcal{M})$ be a clean instance. 
 Observe that for a pair $\sigma = (p,q) \in \mathcal{M}$, the terminals $p$ and $q$ do not belong to $\Gexcl{\sigma}$
 and their drawing in $G$ lie inside some faces of $\Gexcl{\sigma}$. 
 We say that a clean instance $\inst$ is \EMPH{local} if for every $\sigma \in \mathcal{M}$, the terminals
 of $\sigma$ lie in the same face of $\Gexcl{\sigma}$. 
 The first step of the proof of~\Cref{thm:comatching-bound} is a reduction that allows to focus only
 on local instances. This reduction is presented in~\Cref{ss:to-local}.

 \begin{lemma}\label{lem:to-local}
    For a real $\varepsilon > 0$, let $L_{\mathrm{local}}(\varepsilon)$ be the maximum size
    of a $\varepsilon$-comatching in a local instance. 
    Then, the maximum size of a $\varepsilon$-comatching in any instance is of
    $\Oh(L_{\mathrm{local}}(\varepsilon)^9 \varepsilon^{-8})$. 
 \end{lemma}

\begin{figure}[h]
    \centering
    \label{fig:local_noncrossing}
    \begin{subfigure}{0.45\textwidth}
        \centering        \includegraphics[width=\linewidth]{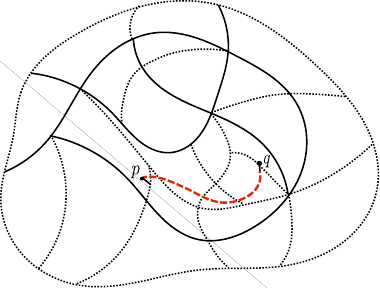}
        \caption{Local instance}
    \end{subfigure}
    \hfill
    \begin{subfigure}{0.45\textwidth}
        \centering        \includegraphics[width=\linewidth]{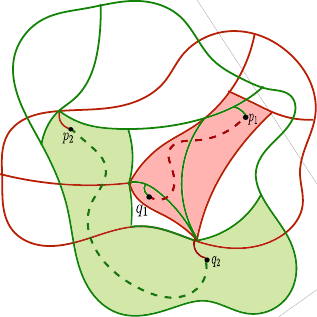}
        \caption{Linked non-crossing instance}
    \end{subfigure}
    \caption{The local instance in (a) is depicted with respect to the pair $\sigma = \{p, q\}$. The graph $G^{-\sigma}$ is shown in red overlayed on $G$. The linked non-crossing instance is shown in (b). Note that the links $\gamma_1$ and $\gamma_2$ do not intersect each other as well as $G^{-\sigma_1}$ and $G^{-\sigma_2}$ respectively.}
\end{figure}

 Let $\inst = (G,\mathcal{M})$ be a local instance. 
 A curve $\gamma$ is in \emph{good position} with regards to $G$
 if it starts and ends in a terminal, besides endpoints intersects the drawing of $G$ only in interiors of edges,
 and for every edge $e \in E(G)$, $\gamma$ intersects the drawing of $e$ in at most one point, and always transversally. %
 For a pair $\sigma \in \mathcal{M}$, a \EMPH{link} is a curve $\gamma_\sigma$ in good position
 w.r.t. $G$, with two endpoints in the two terminals of $\sigma$ and that does not intersect any edge
 of $\Gexcl{\sigma}$ (i.e., $\gamma_\sigma$ is completely contained in the face of $\Gexcl{\sigma}$ that contains
 the terminals of $\sigma$). 

 A \EMPH{linked instance} is a pair $(\inst,\Gamma)$ where $\inst=(G,\mathcal{M})$ is a local instance
 and $\Gamma$ is a function that assigns a link $\Gamma(\sigma)$ to every $\sigma \in \mathcal{M}$. 
 A linked instance is \EMPH{noncrossing} if the links $\{\Gamma(\sigma)~|~\sigma \in \mathcal{M})$
 are pairwise disjoint (as subsets of the plane). 

 The second step of the proof of~\Cref{thm:comatching-bound} is a reduction from a local
 instance to a linked noncrossing instance. This reduction is presented in~\Cref{ss:to-noncrossing}.

 \begin{lemma}\label{lem:to-noncrossing}
    For a real $\varepsilon > 0$, let $L_{\asymp}(\varepsilon)$ be the maximum size
    of a $\varepsilon$-comatching in a linked noncrossing instance. 
    Then, the maximum size of a $\varepsilon$-comatching in a local instance is of 
    $\Oh(L_{\asymp}(\varepsilon)^2 \log^8 L_{\asymp}(\varepsilon) \cdot \varepsilon^{-8} )$. 
 \end{lemma}

 The last and final step is a polynomial-in-$\varepsilon^{-1}$ bound on the maximum size of A
 linked noncrossing instance via a careful analysis of the topology of paths in such instance.
 This is presented in~\Cref{ss:noncrossing-bound}.

 \begin{lemma}\label{lem:noncrossing-bound}
    For a real $\varepsilon > 0$, the maximum size of a $\varepsilon$-comatching
    in a linked noncrossing instance is at most $\varepsilon^{-14}$.
 \end{lemma}

 Combining the above lemmas allows us to prove~\Cref{thm:comatching-bound} restated below.
 \CoMatchingTheorem*
 \begin{proof}%
     Fix any instance $\mathcal{I} = (G,\mathcal{M})$. As argued above, we can assume that $\mathcal{I}$ is clean. Let $L$ denote the size of its maximum $\varepsilon$-comatching. By \Cref{lem:to-local}, $L = \Oh(L_{\mathrm{local}}(\varepsilon)^9 \varepsilon^{-8})$. By \Cref{lem:to-noncrossing}, $L = \Oh(L_{\asymp}(\varepsilon)^{18} \log^{72} L_{\asymp}(\varepsilon) \cdot \varepsilon^{-80})$. Finally, by \Cref{lem:noncrossing-bound}, $L = \Oh(\varepsilon^{-332} \log^{72} (\varepsilon^{-1}))$, which finishes the proof.
 \end{proof}

 \subsection{Distance profiles and separation toolbox}\label{ss:separation}

 In all the subsequent subsections, we implicitly define $\mathcal{I} = (G, \mathcal{M})$ to be a clean $\varepsilon$-comatching instance with at least three terminal pairs. We put $L = |\mathcal{M}|$ and enumerate $\mathcal{M} = \{(p_1, q_1), \dots (p_L, q_L)\}$. We will denote 
 \[ R := \min_{(p,q) \in \mathcal{M}} \dist(p,q) \]
 and set 
 \[ \delta := \frac{\varepsilon}{100} R. \]
Intuitively, $\delta$ will serve as one ``unit'' of distance, and differences of the order of a few
$\delta$s do not matter. 
In particular, for every distinct $(p_i,q_i),(p_j,q_j) \in \mathcal{M}$, we have
\[ \dist(p_i,q_j) \leq (1-\varepsilon) R = R - 100\delta. \] 
For two vertices $u,v \in V(G)$, by $\Udist(u,v)$ we denote the unique integer $i$
such that $i \cdot \delta \leq \dist_G(u,v) < (i+1) \cdot \delta$, that is,
$\Udist(u,v) = \lfloor \frac{\dist_G(u,v)}{\delta} \rfloor$.

\begin{claim}\label{diam-lnci}
    The diameter of $G$ is at most $3 R$.
\end{claim}
\begin{proof}
    Pick any two vertices $v_1, v_2 \in V(G)$. As our instance is clean, they are contained in some paths $\ppath{p_{i_1} q_{j_1}}$ and $\ppath{p_{i_2}, q_{j_2}}$ respectively. Since $\dist(p_{i_1}, q_{j_1}) + \dist(p_{i_2}, q_{j_2}) < 2R$, either $\dist(v_1, p_{i_1}) + \dist(v_2, p_{i_2}) < R$ or $\dist(v_1, q_{i_1}) + \dist(v_2, q_{i_2}) < R$. Assume w.l.o.g. the former.
    
    Pick any pair $(p_{i_3}, q_{i_3}) \in \mathcal{M}$ for some $i_3 \not\in \{i_1, i_2\}$.
    There exists a path of form $v_1 \to p_{i_1} \to q_{i_3} \to p_{i_2} \to v_2$ moving along respective paths between endpoints of $\mathcal{M}$. The total length of this path is at most $3R$, which finishes the proof of the lemma.
\end{proof}
As an immediate corollary, we obtain that $\Udist(u, v) \leq 3R \cdot \delta^{-1} = 300 \varepsilon^{-1}$ for every $u,v \in V(G)$.

For a set $Z \subseteq V(G)$, the \emph{rounded distance profile} of a vertex $v \in V(G)$
is the function 
$ \Uprof{Z}{v} : Z \to \mathbb{Z}$ defined as $\Uprof{Z}{v}(z) = \Udist(v,z)$. ~\Cref{thm:prof-vc} immediately implies the following.
\begin{lemma}\label{lem:Uprof-vc}
 Let $Z \subseteq V(G)$.
 For $v \in V(G)$, let 
 \[ \widehat{\mathcal{Z}}_Z[v] = \{(z,i) \in Z \times \{0,1,\ldots,\lfloor 300\varepsilon^{-1} \rfloor\}~|~\Uprof{Z}{v}(z) \leq i\}.\]
Then, the set system $\{\widehat{\mathcal{Z}}_Z[v]~|~v \in V(G)\}$ has VC-dimension at most $4$. 
In particular, by the Sauer-Shelah Lemma, its size is bounded by $\Oh(|Z|^4 \varepsilon^{-4})$. 
\end{lemma}

 We need the following observations.

 \begin{lemma}\label{lem:sep-two-pairs}
    Let $Z \subseteq V(G)$ and let $(p_i,q_i),(p_j,q_j)
    \in \mathcal{M}$ be two distinct pairs such that $p_i$ and $p_j$ have the same rounded distance profile
    to $Z$. Then, $\Mpath{p_i}{q_j}$ does not contain a vertex that is within distance
    $\delta$ from $Z$. 
 \end{lemma}
 \begin{proof}
    Assume the contrary: let $u$ be a vertex on $\Mpath{p_i}{q_j}$ such that there exists $z \in Z$
    with $\dist(z,u) \leq \delta$. 
    As $\Uprof{Z}{p_i} = \Uprof{Z}{p_j}$, we have $|\dist(p_i,z) - \dist(p_j,z)| < \delta$. 
    Hence,
    \[ |\dist(p_i,u) - \dist(p_j,u)| \leq |\dist(p_i,z) - \dist(p_j,z)| + 2\dist(u,z) \leq 3\delta. \]
    As $u$ lies on $\Mpath{p_i}{q_j}$, which is a shortest path from $p_i$ to $q_j$, 
    we infer that $\dist(p_j,q_j) \leq \dist(p_j, u) + \dist(u, q_j) \leq \dist(p_i, u) + \dist(u, q_j) + 3 \delta = \dist(p_i,q_j) + 3\delta$, which is a contradiction
    to $\dist(p_j,q_j) \geq R$ and $\dist(p_i,q_j) \leq R - 100\delta$.
\end{proof}

 We will say that the pair $(p_i,q_i)$ is \emph{separated} by a set $X \subseteq V(G)$ if $p$ and $q$ are not in the same connected
 component of $G-X$. 

 \begin{lemma}\label{lem:sep-three-pairs}
    Let $X,Z \subseteq V(G)$ be such that every
    vertex of $X$ is within distance $\delta$ from some vertex of $Z$ 
    and let $(p_i,q_i),(p_j,q_j),(p_k,q_k) \in \mathcal{M}$ be three pairwise distinct pairs. 
    If all three vertices $p_i$, $p_j$, $p_k$ have the same rounded distance profile to $Z$,
    then neither of these three pairs is separated by $X$.
 \end{lemma}
 \begin{proof}
    Assume the contrary: w.l.o.g., $(p_i,q_i)$ is separated by $X$. 
    Consider a walk in $G$ being a concatenation of $\Mpath{p_i}{q_j}$, $\Mpath{q_j}{p_k}$, and $\Mpath{p_k}{q_i}$.
    Since $(p_i,q_i)$ is separated by $X$, one of this paths intersects $X$, i.e., contains a vertex
    within distance $\delta$ from $Z$. This is a contradiction with~\Cref{lem:sep-two-pairs}.
 \end{proof}

\begin{restatable}{claim}{sep-bound}\label{lem:sep-bound}
    Let $X \subseteq V(G)$ be such that $X$ is contained in the union of at most $\alpha$ paths,
    each of length at most $\beta R$ for some $\alpha,\beta > 0$. 
    Then, the number of pairs of $\mathcal{M}$ separated by $X$ 
    is $\Oh(\alpha^4 \varepsilon^{-4} (1+\beta\varepsilon^{-1})^4)$.
\end{restatable}
 \begin{proof}
    Since $X$ is contained in the union of at most $\alpha$ paths, each of length at most $\beta R$,
    there exists a set $Z \subseteq V(G)$ of size $\Oh(\alpha (1+ \beta \varepsilon^{-1}))$
    such that every vertex of $X$ is within distance $\delta$ from some vertex of $Z$. 
    By~\Cref{lem:Uprof-vc}, there are
    $\Oh(|Z|^4 \varepsilon^{-4}) = \Oh(\alpha^4 \varepsilon^{-4}(1+\beta \varepsilon^{-1})^4)$
    distinct rounded distance profiles to $Z$. 
    Group pairs $(p_i,q_i) \in \mathcal{M}$ according to the rounded distance profile of $p_i$ to $Z$. 
    ~\Cref{lem:sep-three-pairs} implies that for group that contains at least three pairs,
    these pairs are not separated by $X$. Consequently, the number of pairs separated by $X$ 
    is $\Oh(\alpha^4 \varepsilon^{-4} (1+\beta \varepsilon^{-1})^4)$.
 \end{proof}

%% file: reductions_bounds.tex
\subsection{Reduction to local instances}\label{ss:to-local}

In this subsection, we provide a proof of~\Cref{lem:to-local}. 
Consider a function $\Lambda$ that maps any subset $I$ of ($L_{\mathrm{local}}(\varepsilon)+1$) pairs of points from $\mathcal{M}$ to a subset of $L_{\mathrm{local}}(\varepsilon)$ pairs of points.
The mapping is defined as follows: given a subset $I$ of size $L_{\mathrm{local}}(\varepsilon)+1$, we put $\Lambda(I) = I - \{(p_a, q_a)\}$, where $(p_a, q_a) \in \mathcal{M}$ is an arbitrarily fixed pair such that $p_a$ and $q_a$ lie on different faces of $G^{I - \{(p_a, q_a)\}}$. Since $L_{\mathrm{local}}(\varepsilon)$ is the maximum size of a local instance, such a pair always exists.

By pigeonhole principle, there exists a set $S \in \binom{\mathcal{M}}{L_{\mathrm{local}}(\varepsilon)}$ such that
\[|\Lambda^{-1}(S)| \geq \binom{L}{L_{\mathrm{local}}(\varepsilon)+1}/\binom{L}{L_{\mathrm{local}}(\varepsilon)} = \frac{L-L_{\mathrm{local}}(\varepsilon)}{L_{\mathrm{local}}(\varepsilon)+1}>\frac{L}{L_{\mathrm{local}}(\varepsilon)+1}-1\] 
Let us denote this quantity by $C = \frac{L}{L_{\mathrm{local}}(\varepsilon)+1}-1$. There exist at least $C$ indices $\{a_1, a_2, ..., a_C\}$ such that $\{p_{a_i}, q_{a_i}\}\notin S$ and $G^S$ separates $p_{a_i}$ and $q_{a_i}$.

Since $G^S$ is a union of not more than $L_{\mathrm{local}}(\varepsilon)^2$ paths, each of length at most $R$, we can apply~\Cref{lem:sep-bound} to get $C = \Oh(L_{\mathrm{local}}(\varepsilon)^8 \varepsilon^{-8})$. This in turn implies that $L = \Oh(L_{\mathrm{local}}(\varepsilon)^9 \varepsilon^{-8})$, which finishes the proof of the lemma.

\subsection{Reduction to linked noncrossing instances}\label{ss:to-noncrossing}

Assume that $\mathcal{I}$ is local. We would like $G$ to be a triangulation, hence we subdivide repeatedly all non-triangular faces by arbitrarily adding new edges of infinite weight inside them. We will call such edges \emph{virtual}. Such operation preserves all distances and shortest paths in $G$, however note that $\mathcal{I}$ is no longer clean after that operation.

Pick any pair $\sigma = (p_i, q_i)$ and look at the dual graph $G^*$. Let $F_p, F_q$ denote arbitrary faces of $G$ (hence vertices in $G^*$) incident to resp. $p_i$ and $q_i$. Since our instance is local, there exists a path from $F_p$ to $F_q$ in $G^*$ which does not cross any edge of $G^{\neg \sigma}$. Take $\gamma_\sigma$ to be the drawing of such path, at both ends extended within resp. $F_p, F_q$ to reach resp. $p_i$ and $q_i$. Doing this for each pair $\sigma$ gives us a linked instance. Moreover, in this instance, each link $\gamma_\sigma$ visits any face of $G$ at most once.

For the first step of the reduction to linked noncrossing instances we will utilize a classic tree decomposition method.

\subsubsection{Tree decomposition and separators}

We would like to take the tree decomposition of $G$ as described in~\Cref{thm:tree-decomp}. Its every adhesion induces a separator and the idea is to divide the plane using a number of such separators into a disjoint regions and argue that for every such region, there exists a link fully contained in it. Restricting our instance to one such link per region would result in it being non-crossing.

Assume we have some spanning tree $H$ of $G$ containing only non-virtual edges of $G$ (hence belonging to some $\ppath{p_i q_j}$). Consider a tree decomposition $(T, \beta)$ obtained from \Cref{thm:tree-decomp} for $H$. Let $S$ be a separator induced by an adhesion of the edge $e \in E(T)$ in said decomposition. As stated above, $S$ is formed by taking two upward paths of $H$ with endpoints connected via a closing edge of $G$ not contained in $H$.
For convenience, by $S$ we will refer to the subgraph of $G$ containing the vertices of the separator and with edges restricted only to those two upward paths (with the closing edge excluded).

By $\xi_S$ we denote the planar loop obtained by taking the drawing of both upward paths along with the closing edge. Consider two connected components into which $T$ splits after removal of $e$. All the vertices contained in the interior of $\xi_S$ belong to bags within only one of those components, and all the vertices contained in the exterior of $\xi_S$ belong to bags within the other.

We will say that a pair $\sigma \in \mathcal{M}$ is \emph{bad} with respect to $S$ if the link $\gamma_\sigma$ intersects $\xi_S$. Our goal is to argue that $H$ can be chosen in a way that results in a very small fraction of pairs being bad with respect to any such separator $S$.

We consider two cases of bad pairs. We will say that a pair $\sigma$ is \emph{separated} by $S$ if its endpoints are separated by $S$. We will say that $\sigma$ \emph{contaminates} $S$ if for some edge $e \in S$ it holds that $e \not\in G^{\neg \sigma}$. These cases are indeed exhaustive, which is shown in the following claim.

\begin{restatable}{claim}{types_of_bad}\label{clm:types_of_bad}
    Every bad pair $\sigma$ w.r.t. $S$ is either separated by $S$ or contaminates $S$.
\end{restatable} 
\begin{proof}
    Assume $\sigma$ is bad but not separated by $S$. This means that both endpoints of $\sigma$ lie on the same side of the loop $\xi_S$ on the plane, but the link $\gamma_\sigma$ still intersects $\xi_S$. By our construction, it must cross at least two different edges of $\xi_S$ and at most one of those can be the closing edge not contained in $S$. Thus, for some $e \in E(S)$, $\gamma_\sigma$ crosses $e$, hence $e \not\in G^{\neg \sigma}$.
\end{proof}

\subsubsection{The easily-coverable spanning tree}

To bound the number of both types of bad pairs, we present the construction of a particular spanning tree $H$, with a property that its every subpath can be covered by a small number of paths $\llbracket p_i q_j \rrbracket$.
More formally, for each $i \in [L]$, we define a tree $T_i$ as the union of paths $\cup_{j: j \neq i} \llbracket p_i q_j\rrbracket$. Such union is indeed a tree, given our guarantee that any two shortest paths remain disjoint after branching out for the first time.
Before we show the construction, we will need the following claim.

\begin{restatable}{claim}{instance_size_bound}\label{clm:instance_size_bound}
The number of vertices of $G$ is of $\Oh(L^4)$.
\end{restatable}
\begin{proof}
    Let $G'$ be $G$ before it was triangulated. We clearly have that $|V(G)| = |V(G')|$. As the instance corresponding to $G'$ is clean, by the definition of it being clean, we have that $G'$ is exactly the union of all trees $T_i$. There are $2L$ vertices of degree $1$ (terminals) and the remaining vertices are of degree at least $3$ (again, because our instance is clean). Every vertex of degree at least $3$ either has degree at least $3$ in one of trees $T_i$, or lies at an intersection of two different shortest paths between terminals, with no common terminals. 
    
    In the first case, it is well known that in any tree, the number of vertices of degree at least $3$ is at most the number of its leaves. Since we have $L$ trees and each of those has at most $L$ leaves, this gives us $L^2$ such vertices total.

    In the second case, we have at most $L^2$ different paths that can intersect. By our unique path assumption, every intersection of any two paths is a single subpath, hence such pair can only contribute to at most $2$ vertices of degree $3$ in $G'$, i.e., both endpoints of said common subpath. This gives us at most $2L^4$ such vertices total.
\end{proof}

Assuming $L \geq 3$, the trees $T_i$ satisfy the conditions of~\Cref{lem:log-cover-spanner}: let $H$ be the spanning tree resulting from the lemma. 
By~\Cref{clm:instance_size_bound}, $\log |V(G)| = \Oh(\log L)$, 
so in fact any path in $H$ can be decomposed into $\Oh(\log^2 L)$ paths,
each contained in a single tree $T_i$.

We will argue that we can easily bound the number of bad pairs w.r.t. any separator induced by our tree decomposition. Note, that any separator $S$ induced by an adhesion in $H$ is now coverable by $\Oh(\log^2 L)$ shortest paths in $G$, each of them being a subpath of some $T_i$. This follows from the fact that all edges of such separator always form two vertex-to-root paths in $H$.

\begin{claim}\label{clm:noncrossing-pairs-sep}
    There are at most $\Oh(\log^8 L \cdot \varepsilon^{-8})$ pairs in $\mathcal{M}$ that are \emph{separated} by $S$. 
\end{claim}
\begin{proof}
    Every subpath of any $T_i$ is of length at most $2R$. Let $\alpha$ denote the smallest possible size of a~collection of subpaths of $T_i$'s covering whole $S$ ($\alpha = \Oh(\log^2 L)$ by our previous observation). The claim now follows by applying~\Cref{lem:sep-bound} with said $\alpha$ and $\beta = 2$.
\end{proof}

\begin{claim}\label{clm:noncrossing-pairs-cont}
There are at most $\Oh(\log^2 L)$ pairs in $\mathcal{M}$ that are \emph{contaminating} $S$.
\end{claim}

\begin{proof}
    Let $\mathcal{T}$ denote the smallest subset of $(T_i)_{i=1}^L$ covering whole $S$.
    Each path in a tree $T_i$ consists of at most two subpaths of terminal to terminal paths (like $\ppath{p_i q_j}$)
    for some $j \neq i$; hence there is a set $\mathcal{P}$
    of $\Oh(\log^2 L)$ terminal to terminal paths that
    cover the whole $S$.
    
    For any contaminating pair $\sigma$, there is an edge $e_\sigma \in E(S)$ such that $e_\sigma \not\in G^{\neg \sigma}$.
    The edge $e_\sigma$ needs to be covered by $\mathcal{P}$ by a path
    with endpoint in $\sigma$.
    Since there are at most $2|\mathcal{P}| = \Oh(\log^2 L)$ 
    pairs $\sigma$ with endpoint in $\mathcal{P}$, the claim follows.
\end{proof}

The remaining pairs are \emph{good} with respect to $S$, that is, they are neither separated by $S$ nor contaminate $S$.

\subsubsection{Good pairs with noncrossing links}

We are now ready to prove the main lemma of the subsection.
Here is the main idea behind the proof: because the number of bad pairs is small w.r.t. any separator, if there are sufficiently many pairs in $\mathcal{M}$, then we can choose a balanced separator $S$ of $T$ such that there exists a pair which is good and lies inside $\xi_S$. Furthermore, by our definitions, the link of such pair must be contained entirely within $\xi_S$.

In the following proof, we want to find sufficiently many good pairs whose links are disjoint. And, as a direct consequence of the above observation, if we can find good pairs lying on different faces of some separators, then the desired condition will be satisfied. 

\begin{proof}[Proof of \Cref{lem:to-noncrossing}]
    Consider a spanning tree $H$ produced by \Cref{lem:log-cover-spanner} and consider a tree decomposition $(T, \beta)$ given by \Cref{thm:tree-decomp} for $H$.
    
    Define a weight function $\chi : V(G) \to \{0, 1\}$ which assigns $1$ to terminals of all pairs in $\mathcal{M}$ and assigns $0$ to all other vertices. The total weight over all vertices is then $2L$. Since we assign a positive weight only to leaf vertices, the weight of every bag is at most $3$. We will fix a parameter $\alpha$, roughly of order $L^{1/2}$, the value of which will be derived later.
    
    We use \Cref{lem:tree-decomp-split} to split $T$ in a balanced way w.r.t. $\chi$. That is, we find a set $F \subseteq E(T)$ of $\alpha$ edges such that the union of bags of each connected component of $T - F$ has weight at least $\frac{\chi(V(G))}{4\alpha} = \frac{L}{2\alpha}$. The other case, where one bag is of weight at least $\frac{1}{4\alpha}\chi(V(G))$, cannot hold if $L$ is large enough.
    
    For $f \in F$, put $S_f$ as the separator induced by the adhesion of $f$, as defined at the beginning of the subsection. Put $X = \bigcup_{f \in F} V(S_f)$. Fix any connected component $C$ of $T - F$. Look at the set $V_C$ of vertices defined as the union of bags of $C$ minus the vertices of $X$. We argue that $V_C$ contains both terminals of at least one good pair.
    
    By \Cref{clm:noncrossing-pairs-sep} and \ref{clm:noncrossing-pairs-cont}, there are at most $\alpha \cdot \Oh(\log^8 L \cdot \varepsilon^{-8})$ pairs which are not good w.r.t. to any $S_f$, hence twice as much terminal vertices belonging to any such pair. Moreover, every set $V(S_f)$ is a union of two paths on non-virtual edges, hence contains at most $2$ vertices of originally degree $1$, hence at most $2$ terminal vertices. Therefore, the total number of terminal vertices in $V_C$ which are either bad or in $X$ is at most $\alpha \cdot c \cdot \log^8 L \cdot \varepsilon^{-8}$, where $c$ is some universal constant. The number of all terminal vertices in $V_C$ is $\chi(V_C) \geq \frac{L}{2\alpha}$, hence if we set $\alpha$ such that
    $$
        4\cdot \alpha^2 \cdot c \leq L \log^{-8} L \cdot \varepsilon^8,
    $$
    then every $V_C - X$ will contain at least one terminal of a pair, say $\sigma_C$, which will be good w.r.t. all $S_f$.

    Now, consider any two components $C_1, C_2$ and their corresponding pairs $\sigma_{C_1}, \sigma_{C_2}$. There is some edge $f' \in F$ separating them in $T$, hence they lie on different sides of the loop $\xi_{f'}$. In particular, the links $\sigma_{C_1}, \sigma_{C_2}$ are separated by $\xi_{f'}$, and hence are nonintersecting.

    This way, we have picked $\alpha = \Theta(L^{1/2} \log^{-4} L \cdot \varepsilon^4)$ pairs with pairwise nonintersecting links, one per each component of $T - F$. Restricting our instance to these pairs gives us a noncrossing instance, hence $\alpha \leq L_{\asymp}(\varepsilon)$, and therefore $L = \Oh(L_{\asymp}(\varepsilon)^2 \log^8 L_{\asymp}(\varepsilon) \cdot \varepsilon^{-8})$, which finishes the proof.
 \end{proof}

\subsection{Bounding linked noncrossing instances} \label{ss:noncrossing-bound}

Let $\mathcal{I} = (G, \mathcal{M})$ be a linked non-crossing instance.
We will frequently use the following uncrossing argument.

\begin{claim}\label{clm:uncross}
    Let $\sigma_a,\sigma_b,\sigma_c,\sigma_d \in \mathcal{M}$ be
    not necessarily distinct.
    If $\ppath{p_a q_b} \cap \ppath{p_c q_d} \neq \emptyset$, 
    then 
    \[ \dist(p_a, q_b) + \dist(p_c, q_d) \geq \dist(p_a, q_d) + \dist(p_c, q_b).\]
\end{claim}
\begin{proof}
    Let $u$ be a vertex in the intersection of $\ppath{p_a q_b}$
    and $\ppath{p_c q_d}$. 
    Split both $\ppath{p_a q_b}$ and $\ppath{p_c q_d}$ at $u$
    and observe that the subpath of $\ppath{p_a q_b}$ from $p_a$ to $u$
    and the subpath of $\ppath{p_c q_d}$ from $u$ to $q_d$ form a walk
    from $p_a$ to $q_d$ while the remaining parts,
    i.e., the subpath of $\ppath{p_c q_d}$ from $p_c$ to $u$
    and the subpath of $\ppath{p_a q_b}$ from $u$ to $q_b$ form a
    walk from $p_c$ to $q_b$. The claim follows.
\end{proof}

The first immediate corollary is the following:

\begin{claim} \label{clm:disjoint-black-side}
    $\llbracket p_iq_j \rrbracket \cap  \llbracket p_jq_i \rrbracket= \emptyset, \forall \sigma_i, \sigma_j \in \mathcal{M}.$
\end{claim}

\begin{proof}

    Suppose, for contradiction, that $\llbracket p_iq_j \rrbracket$ intersects $\llbracket p_jq_i \rrbracket$. Then, 
    by~\Cref{clm:uncross},
    \[ 2(1-\varepsilon)R \geq \dist(p_i,q_j) + \dist(p_j, q_i) \geq \dist(p_i, q_i) + \dist(p_j, q_j) \geq 2R.\]
    This is a contradiction. 
\end{proof}

\subsubsection{Equidistant structure}
We now filter $\mathcal{M}$ to get the following substructure.
\begin{definition}
    An \emph{equidistant structure} in $\mathcal{I}$
    is a triple $(\sigma' = (p',q'), \sigma'' = (p'', q''), \mathcal{N})$
    where $\sigma',\sigma'' \in \mathcal{M}$, $\sigma' \neq \sigma''$,
    $\mathcal{N} \subseteq \mathcal{M} \setminus \{\sigma', \sigma''\}$
    and the following holds:
    \begin{enumerate}
    \item For every $\sigma=(p,q) \in \mathcal{N}$, 
    \[ \Udist(p, q') = \Udist(p'', q') \qquad \mathrm{and} \qquad \Udist(q, p') = \Udist(q'', p'). \]
    \item For every $\sigma_1=(p_1,q_1), \sigma_2 = (p_2,q_2) \in \mathcal{N}$, 
    \[ \Udist(p_1, q'') = \Udist(p_2, q'') \qquad \mathrm{and} \qquad \Udist(q_1, p'') = \Udist(q_2, p''). \]
    \item For every distinct $\sigma_1=(p_1,q_1), \sigma_2 = (p_2,q_2) \in \mathcal{N}$,
    the path $\ppath{p_1 q_2}$ does not intersect
    $\ppath{p' q''}$ nor $\ppath{p'' q'}$.
    \end{enumerate}
    The \emph{size} of the equidistant structure $(\sigma', \sigma'', \mathcal{N})$ is $|\mathcal{N}|$.
\end{definition}

\begin{lemma}\label{lem:find-equidistant}
    Assume $\mathcal{I}$ does not contain an equidistant structure
    of size larger than $L$. Then, $|\mathcal{M}| = \Oh(L \varepsilon^{-12}).$
\end{lemma}
\begin{proof}
Pick arbitrary $\sigma' = (p',q') \in \mathcal{M}$.
Tag each $\sigma=(p,q) \in \mathcal{M} \setminus \{\sigma'\}$
with the following tuple: $(\Udist(p, q'), \Udist(q, p'))$. 
There are $\Oh(\varepsilon^{-2})$ possible tags.
Let $\mathcal{M}_1 \subseteq \mathcal{M} \setminus \{\sigma'\}$ 
be a subset of pairs with the same tag.

Pick arbitrary $\sigma'' = (p'', q'') \in \mathcal{M}_1$ 
and re-tag each pair $\sigma = (p,q) \in \mathcal{M}_1 \setminus \{\sigma''\}$ 
with the new tuple $(\Udist(p, q''), \Udist(q, p''))$.
Again, there are $\Oh(\varepsilon^{-2})$ possible tags
and 
let $\mathcal{M}_2 \subseteq \mathcal{M}_1 \setminus \{\sigma''\}$ 
be a subset of pairs with the same tag.

Let $Z$ be a set of $\Oh(\varepsilon^{-1})$ vertices such that
every vertex of $\ppath{p' q''} \cup \ppath{p'' q'}$ is within distance
$\delta$ from a vertex of $Z$. 
By~\Cref{lem:Uprof-vc}, there are $\Oh(\varepsilon^{-8})$ possible
rounded distance profiles to $Z$.
By~\Cref{lem:sep-two-pairs}, if $\sigma_1 = (p_1,q_1), \sigma_2=(p_2,q_2) \in \mathcal{M}_2$ are such that the rounded distance profiles
of $p_1$ and $p_2$ to $Z$ are equal, then $\ppath{p_1 q_2}$ does not
intersect $\ppath{p' q''}$ nor $\ppath{p'' q'}$. 
Consequently, if $\mathcal{N} \subseteq \mathcal{M}_2$ is a set
of pairs $(p,q)$ with the same rounded distance profile of $p$ to $Z$,
then $(\sigma', \sigma'', \mathcal{N})$ is an equidistant structure.
The lemma follows.
\end{proof}

Fix an equidistant structure $(\sigma' = (p',q'), \sigma''=(p'',q''), \mathcal{N})$ in $\mathcal{I}$. Our goal is to bound $L = |\mathcal{N}|$ by $\Oh(\varepsilon^{-2})$.

\subsubsection{Forbidden intersections}

We now make a  number of observations about disjointness of some paths in 
the equidistant structure $(\sigma', \sigma'', \mathcal{N})$. Note that $\ppath{p'q''} \cap \ppath{q'p''} = \emptyset$ based on \Cref{clm:disjoint-black-side}.

\begin{claim}\label{clm:direct1}
  For every $\sigma = (p,q) \in \mathcal{N}$, 
  we have $\ppath{q'' p} \cap \ppath{q' p''} = \emptyset$
  and $\ppath{p'' q} \cap \ppath{p' q''} = \emptyset$.
\end{claim}
\begin{proof}
    By contradiction, assume that $\ppath{q'' p}$ and $\ppath{q' p''}$ intersect.
    By~\Cref{clm:uncross},
    \[ (1-\varepsilon)R + \dist(q', p'') \geq \dist(q'', p) + \dist(q', p'') \geq \dist(q'', p'') + \dist(q', p) \geq R + \dist(q', p). \]
    However, from the definition of an equidistant structure,
    $\Udist(q', p'') = \Udist(q', p)$, which implies
    $|\dist(q', p'') - \dist(q', p)| < \delta$, a contradiction.

    The proof of $\ppath{p'' q} \cap \ppath{p' q''} = \emptyset$
    is symmetrical.
\end{proof}

~\Cref{clm:direct1} motivates the following definition.
Let $\hat{\gamma}'$ be a minimal part of the link of $\sigma'$ that connects
$\ppath{p' q''}$ with $\ppath{q' p''}$
and let $\gamma'$ be a curve from $q''$ to $p''$ consisting
of the subpath of $\ppath{p' q''}$ from $q''$ to the endpoint of $\hat{\gamma}'$,
the curve $\hat{\gamma}'$, and the subpath of $\ppath{q' p''}$ from the endpoint
of $\hat{\gamma}'$ to $p''$. As for $(p, q) \in \mathcal{N}$ we have that $\ppath{q''p}$ intersects with $\ppath{q''p'}$ only on its prefix, does not intersect with $\hat{\gamma}'$ (as the link of $\sigma'$ does not intersect $G^{\neg \sigma'}$, which contains $\ppath{q''p}$) and does not intersect $\ppath{q'p''}$ (per \Cref{clm:direct1}), we conclude that the only intersection of $\ppath{q''p}$ with $\gamma'$ is the common prefix of $\ppath{q''p}$ and $\ppath{q''p'}$. Similarly, the only intersection of $\ppath{p''q}$ with $\gamma'$ is the common prefix of $\ppath{p''q}$ and $\ppath{p''q'}$.
Furthermore, the definition of an equidistant structure ensures that
for every distinct $\sigma_1 = (p_1,q_1), \sigma_2 = (p_2,q_2) \in \mathcal{N}$, 
the path $\ppath{p_1 q_2}$ is disjoint with $\gamma'$. 

Intuitively, in what follows, one should think of $\gamma'$ as a ``rupture''
in the plane that connects $q''$ and $p''$. We note that we defined $\gamma'$ this way, instead of possibly more natural connection consisting of $\ppath{q''p'}$, the link of $\sigma'$ and $\ppath{q'p''}$, because the latter one might have self-intersections, while $\gamma'$ cannot.
We will argue about relative positions of different
objects with respect to $\gamma'$.

Before we continue, we make one more observation.
\begin{claim}\label{clm:direct2}
  For every distinct $\sigma_1 = (p_1,q_1), \sigma_2 = (p_2,q_2) \in \mathcal{N}$, 
    the path $\ppath{p_1 q_2}$ does not intersect 
    $\ppath{q'' p_2}$ nor $\ppath{p'' q_1}$.
\end{claim}
\begin{proof}
    By contradiction, assume that $\ppath{p_1 q_2}$ and $\ppath{q'' p_2}$ intersect. By~\Cref{clm:uncross},
    \[ (1-\varepsilon)R + \dist(q'', p_2) \geq \dist(p_1, q_2) + \dist(q'', p_2) \geq \dist(q'', p_1) + \dist(p_2, q_2) \geq \dist(q'', p_1) + R. \]
    However, from the definition of an equidistant structure,
    $\Udist(q'', p_1) = \Udist(q'', p_2)$, which implies
    $|\dist(q'', p_1) - \dist(q'', p_2)| < \delta$, a contradiction.

    The proof for $\ppath{p'' q_1}$ is symmetrical.
\end{proof}

\subsubsection{Order around \texorpdfstring{$p''$}{p''} and \texorpdfstring{$q''$}{q''}}

At this stage, we now have a sufficient toolbox to reason about the order on the pairs of $\mathcal{N}$. 
Consider the shortest path tree with the root $q''$ and leaves
$\{p~|~(p,q) \in \mathcal{N}\} \cup \{p'\}$. Since every terminal
has only one incident edge of finite weight (i.e., not a virtual one),
there is a well-defined clockwise cyclic order of the leaves
$\{p~|~(p,q) \in \mathcal{N}\} \cup \{p'\}$. Let us transform this cyclic order into a~linear order by declaring $p'$ as its first element, and then delete $p'$ 
from this order, obtaining as a~result a~linear clockwise ordering of 
$\{p~|~(p,q) \in \mathcal{N}\}$; we denote the induced
ordering of $\mathcal{N}$ as $\pi_{q''}$.
We define the order $\pi_{q''}$ analogously.
The following claim unifies the two orders: 

\begin{claim} \label{clm:one-order}
    $\pi_{p''} = \pi_{q''}$. 
\end{claim}
\begin{proof}
    Assume, contrary to our claim that $\pi_{p''} \neq \pi_{q''}$. 
    Let $\sigma_1 = (p_1,q_1), \sigma_2=(p_2,q_2) \in \mathcal{N}$
    be such that $\sigma_1$ is before $\sigma_2$ in $\pi_{p''}$
    but $\sigma_1$ is after $\sigma_2$ in $\pi_{q''}$. 

    We consider a closed curve consisting of paths:
    \[ \ppath{q'' p_1}, \ppath{p_1 q_2}, \ppath{q_2 p''},\]
    
    and the curve $\gamma'$. 
    Since $\mathcal{I}$ is a linked noncrossing instance, 
    the link of $\sigma'$ does not intersect 
    the paths $\ppath{q'' p_1}, \ppath{p_1 q_2}, \ppath{q_2 p''}$.
    The properties of the equidistant structure ensure that $\ppath{p_1 q_2}$
    does not intersect $\ppath{p' q''} \cup \ppath{p'' q'}$. 
    We infer that
    the paths $\ppath{q'' p_1}, \ppath{p_1 q_2}, \ppath{q_2 p''}$, 
    and the curve $\gamma'$
    form a closed curve without self-intersections, 
    which we call $\gamma$.
    
    Observe that the path $\ppath{q'' p_2}$ does not intersect
    $\ppath{p_1 q_2}$ (\Cref{clm:direct2}),
    nor $\ppath{q_2 p''}$ (\Cref{clm:disjoint-black-side}),
    nor $\ppath{p'' q'}$ (\Cref{clm:direct1}),
    nor the link of $\sigma'$. Furthermore,
    it may share a prefix starting at $q''$ with $\ppath{q'' p'}$
    and $\ppath{q'' p_1}$, but then departs so that the clockwise
    order in the shortest path tree rooted at $q''$ is $p_2$, $p_1$, $p'$.

    Symmetrically, the path $\ppath{p'' q_1}$ does not intersect
    $\ppath{p_1 q_2}$ (\Cref{clm:direct2}),
    nor $\ppath{p_1 q''}$ (\Cref{clm:disjoint-black-side}),
    nor $\ppath{q'' p'}$ (\Cref{clm:direct1}),
    nor the link of $\sigma'$. Furthermore,
    it may share a prefix starting at $p''$ with $\ppath{p'' q'}$
    and $\ppath{p'' q_2}$, but then departs so that the clockwise
    order in the shortest path tree rooted at $p''$ is $q'$, $q_1$, $q_2$.

    This implies that $q_1$ and $p_2$ lie on different sides
    of the closed curve without self-intersections $\gamma$.

    Consider now the path $\ppath{q_1 p_2}$. 
    It does not intersect $\ppath{p_1 q_2}$ (\Cref{clm:disjoint-black-side}),
    nor $\ppath{p'' q_2}$ nor $\ppath{q'' p_1}$ (\Cref{clm:direct2}),
    nor $\ppath{p' q''}$ nor $\ppath{q' p''}$ (equidistant structure)
    nor the link of $\sigma'$ (linked noncrossing instance $\mathcal{I}$).
    Hence, it cannot intersect $\gamma$.
    This is a contradiction.
\end{proof}

~\Cref{clm:one-order} allows us to enumerate
$\mathcal{N}$ as $\sigma_1,\sigma_2,\ldots,\sigma_L$ according
to $\pi_{p''} = \pi_{q''}$. For brevity, let $\sigma_i = (p_i,q_i)$ for $1 \leq i \leq L$.

\subsubsection{A box, Type A and Type B pairs}
Let $1 \leq i \leq L$. Let $\hat{\gamma}_i$ be a minimal part
of the link of $\sigma_i$ that connects $\ppath{p_i q''}$ with $\ppath{q_i p''}$.
Let $\gamma_i$ be a curve consisting of a subpath of $\ppath{q'' p_i}$
from $q''$ to the endpoint of $\hat{\gamma}_i$, the curve $\hat{\gamma}_i$,
and a subpath of $\ppath{p'' q_i}$ from the endpoint of $\hat{\gamma}_i$
to $p''$. Finally, let $\gamma_i^\circ$ be the concatenation
of $\gamma'$ and $\gamma_i$.
By~\Cref{clm:direct1}, \Cref{clm:disjoint-black-side} and the assumption that our instance is a linked noncrossing one, we deduce that $\gamma_i^\circ$ is a closed curve
without self-intersections (only parts of $\ppath{q'' p'}$ and $\ppath{q'' p_i}$
may share a prefix and $\ppath{p'' q'}$ and $\ppath{p'' q_i}$ may share a prefix). We call $\gamma_i^\circ$ the \emph{box} of $i$. 

As $\gamma_i^\circ$ is a closed curve without self-intersections,
it splits the plane into two regions. 
Without loss of generality, 
we can fix the face immediately clockwise after the first edge
of $\ppath{q'' p'}$, looking from $q''$, as the outerface of the embedding;
this defines the \emph{inside} and \emph{outside} of $\gamma_i^\circ$.
(This is only so that the names ``inside'' and ``outside'' are natural;
 we could as well go with ``the same side as the said face'' and ``the other side
 as the said face''.)

Note that for every $j \neq i$, 
the path $\ppath{q'' p_j}$ needs to leave $\ppath{q'' p_i}$ before
the endpoint of $\hat{\gamma}_i$, because our instance is linked noncrossing. Note that the
path $\ppath{q'' p_j}$ leaves $\ppath{q'' p_i}$ towards the inside of $\gamma_i^\circ$ for $j > i$ and towards the outside for $j < i$.
Similarly, the path $\ppath{p'' q_j}$ leaves $\ppath{p'' q_i}$ 
to the inside of $\gamma_i^\circ$ for $j < i$ and to the outside for $j > i$.

\begin{figure}[ht]
    \centering
    \begin{subfigure}{0.35\textwidth}
        \centering
        \includegraphics[width=1\linewidth]{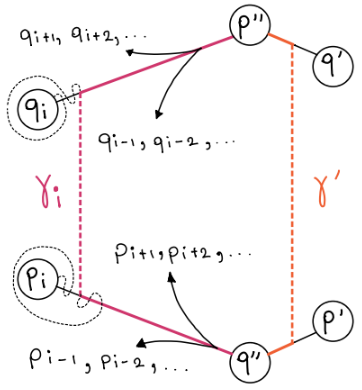}
        \caption{The box of $\sigma_i$}
        \label{fig:box}
    \end{subfigure}
    \begin{subfigure}{0.31\textwidth}
        \centering
        \includegraphics[width=1\linewidth]{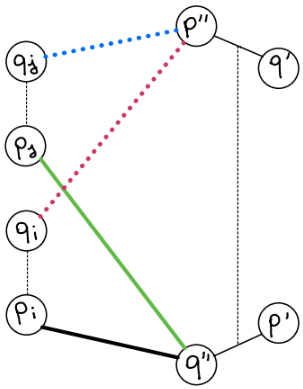}
        \caption{Type A Relationship}
        \label{fig:type_A}
    \end{subfigure}
    \begin{subfigure}{0.26\textwidth}
        \centering
        \includegraphics[width=1\linewidth]{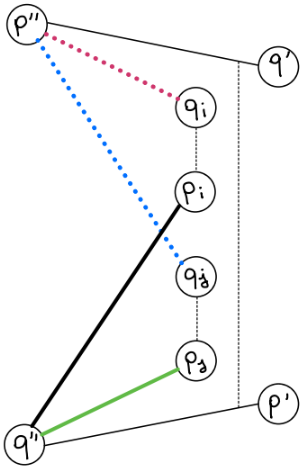}
        \caption{Type B Relationship}
        \label{fig:type_B}
    \end{subfigure}
    \caption{The dotted lines represent links. (a) $\gamma_i$ is the pink right-open trapeze shape, $\gamma^\prime$ is the orange left-open trapeze shape, and $\gamma_i^\circ$ is the concatenation of the two. $\gamma_i$ is drawn curly about $p_i, q_i$ to highlight the reason for taking the \textit{minimal} crossing component for the box, instead of the whole link. Shortest paths from $p^{\prime\prime}$ to later $q$s leave $\gamma_i^\circ$ to the outside, and shortest paths to earlier $q$s leave to the inside. Likewise, shortest paths from $q^{\prime\prime}$ to later $p$s, earlier $p$s leave $\gamma_i^\circ$ resp. to the inside, outside. (b) $\sigma_i, \sigma_j$ are two pairs with $i <_A j$. Note that $\ppath{q'' p_j}$ (green) intersects $\ppath{p'' q_i}$ (pink dotted) transversally, and $\ppath{p'' q_j}$ (blue dotted) does not intersect $\ppath{q'' p_i}$ (black). (c) Here $i <_B j$ instead. Symmetrically, $\ppath{p'' q_j}$ (blue) intersects $\ppath{q'' p_i}$ (black) transversally, and $\ppath{q'' p_j}$ (green) does not intersect $\ppath{p'' q_i}$ (pink).}
    \label{fig:box-typeA-typeB}
\end{figure}

We observe that the link of $\sigma_j$ cannot intersect the box $\gamma_i^\circ$ as it does not intersect links of $\sigma'$ and $\sigma_i$ (as our instance is linked noncrossing) and does not intersect $\ppath{p'q''}, \ppath{q''p_i}, \ppath{q_ip''}$ and $\ppath{p''q'}$ as all these paths are contained within $G^{\neg \sigma_j}$. This leads to the conclusion that there are two ways two pairs $\sigma_i$, $\sigma_j$
can be drawn for $i < j$, which we denote Type A and Type B;
see~\Cref{fig:box-typeA-typeB}.

\begin{enumerate}
    \item If the link of $\sigma_j$ lies \emph{outside}
    of the box $\gamma_i^\circ$, the relation $i < j$ is of Type A,
    which we denote $i <_A j$.

    \item If the link of $\sigma_j$ lies \emph{inside}
    of the box $\gamma_i^\circ$, the relation $i < j$ is of Type B,
    which we denote $i <_B j$.
\end{enumerate}

We make the following observation:
\begin{claim}\label{clm:type2cross}
Let $1 \leq i < j \leq L$.
If $i <_A j$, then $\ppath{q'' p_j}$ intersects $\ppath{p'' q_i}$ transversally,
while $\ppath{p'' q_j}$ does not intersect $\ppath{q'' p_i}$ transversally
(i.e., they touch or are disjoint).
Symmetrically, 
if $i <_B j$, then $\ppath{p'' q_j}$ intersects $\ppath{q'' p_i}$ transversally,
while $\ppath{q'' p_j}$ does not intersect $\ppath{p'' q_i}$ transversally. 
\end{claim}
\begin{proof}
    Assume $i <_A j$, that is, the link of $\sigma_j$ is outside
    $\gamma_i^\circ$.
    The path $\ppath{q'' p_j}$ leaves $\ppath{q'' p_i}$ to the inside
    of $\gamma_i^\circ$, does not intersect transversally 
    $\ppath{q'' p'}$, does not intersect the link of $\sigma'$
    nor the link of $\sigma_i$, and does not intersect $\ppath{p'' q'}$
    (\Cref{clm:direct1}). 
    Hence, to reach $p_j$ outside of $\gamma_i^\circ$, it needs
    to intersect $\ppath{p'' q_i}$ transversally.
    Similarly, $\ppath{p'' q_j}$ leaves $\ppath{p'' q_i}$ to the
    outside of $\gamma_i^\circ$, does not intersect transversally
    $\ppath{p'' q'}$, and does not intersect the links of $\sigma'$ nor $\sigma_i$ nor $\ppath{q'' p'}$ (\Cref{clm:direct1}). 
    Hence, to reach $q_j$ outside $\gamma_i^\circ$,
    the path $\ppath{p'' q_j}$ cannot cross $\ppath{q'' p_i}$ transversally.

    The proof for $i <_B j$ is symmetrical.
\end{proof}

\subsubsection{Type A and Type B chains}

We now show that Types A and B do not mix well.
\begin{claim} \label{clm:transitivity}
    If $1 \leq i < j < k \leq L$, then:
    
    \begin{enumerate}
        \item $i <_{A} j <_{A} k \implies i <_{A} k$,
        \item $i <_{B} j <_{B} k \implies i <_{B} k$. 
    \end{enumerate}
    
\end{claim}

\begin{proof}
    We start with the first point.
    For contradiction, suppose that $i <_A j <_{A} k$, but \textbf{not} $i <_{A} k$. 

    The path $\ppath{q'' p_k}$ leaves $\ppath{q'' p_j}$ to the inside
    of $\gamma_j^\circ$ and leaves $\ppath{q'' p_i}$ to the inside of $\gamma_i^\circ$. 
    By~\Cref{clm:type2cross}:
    \begin{itemize}
        \item the path $\ppath{p''q_j}$ does not contain any point inside $\gamma_i^\circ$ (as $i <_A j$), 
        \item the path $\ppath{q'' p_k}$ does not contain any point
        outside $\gamma_i^\circ$ (as $i <_B k$),
        \item the path $\ppath{q'' p_k}$ needs to intersect 
         $\ppath{p'' q_j}$ transversally (as $j <_A k$). 
    \end{itemize}
    This is a contradiction.

    The second point is symmetric, but let us outline it for completeness.
    For contradiction, suppose that $i <_B j <_{B} k$, but \textbf{not} $i <_{B} k$.
    The path $\ppath{p'' q_k}$ leaves $\ppath{p'' q_j}$ to the outside
    of $\gamma_j^\circ$ and leaves $\ppath{p'' q_i}$ to the outside of $\gamma_i^\circ$. 
    By~\Cref{clm:type2cross}:
    \begin{itemize}
        \item the path $\ppath{q''p_j}$ does not contain any point outside $\gamma_i^\circ$ (as $i <_B j$), 
        \item the path $\ppath{p'' q_k}$ does not contain any point
        inside $\gamma_i^\circ$ (as $i <_A k$),
        \item the path $\ppath{p'' q_k}$ needs to intersect 
         $\ppath{q'' p_j}$ transversally (as $j <_B k$). 
    \end{itemize}
    This is again a contradiction.
\end{proof}

Let $\mathcal{S} = \{(i,i)~|~1 \leq i \leq L\} \cup \{(i, j) ~|~ 1 \leq i < j\leq L \wedge i <_A j\}$. (Type A is chosen arbitrarily here, Type B would work equally well.) 
~\Cref{clm:transitivity} implies the following.
\begin{claim}
    $\mathcal{S}$ is a partial order on $\{1,2,\ldots,L\}$. 
\end{claim}

\begin{proof}

    Clearly $\leq_{A}$ is a homogeneous relation. Then what remains is to verify the three axioms: reflexivity, anti-symmetry, and transitivity.

    \begin{itemize}
        \item \textbf{Reflexivity:} Immediate
        as we explicitly added $\{(i,i)~|~1 \leq i \leq L\}$ to $\mathcal{S}$.

        \item \textbf{Anti-symmetry:} 
        Immediate as $(i,j) \in \mathcal{S}$ implies $i \leq j$.
        
        \item \textbf{Transitivity:} This axiom is equivalent to~\Cref{clm:transitivity}. (In this case specifically \textit{$1$.}, and \textit{$2$.} for the analogous type B argument.) 
    \end{itemize}
\end{proof}

The remainder of the argument of this subsection relies on Dilworth’s theorem, a fundamental result in the theory of partially ordered sets, presented below.
A set $I \subseteq \{1,\ldots,L\}$ is \emph{type-homogeneous} if
either for every $i,j \in I$ with $i<j$, we have $i <_A j$,
    or for every $i,j \in I$ with $i <j$, we have $i <_B j$. 

\begin{theorem} [Dilworth's Theorem]
     In every finite partially ordered set, the size of a maximal antichain equals the size of a minimum chain cover.
\end{theorem}

\begin{claim}\label{clm:dilworth}
    If $|\mathcal{N}| = L \geq 1 + (K-1)^2$ for an integer $K$,
    then there exists a set $I \subseteq \{1,\ldots,L\}$ of size $K$
    that is type-homogeneous.
\end{claim}
\begin{proof}
    Consider the partially ordered set $(\{1,\ldots,L\}, \mathcal{S})$.
    If it contains a chain of length at least $K$, then any $K$ elements
    of this chain form a set $I$ such that for every $i,j \in I$ with
    $i < j$ we have $i <_A j$. 
    Otherwise, any minimum chain cover needs to contain at least
    $\lceil L / (K-1) \rceil \geq K$ chains.
    By Dilworth's theorem, there exists an antichain of size at least $K$.
    Then, any $K$ elements of this antichain form a set $I$ such that
    for every $i,j \in I$ with $i < j$, we have $i <_B j$.
\end{proof}

\subsubsection{Final contradiction}

\begin{figure}[ht]
    \centering
    \begin{subfigure}{0.4\textwidth}
        \centering
        \includegraphics[width=1\linewidth]{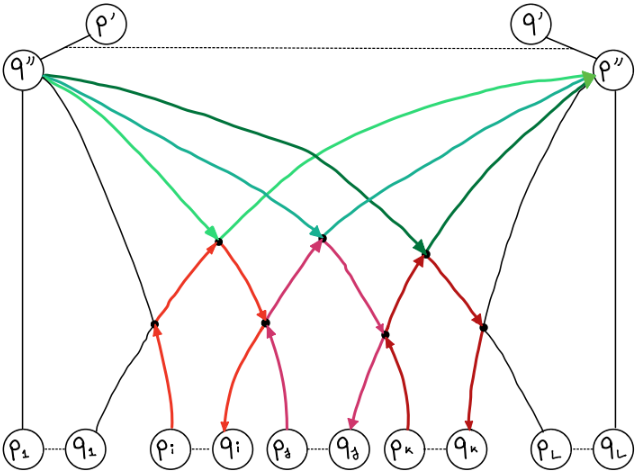}
        \caption{Intersections of Paths $\ppath{q'' p_i}$ and $\ppath{p'' q_i}$}
        \label{fig:final-picture}
    \end{subfigure}
    \begin{subfigure}{0.4\textwidth}
        \centering
        \includegraphics[width=1\linewidth]{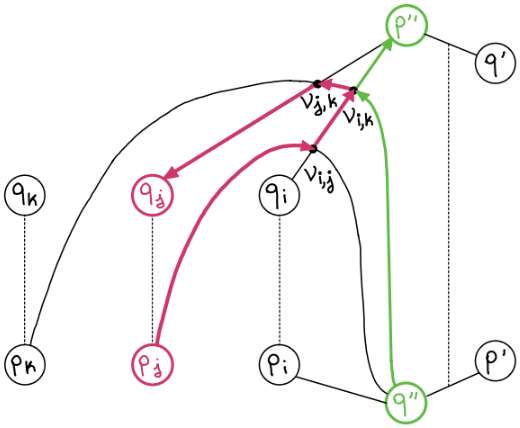}
        \caption{Key Construction for~\Cref{clm:final-punch}}
        \label{fig:final-construction}
    \end{subfigure}
    \caption{$i <_A j <_A k$ are three consecutive indices of $I$. (a) The paths $\ppath{q'' p_i}$, $\ppath{p'' q_i}$ form a web like structure which can mostly be covered by long paths (of length at least R), shown here as greens and reds. The arrows are explained by: (b) The long paths $P_j$ (pink) connecting $p_j$ and $q_j$, and $Q_j$ (green) connecting $p''$ and $q''$, are constructed in~\Cref{clm:final-punch} following the direction of the arrows. It's clear from the picture that the intersections $v_{i, j}, v_{i, k}$, and $v_{j, k}$ can be equal for consecutive indices without breaking either $P_j$ or $Q_j$.}   
    \label{fig:final-contradiction-drawn}
\end{figure}

The time has come to show the final contradiction, thereby completing the proof of~\Cref{lem:noncrossing-bound}. 
In the following claim we bound the size of a type-homogeneous set.
The main intuition is as follows:~\Cref{clm:type2cross}   implies
that the intersections of paths $\ppath{q'' p_i}$ and $\ppath{p'' q_i}$
look as in~\Cref{fig:final-contradiction-drawn} (for Type A;
the case of Type B is symmetric). 
We show how to cut $2|I|$ paths of length at most $(1-\varepsilon)R$
each and arrange them into $2|I|-4$ paths that connect terminals
of the same pair; thus they are of length at least $R$. This
implies $|I| = \Oh(\varepsilon^{-1})$, as desired.

\begin{claim}\label{clm:final-punch}
Let $I \subseteq \{1,\ldots,L\}$ be a type-homogeneous set.
Then, $|I| = \Oh(\varepsilon^{-1})$.
\end{claim}
\begin{proof}
    We consider the case of Type A, that is, for every $i, j \in I$ with
    $i < j$ it holds that $i <_A j$. the case of Type B is symmetric
    (with swapping the roles of terminals $p$ with terminals $q$).

    By~\Cref{clm:type2cross}, for every $i,j \in I$ with $i < j$,
    the paths $\ppath{q'' p_j}$ and $\ppath{p'' q_i}$ intersect transversally.
    Let $v_{i,j}$ be any vertex of this intersection.

    Pick arbitrary $i, j, k \in I$ with $i < j < k$. 
    Recall that,
    by~\Cref{clm:type2cross}, $\ppath{q'' p_j}$ and $\ppath{q'' p_k}$
    both intersect $\ppath{p'' q_i}$ transversally, 
    and also $\ppath{q'' p_k}$ intersects $\ppath{p'' q_j}$ transversally.
    
    As the clockwise order of the leaves of the shortest path tree rooted
    at $q''$ is $p_i$, $p_j$, $p_k$ (and then the curve $\gamma'$),
    on the path $\ppath{p'' q_i}$ 
    the intersection with $\ppath{q'' p_j}$ lies
    closer to $q_i$ than the intersection with $\ppath{q'' p_k}$
    (these intersections may be the same single vertex). 
    Consequently, on $\ppath{p'' q_i}$
    the vertices $v_{i,j}$ for $j > i$, $j \in I$
    lie in the increasing order of the index $j$
    if we traverse the path from $p_i$ to $q''$ (the vertices $v_{i,j}$ may be equal for consecutive indices $j \in I$). 

    Similarly, as the clockwise order of the leaves of the shortest
    path tree rooted at $p''$ is $q_i$, $q_j$, $q_k$ (and then the curve $\gamma'$), on the path $\ppath{q'' p_k}$
    the intersection with $\ppath{p'' q_j}$ lies closer to $p_k$ than
    the intersection with $\ppath{p'' q_i}$
    (these intersections may be the same single vertex). 
    Consequently, on $\ppath{q'' p_k}$ 
    the vertices $v_{i,k}$ for $i < k$, $i \in I$
    lie in the increasing order of the index $i$
    if we traverse the path from $q''$ to $p_k$ (the vertices $v_{i,k}$ may be equal for consecutive indices $i \in I$). 

    We now perform the rearrangement.
    Let $i < j < k$ be three consecutive indices of $I$. 
    We construct a path $P_j$ connecting $p_j$ and $q_j$ as follows:
    start in $p_j$, continue along $\ppath{p_j q''}$ to $v_{i,j}$
    then along $\ppath{q_i p''}$ in the direction of $p''$ to $v_{i,k}$, then 
    along $\ppath{q'' p_k}$ in the direction of $p_k$ to $v_{j,k}$, 
    and finally along $\ppath{p'' q_j}$ to $q_j$.
    We construct a path $Q_j$ connecting $p''$ and $q''$ as follows:
    start in $q''$, continue along $\ppath{q'' p_k}$ to $v_{i,k}$,
    and then along $\ppath{p'' q_i}$ to $p''$.
    There are $|I|-2$ paths $P_j$ and $|I|-2$ paths $Q_j$, and each
    path $P_j$ or $Q_j$ is of length at least $R$.
    
    In the rearrangement we only used edges
    of paths $\ppath{p'' q_i}$ and $\ppath{q'' p_i}$ for $i \in I$
    and we do not use the same edge twice. Each of these paths
    is of length at most $(1-\varepsilon)R$. Hence,
    \[ 2|I| \cdot (1-\varepsilon)R \geq (2|I|-4) \cdot R.\]
    This implies $|I| \leq 2\varepsilon^{-1}$, as desired.
\end{proof}
This finishes the proof of~\Cref{lem:noncrossing-bound}.

%% file: kcentre.tex
\section{Coreset for $k$-Center}\label{sec:kcenter}

In this section we build on the ideas developed so far to show a coreset for \textsc{$k$-Center}. We first show a reduction from $(k, \eps)$-Comatching to $k$-Center $\eps$-Coreset, by mimicking the proof of~\Cref{thm:comatching2coreset} shown in ~\Cref{sec:comatchingtocoreset}. However, $(k, \eps)$-Comatching is still not a sufficiently structured object to analyze directly. To address this, we apply a Ramsey-like argument which extracts either a comatching or double ladder from a large $(k, \eps)$-Comatching. Since we already analyzed the size of comatching, we dedicate~\Cref{sec:doubleladder} to bound the size of double ladder. 

We now recall the notion of $k$-Center $\eps$-coreset, stated here in a slightly more convenient form. This definition is equivalent to~\Cref{def:coreset-kcenter}, since $(1-\eps)^{-1} \approx 1 + \eps$, for small $\eps>0$.
\begin{definition}[$k$-Center $\varepsilon$-coreset]
    A $k$-Center $\varepsilon$-coreset in a metric space $(V, \dist)$ is a subset $C$ of $V$ such that for every set $X$ of at most $k$ potential centers, it holds that
    \[\max_{v\in V}\dist(v, X) \leq (1+\varepsilon) \max_{v\in C}\dist(v, X)\]
\end{definition}

\subsection{$k$-Comatching to $k$-center \texorpdfstring{$\varepsilon$}{e}-coreset}

In this section, we prove \Cref{thm:kcenter-comatching} restated below.

\KCenterComatching*

We start by mimicking the proof of~\Cref{thm:comatching2coreset}
to link the size of a coreset with an appropriately adapted notion of 
a comatching. 
\begin{definition}[$(k, \varepsilon)$-comatching]
    For a real $\varepsilon > 0$, a $(k, \varepsilon)$-comatching in a metric space $(V, \dist)$ is a set $\mathcal{M}\subseteq V \times V^k$ such that there exists a  real $R>0$ such that:
    \begin{enumerate}
        \item for every $(p,X)\in \mathcal{M}$, we have
    $\dist(p, X)> R$, and
        \item for every distinct $(p_1, X_1)$, $(p_2, X_2)\in \mathcal{M}$, we have $\dist(p_1, X_2)\leq (1-\varepsilon)R$ and $\dist(p_2, X_1)\leq (1 - \varepsilon)R$.
    \end{enumerate}
\end{definition}

\begin{proof}[Proof of~\Cref{thm:kcenter-comatching}]
This proof follows very closely the proof of~\Cref{thm:comatching2coreset}.
As we are now studying coresets for $k$-center, some individual vertices become
sets of at most vertices; it will be more convenient to think of them as 
ordered tuples of $k$ vertices, i.e., elements of $V^k$.

Fix $0 < \varepsilon < 1$ and positive integers $k$, $d$.
Let $(V,\dist)$ be a metric space 
such that the VC-dimension of the ball set system of $(V,\dist)$
is at most $d$.
Let $P \subseteq V$.
Without loss of generality, assume $|P| > k + 1$, as otherwise we can just take
the entire $P$ as a coreset.

For every tuple $\alpha$ of vertices of $V$, 
pick $z_\alpha \in P$ 
to be a vertex maximizing the distance
from $\alpha$, breaking ties arbitrarily. 

We start by taking care of the analog of the set $V_\mathrm{far}$ from
the proof of~\Cref{thm:comatching2coreset}, that is, tuples
of $k$ vertices that are terrible $k$-centers.
Construct a tuple $\alpha^0 = (v_1^0,v_2^0,\ldots,v_{k+1}^0)$ as follows.
Let $v_1^0 \in P$ be arbitrary. For $i=2,3,\ldots,k+1$, let
$v_i^0 = z_{(v_1^0, \ldots, v_{i-1}^0)}$. 
Let $\Delta = \dist(\alpha^0, z_{\alpha^0})$. Note that $\Delta > 0$ as $|P| > k+1$.

By the construction of $\alpha^0$, 
for every $1 \leq a < b \leq k+1$, $\dist(v_a^0, v_b^0) \geq \Delta$.
Consequently, every ball of radius less than $\Delta/2$ can contain
at most one vertex of $\alpha^0$, so for every $\alpha \in V^k$ we have $\dist(\alpha, z_\alpha) \geq \Delta/2$. 
On the other hand, for every $v \in P$ there exists $i \in [k+1]$ with
$\dist(v, v_i^0) \leq \Delta$, so for every $\alpha \in V^k$
we have 
\[ \max_{v \in P} \dist(v, \alpha) \leq \Delta + \max_{i \in [k+1]} \dist(v_{i}^0, \alpha). \]
Let 
\[ \mathcal{V} = \{\alpha \in V^k~|~\dist(\alpha, z_\alpha) \leq \varepsilon^{-1}\Delta \}. \]
We infer that for every $\alpha \in V^k \setminus \mathcal{V}$, 
\[ \max_{i \in [k+1]} \dist(v_{i}^0, \alpha) \geq \max_{v \in P} \dist(v,\alpha) - \Delta \geq (1-\varepsilon) \max_{v \in P} \dist(v,\alpha). \]
Hence, $\{v_i^0~|~i \in [k+1]\}$ is a $k$-center $\varepsilon$-coreset
of size $k+1$ for tuples not in $\mathcal{V}$.

In the remainder of the proof we focus on tuples of $\mathcal{V}$.
Let
\[ \delta := \frac{\varepsilon}{4} \cdot \Delta. \]

We partition $\mathcal{V}$ into sets $\mathcal{V}_i$ for nonnegative integers
$i$; a tuple $\alpha$ belongs to $\mathcal{V}_i$ if
\[ i \delta \leq \dist(\alpha, z_\alpha) < (i+1) \delta. \]
Note that, as $\Delta/2 \leq \dist(\alpha,z_\alpha) \leq \varepsilon^{-1} \Delta$
for $\alpha \in \mathcal{V}$, the set $\mathcal{V}_i$ is nonempty
only for $ \lfloor 2 \varepsilon^{-1} \rfloor \leq i \leq \lfloor 4 \varepsilon^{-2} \rfloor$. 

Fix $i$ for which $\mathcal{V}_i \neq \emptyset$.
For $\alpha \in \mathcal{V}_i$ define
\[ S_i^\alpha = \{u \in P \mid \dist(\alpha,u) \geq (i-1) \delta \}. \]

Clearly, $z_\alpha \in S_i^\alpha$.
We consider set system $\mathcal{A}_i = \{S_i^\alpha \mid \alpha \in \mathcal{V}_i\}$
of subsets of $V$.
First, we observe that it suffices to find a hitting set of $\mathcal{A}_i$.
\begin{claim}\label{cl:kcenter-hs2coreset}
  Any hitting set of $\mathcal{A}_i$ is also a $k$-center $\varepsilon$-coreset for furthest neighbor of $\mathcal{V}_i$.
\end{claim}
\begin{proof}
The proof is identical to the proof of~\Cref{cl:hs2coreset}.
We copy-paste it below for completeness.

  Let $X$ be a hitting set of $\mathcal{A}_i$.
  Let $\alpha \in \mathcal{V}_i$ and let $x \in X \cap S_i^\alpha$, which exists as $X$ is a hitting set. 
  Then, by the definition of $S_i^\alpha$, 
  \[ \dist(\alpha, x) \geq (i-1) \delta. \]
  Hence,
  \[ \dist(\alpha,z_\alpha) - \dist(\alpha, x) \leq 2\delta. \]
  On the other hand, $\dist(\alpha,z_\alpha) \geq \Delta/2$. Thus, 
  \[ \dist(\alpha,z_\alpha) - \dist(\alpha,x) \leq 2\delta \leq 4 \cdot \frac{\delta}{\Delta} \cdot \dist(\alpha,z_\alpha) \leq \varepsilon \cdot \dist(\alpha,z_\alpha). \]
  Thus, indeed $X$ is a $k$-center $\varepsilon$-coreset for furthest neighbor of $\mathcal{V}_i$. 
\end{proof}

To use~\Cref{thm:hs-round}, we need to bound the VC-dimension
of $\mathcal{A}_i$.
\begin{claim} \label{cl:vc-dim-ai}
    The VC-dimension of $\mathcal{A}_i$ is $\Oh(d k \log k)$. 
\end{claim}
\begin{proof} Observe that every set in $\mathcal{A}_i$ is an intersection of $k$ complements of balls in $(V,\dist)$.  Since taking complements of all sets in a set system does not change its VC-dimension, the claim follows from \Cref{lm:vc-facts}.
\end{proof}

Consequently,
by~\Cref{thm:hs-round}, to bound the size of a hitting set of $\mathcal{A}_i$
it suffices to bound the optimum value of its LP relaxation;
let $\tau_i^\ast$ be this optimum value.
We connect $\tau_i^\ast$ to the bound on the size of a $(k,\varepsilon)$-comatching in $(V,\dist)$
via the following simple rounding.
\begin{claim}\label{cl:kcenter-lp2comatching}
There exists a $(k, \varepsilon^2/4)$-comatching in $(V,\dist)$ of size
$\lfloor \tau_i^\ast/4 \rfloor$.
\end{claim}
\begin{proof}
The proof is identical to the proof of~\Cref{cl:lp2comatching}.
We copy-paste it below for completeness.

Consider an optimum solution to the dual LP to \textsc{Hitting Set} on $\mathcal{A}_i$;
for clarity of notation, we use $(y_\alpha)_{\alpha \in \mathcal{V}_i}$ for the dual variables
instead of $(y_{S_i^\alpha})_{\alpha \in \mathcal{V}_i}$. 
Clearly, $\sum_{\alpha \in \mathcal{V}_i} y_\alpha = \tau_i^\ast$.
Let $\mu$ be a probability distribution over $\mathcal{V}_i$, where $\alpha \in \mathcal{V}_i$ is sampled 
with probability $y_\alpha / \tau_i^\ast$.

For $\beta,\alpha \in \mathcal{V}_i$, we say that \emph{$\beta$ threatens $\alpha$} if $z_\alpha \in S_i^\beta$. 
Note that any $\alpha \in \mathcal{V}_i$ threatens itself. 
Let $\alpha \in \mathcal{V}_i$ be fixed and let $\beta$ be sampled from $\mathcal{V}_i$ according to $\mu$.
Then,

\begin{equation}\label{eq:kcenter-threaten}
 \mathrm{Prob}(\alpha \text{~is threatened by}~\beta) \leq  \sum_{\beta: z_\alpha\in S_\beta} \frac{y_\beta}{\tau_i^*} \leq \frac{1}{\tau_i^\ast}. 
\end{equation}

Indeed, \eqref{eq:kcenter-threaten} follows directly from the dual LP constraint for the 
vertex $z_\alpha$:
\[ \sum_{\beta \in \mathcal{V}_i \mid z_\alpha \in S_i^\beta} y_\beta \leq 1. \]

Let $K := \lfloor \tau_i^\ast / 4 \rfloor$ and let $\alpha_1,\alpha_2,\ldots,\alpha_{2K}$ be vertices
sampled independently from $\mathcal{V}_i$ according to the distribution $\mu$.
Then, by~\eqref{eq:kcenter-threaten},
the expected number of pairs $(i,j)$, $1\leq i,j \leq 2K$, $i \neq j$ such that
$\alpha_i$ threatens $\alpha_j$ is at most 
 \[ 4K^2 \cdot \frac{1}{\tau_i^\ast} \leq K. \]
Therefore, there exists a choice of $\alpha_1,\alpha_2,\ldots,\alpha_{2K}$ such that the number of such pairs
$(i,j)$ is at most $K$. 
Pick such a choice and, for every pair $(i,j)$, $1 \leq i,j \leq 2K$, $i \neq j$,
 such that $\alpha_i$ threatens $\alpha_j$, discard $\alpha_i$. We are left with a set $W$
 of at least $K$ tuples of $\mathcal{V}_i$ such that for every $\beta,\alpha \in W$, $\beta \neq \alpha$, 
 $\beta$ does \emph{not} threaten $\alpha$. 

 We claim that $\{(z_\alpha, \alpha)~|~\alpha \in W\}$ is the desired $(k,\varepsilon^2/4)$-comatching in $G$;
 it clearly has size at least $K$.
 Since $W \subseteq \mathcal{V}_i$, we have for every $\alpha \in W$
  \[ \dist(\alpha,z_\alpha) \geq i \delta. \]
On the other hand, for every $\beta,\alpha \in W$, $\beta \neq \alpha$, as $\beta$ does not threaten $\alpha$,
we have $z_\alpha \notin S_i^\beta$ so 
 \[ \dist(z_\alpha, \beta) \leq (i-1)\delta. \]
As $i \leq \lfloor 4\varepsilon^{-2} \rfloor$, we have
\[ \frac{i-1}{i} \leq 1 - \frac{1}{\lfloor 4\varepsilon^{-2} \rfloor} \leq 
1 - \frac{1}{4\varepsilon^{-2}} = 
1 - \frac{\varepsilon^2}{4}. \]
\end{proof}

By~\Cref{cl:kcenter-lp2comatching}, $\tau_i^\ast \leq 4L+3$, where $L$ is the maximum
size of an $(\varepsilon^2/4)$-comatching in $(V,\dist)$.
By~\Cref{thm:hs-round} and \Cref{cl:vc-dim-ai}, there exists a hitting set $X_i$ of $\mathcal{A}_i$
of size $\Oh(d k \log k \cdot L \log L)$.
\Cref{cl:kcenter-hs2coreset} asserts that $X_i$ is a $k$-center $\varepsilon$-coreset
for furthest neighbor of $\mathcal{V}_i$. 
Hence, for the $k$-center $\varepsilon$-coreset of the whole $V^k$, it suffices to take the $\{v_1^0, \ldots, v_{k+1}^0\}$ and the union 
of $\Oh(\varepsilon^{-2})$ sets $X_i$ for which $\mathcal{V}_i$ is nonempty. 
This finishes the proof of~\Cref{thm:kcenter-comatching}. 
\end{proof}

\begin{remark}\label{rm:nktime} The running time in the construction of \Cref{thm:kcenter-comatching} is $n^{O(k)}$ since we need to solve a \textsc{Hitting Set} LP with $n^{O(k)}$ constraints, one for each $k$-tuple. We can reduce $n$ to $k^{O(\poly(1/\eps))}$ by applying the coreset  for $k$-center of in planar metrics by  Bourneuf and Pilipczuk~\cite{BP25}; this coreset can be constructed in $n^{O(1)}$ time. Now we run our algorithm on the planar metrics induced by the coreset of Bourneuf and Pilipczuk. The running time is now $(2^{O(\poly(k/\eps))} + n^{O(1)})$. 
\end{remark}

\subsection{Reduction from $k$-comatching to double ladder}

In this section, we reduce bounding the size of  $k$-comatchings to bounding the size of double ladders as stated in \Cref{thm:comatching-to-ladder}.

\KComatchingToDoubleLadder*

First, we recall the definition of a double ladder. 

\DoubleLadderDef*

To prove~\Cref{thm:comatching-to-ladder},
we will use the polynomial Ramsey bounds for graphs of bounded VC-dimension
of Nguyen, Scott, and Seymour~\cite{NSS24}.
A \emph{VC-dimension} of an (undirected, simple) graph $G$ is the VC-dimension
of the set system of open neighborhoods, that is, $\{N_G(v)~|~v \in V(G)\}$. 

\begin{theorem}[Theorem 1.2 of \cite{NSS24}]\label{thm:NSS}
    For every $d\geq 1$, there exists $\delta>0$ such that every $n$-vertex graph $G$ of VC-dimension at most $d$ contains either a clique or an independent set of size at least $n^\delta$.
\end{theorem}

\begin{proof}[Proof of~\Cref{thm:comatching-to-ladder}.]
Fix positive integers $k,d$ and a real $\varepsilon > 0$.
Let $(V,\dist)$ be a metric space whose ball set system has
VC-dimension at most $d$ and let $\mathcal{M}$ be a $(k,\varepsilon)$-comatching
in $(V,\dist)$ of size $L > 1$.
Enumerate $\mathcal{M}$ as 
 \[\mathcal{M} = \{(p_1, X_1), (p_2, X_2), \ldots, (p_L, X_L)\}. \]
Furthermore, for each $j \in [L]$, enumerate
$X_j$ as $(q_j^1, q_j^2, \ldots, q_j^k)$.
Let $R$ be the distance from the definition of the $(k,\varepsilon)$-comatching
$\mathcal{M}$. 

For every $i \in [k]$, we define four graphs 
$H_i^\rightarrow$, $H_i^\leftarrow$,
$\bar{H}_i^\rightarrow$, $\bar{H}_i^\leftarrow$
as follows.
The vertex set of all four graphs is $[L]$.
For $1 \leq a < b \leq L$, we have $ab \in E(H_i^\rightarrow)$ if
and only if $\dist(p_a, q_b^i) \leq (1-\varepsilon)R$
and $ab \in E(\bar{H}_i^\rightarrow)$ if
and only if $\dist(p_a, q_b^i) \leq (1-\varepsilon/2)R$.
Similarly,
for $1 \leq a < b \leq L$, we have $ab \in E(H_i^\leftarrow)$
if and only if $\dist(p_b, q_a^i) \leq (1-\varepsilon)R$
and $ab \in E(\bar{H}_i^\leftarrow)$
if and only if $\dist(p_b, q_a^i) \leq (1-\varepsilon/2)R$
(Note the difference of $\varepsilon$ vs $\varepsilon/2$ 
between $H$ and $\bar{H}$.)

\begin{claim}\label{cl:Hi-vcdim}
    For every $i \in [k]$, the VC-dimension
    of $H_i^\leftarrow$, $H_i^\rightarrow$,
    $\bar{H}_i^\leftarrow$, $\bar{H}_i^\rightarrow$
    is $\Oh(d)$. 
\end{claim}
\begin{proof}
    Observe that every neighborhood in any of the four aforementioned graphs
    can be constructed by means of basic set operations
    from a constant number of sets
    being either balls in $(V,\dist)$, which has VC-dimension $d$, or sets of consecutive indices in $[L]$, and the family of such sets has VC-dimension
    $\Oh(1)$. The claim follows from bounds on VC-dimension
    under basic set operations, as shown in~\cite{DBLP:journals/jacm/BlumerEHW89}.
\end{proof}

Let $\delta = \delta(d)$ be the exponent of~\Cref{thm:NSS}
for the VC-dimension bound of~\Cref{cl:Hi-vcdim}.
Repeated application of~\Cref{thm:NSS} to every
$H_i^\leftarrow$, $H_i^\rightarrow$,
$\bar{H}_i^\leftarrow$, $\bar{H}_i^\rightarrow$,
yields a set $I \subseteq [L]$ 
of size at least $L^{\delta^{4k}}$ such that for every $i \in [k]$,
each of the graphs
$H_i^\leftarrow[I]$, $H_i^\rightarrow[I]$,
$\bar{H}_i^\leftarrow[I]$, $\bar{H}_i^\rightarrow[I]$,
is a clique or an independent set.

In the remainder of the proof, we will show that $(V,\dist)$ admits
either an $(\varepsilon/2)$-comatching or an $(\varepsilon/2)$-double ladder
of size $|I| \geq L^{\delta^{4k}}$. This will conclude the proof. 
Note that as $L > 1$, we have $|I| > 1$.

Consider first the case where there exists $i \in [k]$ such that
both $\bar{H}_i^\leftarrow[I]$ and $\bar{H}_i^\rightarrow[I]$ 
are cliques. 
Then, $\{(p_j, q_j^i)~|~j \in I\}$ is an $(\varepsilon/2)$-comatching
of size $|I| \geq L^{\delta^{4k}}$ and we are done. 
Assume then that no such $i$ exists.

By the definition of a comatching, for every $a,b \in I$, $a < b$,
we have $\dist(p_a, X_b) \leq (1-\varepsilon)R$ and 
$\dist(p_b, X_a) \leq (1-\varepsilon)R$. 
As $|I| > 1$, this implies that there exists $i^\rightarrow \in [k]$
such that $H_{i^\rightarrow}^\rightarrow[I]$ is a clique
and there exists $i^\leftarrow \in [k]$ such that 
$H_{i^\leftarrow}^\leftarrow[I]$ is a clique.
Clearly, $\bar{H}_{i^\rightarrow}^\rightarrow[I]$
and $\bar{H}_{i^\leftarrow}^\leftarrow[I]$ are also cliques.
By the case excluded in the previous paragraph, 
$i^\leftarrow \neq i^\rightarrow$ and, furthermore, 
both $\bar{H}_{i^\leftarrow}^\rightarrow[I]$
and $\bar{H}_{i^\rightarrow}^\leftarrow[I]$ are independent sets. 
In other words, for every $a,b \in I$, $a < b$, we have
\begin{align*}
    \dist(p_b, q_a^{i^\leftarrow}) & \leq (1-\varepsilon)R, \\
    \dist(p_b, q_a^{i^\rightarrow}) & > (1-\varepsilon/2)R, \\
    \dist(p_a, q_b^{i^\rightarrow}) & \leq (1-\varepsilon)R, \\
    \dist(p_a, q_b^{i^\leftarrow}) & > (1-\varepsilon/2)R. \\
\end{align*}
Recall that, from the definition of a $(k,\varepsilon)$-comatching,
for every $a \in I$ we have 
\[ \dist(p_a, q_a^{i^\leftarrow}) > R \quad \mathrm{and} \quad 
    \dist(p_a, q_a^{i^\rightarrow}) > R. \]
This implies that the sequence
\[ \left((p_a, q_a^{i^\leftarrow}, q_a^{i^\rightarrow})~|~ a \in I\right)\]
forms an $(\varepsilon/2)$-double ladder (for the distance $(1-\varepsilon/2)R$)
of size $|I| \geq L^{\delta^{4k}}$ in $(V,\dist)$.

This finishes the proof of~\Cref{thm:comatching-to-ladder}. \hfill$\square$
\end{proof}

Note that degree of the polynomial $f_{k,d}$ of~\Cref{thm:comatching-to-ladder}
is $\delta^{-4k}$, where $\delta$ is the degree
of~\Cref{thm:NSS} for VC-dimension $\Oh(d)$. 
A close inspection of~\cite{NSS24}
reveals that the dependency of $\delta^{-1}$ on the VC-dimension
is exponential; hence, the degree of $f_{k,d}$ depends exponentially
on $k$ and $d$.

\section{The Size of Double Ladders}\label{sec:doubleladder}

This section is devoted to the proof of~\Cref{thm:doubleladder-bound}, which establishes a bound on the size of $\eps$-double ladder. The structure of this section closely parallels that of~\Cref{sec:comatching-bound}, with most arguments applying on the set of triplets $\mathcal{D = }\{\dltri{i}\}_{i=1}^L$ in place of the pairs $\mathcal{M} = \{(p_i, q_i)\}_{i=1}^L$. The analysis diverges slightly during the final contradiction involving some topological case analysis, but the overall strategy remains the same. Consequently, the resulting bound differs from the $\eps$-comatching bound by a $\poly(\eps^{-1})$ factor.

Fix $\varepsilon > 0$. We will use a notation and definitions analogous to that introduced for $\varepsilon$-comatchings in~\Cref{sec:comatching-bound}. A pair $(G, \DL)$ is an \EMPH{instance} if $G$ is a plane edge-weighted graph and $\DL$ is an $\varepsilon$-double ladder. Again, we assume that $G$ comes with a fixed embedding in a plane. The vertices $p, \dlpt{b}, \dlpt{t}$ for any triple $\dltri{} \in \DL$ are called \EMPH{terminals}.

Similarly to the proof for comatchings, we assume that the lengths of all paths are pairwise different, and that the shortest paths between any pairs of terminals $(p_1, \dlpt{t}_2)$ and $(p_1, \dlpt{b}_2)$ are unique.
Again, we assume that our instance is \EMPH{clean}, that is, that all the terminals are of degree $1$, that $G$ contains no vertices of degree $2$, and that every edge lies on some shortest path between either $(p_i, \dlpt{t}_j)$ or $(p_j, \dlpt{b}_i)$ for some triples $\dltri{i}, \dltri{j} \in \DL$ with $i < j$.

We use the same notation: $\ppath{p_1 \dlpt{t}_2}, \ppath{p_1 \dlpt{b}_2}$ to denote the shortest paths between terminals; $G^{\DL_0}$ for $\DL_0 \subseteq \DL$ to denote the subgraph of $G$ consisting only on shortest paths between terminals of $\DL_0$; and $G^{\neg \sigma}$ as a shorthand for $G^{\DL - \{\sigma\}}$ for any $\sigma \in \DL$.

We say that an instance $(G, \DL)$ is \EMPH{local} if for every $\sigma \in \DL$, the terminals of $\sigma$ lie in the same face of $G^{\neg \sigma}$. For a triple $\sigma = \dltri{} \in \DL$, a \EMPH{link} is a curve in good position w.r.t. G, connecting $p$ to either $\dlpt{b}$ or $\dlpt{t}$, which does not intersect any edge of $G^{\neg \sigma}$. A \EMPH{linked instance} is an instance $(G, \DL)$ with two links between resp. $p$ and $\dlpt{t}$, and $p$ and $\dlpt{b}$, fixed for every $\dltri{} \in \DL$. A linked instance is \EMPH{noncrossing} if all links are pairwise disjoint.

The proof of~\Cref{thm:doubleladder-bound} will, similarly to comathing case, consist of the three following steps.

\begin{lemma}\label{lem:DL-to-local}
    For a real $\varepsilon > 0$, let $\LDL_{\mathrm{local}}(\varepsilon)$ be the maximum size
    of a $\varepsilon$-double ladder in a local instance. 
    Then, the maximum size of a $\varepsilon$-double ladder in any instance is of
    $\Oh(\LDL_{\mathrm{local}}(\varepsilon)^9 \varepsilon^{-8})$. 
\end{lemma}

\begin{restatable}{lemma}{DLtoNoncrossing}\label{lem:DL-to-noncrossing}
    For a real $\varepsilon > 0$, let $\LDL_{\asymp}(\varepsilon)$ be the maximum size
    of a $\varepsilon$-double ladder in a linked noncrossing instance. 
    Then, the maximum size of a $\varepsilon$-double ladder in a local instance is of 
    $\Oh(\LDL_{\asymp}(\varepsilon)^2 \log^8 \LDL_{\asymp}(\varepsilon) \cdot \varepsilon^{-8} )$. 
 \end{restatable}

 \begin{restatable}{lemma}{DLtoNoncrossingBound}\label{lem:DL-noncrossing-bound}
    For a real $\varepsilon > 0$, the maximum size of a $\varepsilon$-double ladder
    in a linked noncrossing instance is at most $\Oh(\varepsilon^{-4})$.
 \end{restatable}
 Putting these three lemmas together immediately gives the proof of~\Cref{thm:doubleladder-bound}.

\subsection{Distance profiles and separation toolbox}
We start by introducing the notation for distance-based arguments, analogous to the one used in the case of comatching. We additionally restate and prove some separation lemmas adapted to our current setting. 
For the remainder of the proof, $(G, \mathcal{D})$ is a clean $\varepsilon$-double ladder instance with at least three terminal triplets. We put $L = |\DL|$, enumerate $\DL = \{\dltri{1}, \dots, \dltri{L}\}$ and put

$$
    R := \min_{\dltri{} \in \DL} \min(\dist(p, \dlpt{t}), \dist(p, \dlpt{b})),
    \qquad
    \delta := \frac{\varepsilon}{100}R.
$$

Recall that $\delta$ serves as one ``unit'' of distance, and differences of the order of a few
$\delta$s do not matter. 
\begin{claim}\label{diam-lnci-DL}
    If $L \geq 3$, then the diameter of $G$ is at most $8(1-\varepsilon) R$.
\end{claim}
\begin{proof}
    Recall that the instance is clean, that is, 
    every vertex lies on a path 
    $\ppath{p_{i_1} \dlpt{t}_{i_2}}$
    or $\ppath{\dlpt{b}_{i_1} p_{i_2}}$ for some $i_1 < i_2$.
    
    Fix $1 \leq i_1 < i_2 \leq L$. 
    Every vertex on $\ppath{p_{i_1} \dlpt{t}_{i_2}}$
    is within distance at most $2(1-\varepsilon)R$ from $p_1$, as
    the lengths of $\ppath{p_{i_1} \dlpt{t}_{i_2}}$
    and $\ppath{p_1\dlpt{t}_{i_2}}$ are at most $(1-\varepsilon)R$ each.
    Similarly, 
    every vertex on $\ppath{\dlpt{b}_{i_1} p_{i_2}}$
    is within distance at most $2(1-\varepsilon)R$ from $p_L$, as
    $\ppath{\dlpt{b}_{i_1} p_{i_2}}$
    and $\ppath{\dlpt{b}_{i_1} p_L}$ are at most $(1-\varepsilon)R$ each.
    Finally, the distance between $p_1$ and $p_L$ is at most $4(1-\varepsilon)R$, as they are connected by paths
    $\ppath{p_1 \dlpt{t}_3}$, $\ppath{\dlpt{t}_3 p_2}$, $\ppath{p_2 \dlpt{b}_1}$, $\ppath{\dlpt{b}_1 p_L}$, each of length
    at most $(1-\varepsilon)R$.
\end{proof}

As an immediate corollary, we obtain that $\Udist(u, v) \leq 8R \cdot \delta^{-1} = 800 \varepsilon^{-1}$ for every $u,v \in V(G)$.

Analogously to the comatching case, for a set $Z \subseteq V(G)$, the \emph{rounded distance profile} of a vertex $v \in V(G)$
is the function 
$ \Uprof{Z}{v} : Z \to \mathbb{Z}$ defined as $\Uprof{Z}{v}(z) = \Udist(v,z)$. 
Again,~\Cref{thm:prof-vc} immediately implies the following.
\begin{lemma}\label{lem:Uprof-vc-DL}
 Let $Z \subseteq V(G)$.
 For $v \in V(G)$, let 
 \[ \widehat{\mathcal{Z}}_Z[v] = \{(z,i) \in Z \times \{0,1,\ldots,\lfloor 800\varepsilon^{-1} \rfloor\}~|~\Uprof{Z}{v}(z) \leq i\}.\]
Then, the set system $\{\widehat{\mathcal{Z}}_Z[v]~|~v \in V(G)\}$ has VC-dimension at most $4$. 
In particular, by the Sauer-Shelah Lemma, its size is bounded by $\Oh(|Z|^4 \varepsilon^{-4})$. 
\end{lemma}

The following observations are restated, but require some different arguments to prove in our new setting.

\begin{lemma}\label{lem:general-sep-cont}
    Let $Z \subseteq V(G)$ and $\dltri{i}, \dltri{j}$ be two distinct triplets of  $\DL$ such that $i < j$. Then neither $\Mpath{p_i}{\dlpt{t}_j}$, nor $\Mpath{p_j}{\dlpt{b}_i}$ contains a vertex that is within distance $\delta$ from $Z$. 
\end{lemma}
\begin{proof}
    For contradiction, let $u$ be a vertex on $\Mpath{p_i}{\dlpt{t}_j}$ such that there exists $z \in Z$
    with $\dist(z,u) \leq \delta$. 
    As $\Uprof{Z}{p_i} = \Uprof{Z}{p_j}$, we have $|\dist(p_i,z) - \dist(p_j,z)| < \delta$. 
    Hence,
    \[ |\dist(p_i,u) - \dist(p_j,u)| \leq |\dist(p_i,z) - \dist(p_j,z)| + 2\dist(u,z) \leq 3\delta. \]
    We infer that $\dist(p_j,\dlpt{t}_j) \leq \dist(p_i, \dlpt{t}_j) + 3\delta$, which is a contradiction to $\dist(p_j, \dlpt{t}_j) \geq R$ and $\dist(p_i, \dlpt{t}_j) \leq R - 100\delta$.
\end{proof}
 \begin{lemma} \label{lem:sep-three-pairs-DL}
    Let $X, Z \subseteq V(G)$ be such that every vertex of $X$ is within distance $\delta$ from some vertex of $Z$, and $\dltri{i}, \dltri{j}, \dltri{k} \in \DL$ be three distinct triplets with $i < j < k$. If $p_i, p_j, p_k$ all have the same rounded distance profile to $Z$, then 
    $\dltri{j}$ is not separated by $X$.
 \end{lemma}
 \begin{proof}
    $p_i$ cannot be separated by $X$ from either of $\dlpt{t}_j, \dlpt{t}_k$, because otherwise $\Mpath{p_i}{\dlpt{t}_j}$ (likewise $\Mpath{p_i}{\dlpt{t}_k}$) intersects $X$, which is a contradiction with~\Cref{lem:general-sep-cont}.
    By the same reasoning, $p_j$ is not separated by $X$ from $\dlpt{b}_i, \dlpt{t}_k$ and $p_k$ is not separated by $X$ from $\dlpt{b}_i, \dlpt{b}_j$. 
    Cumulatively, $p_i, \dlpt{b}_i, p_j, \dlpt{t}_j, \dlpt{b}_j, p_k, \dlpt{t}_k$ are in one connected component of $G-X$, henceforth
    denoted by $C$.
    Therefore, only $\dlpt{t}_i$ and/or $\dlpt{b}_k$ can be separated by $X$ from $C$. 
 \end{proof}
\begin{claim}\label{lem:sep-bound-DL}
    Let $X \subseteq V(G)$ be such that $X$ is contained in the union of at most $\alpha$ paths, each of length at most $\beta R$ for some $\alpha,\beta > 0$. Then, the number of triplets of $\DL$ separated by $X$ is $\Oh(\alpha^4 \varepsilon^{-4} (1+\beta\varepsilon^{-1})^4)$.
\end{claim}
 \begin{proof}
    Since $X$ is contained in the union of at most $\alpha$ paths, each of length at most $\beta R$,
    there exists a set $Z \subseteq V(G)$ of size $\Oh(\alpha (1+ \beta \varepsilon^{-1}))$ such that every vertex of $X$ is within distance $\delta$ from some vertex of $Z$. By~\Cref{lem:Uprof-vc}, there are $\Oh(|Z|^4 \varepsilon^{-4}) = \Oh(\alpha^4 \varepsilon^{-4}(1+\beta \varepsilon^{-1})^4)$ distinct rounded distance profiles to $Z$. Group triplets $\dltri{i} \in \DL$ according to the rounded distance profile of $p_i$ to $Z$. 
    For any group that has at least three triplets, say $s$ triplets, order the them with respect to their indices, and relabel them as $\cup_{1 \leq i \leq s} \dltri{i}$. For $k \in s, s-1, \dots, 3$ we consider the set $\{\dltri{k-2}, \dltri{k-1}, \dltri{k}\}$. At step $k$,~\Cref{lem:sep-three-pairs-DL} implies that $\dltri{k-1}$ cannot be separated by $X$.Thus, in the group, only $\dltri{1}$ and $\dltri{s}$ can be separated by $X$. Consequently, the number of pairs separated by $X$ is $\Oh(\alpha^4 \varepsilon^{-4} (1+\beta \varepsilon^{-1})^4)$.
 \end{proof}

\subsection{Reduction to local double ladder}

In this subsection, we provide a proof of~\Cref{lem:DL-to-local}.
The proof is essentially the same as the proof of~\Cref{lem:to-local}, with the only minor technical difference that we look at triples, instead of pairs of terminals, hence we need to substitute the paths separation lemma with its generalization to triplets, i.e.,~\Cref{lem:sep-bound-DL}.
The rest of the proof remains exactly the same. We repeat it for completeness.

Consider a function $\Lambda$ that maps any subset $I$ of ($\LDL_{\mathrm{local}}(\varepsilon)+1$) triples of points from $\DL$ to a subset of $\LDL_{\mathrm{local}}(\varepsilon)$ triples of points.
The mapping is defined as follows: given a subset $I$ of size $\LDL_{\mathrm{local}}(\varepsilon)+1$, we put $\Lambda(I) = I - \{\sigma_a\}$, where $\sigma_a \in \DL$ is an arbitrarily fixed triple such that $\sigma_a$ is separated by the vertices of $G^{I - \{\sigma_a\}}$. Since $\LDL_{\mathrm{local}}(\varepsilon)$ is the maximum size of a local instance, such a pair always exists.

By pigeonhole principle, there exists a set $S \in \binom{\DL}{\LDL_{\mathrm{local}}(\varepsilon)}$ such that
\[|\Lambda^{-1}(S)| \geq \binom{\LDL}{\LDL_{\mathrm{local}}(\varepsilon)+1}/\binom{\LDL}{\LDL_{\mathrm{local}}(\varepsilon)} = \frac{\LDL-\LDL_{\mathrm{local}}(\varepsilon)}{\LDL_{\mathrm{local}}(\varepsilon)+1}>\frac{\LDL}{\LDL_{\mathrm{local}}(\varepsilon)+1}-1\] 
Let us denote this quantity by $C = \frac{\LDL}{\LDL_{\mathrm{local}}(\varepsilon)+1}-1$. There exist at least $C$ indices $\{a_1, a_2, ..., a_C\}$ such that $\sigma_a \not\in S$ and $G^S$ separates $\sigma_a$ for each $i \in [C]$. 

Since $G^S$ is a union of not more than $\LDL_{\mathrm{local}}(\varepsilon)^2$ paths, each of length at most $R$, we can apply~\Cref{lem:sep-bound-DL} to get $C = \Oh(\LDL_{\mathrm{local}}(\varepsilon)^8 \varepsilon^{-8})$. This in turn implies that $\LDL = \Oh(\LDL_{\mathrm{local}}(\varepsilon)^9 \varepsilon^{-8})$, which finishes the proof of the lemma.

\subsection{Reduction to the linked noncrossing double ladder}

In this subsection, we prove \Cref{lem:DL-to-noncrossing}.
Many of the following arguments are borrowed directly or adapted from the proof of~\Cref{lem:to-noncrossing}; in light of this, we will often simply restate claims or lemmas, omitting their proofs, and keep discussions brief, until the point where we deviate from the prior. 

Let $(G, \DL)$ be a local double ladder instance. Fix a triangulation $G'$ of $G$ obtained by repeatedly subdividing non-triangular faces of $G$ by arbitrarily adding virtual edges. 

To make our instance linked, we iteratively assign a link to every triplet via the following process. Pick any unassigned $\sigma = \dltri{}$, and let $F_p, F_\dlpt{t}$ and $F_\dlpt{b}$ denote arbitrary faces of $G$ incident resp. to $p, \dlpt{t}, \dlpt{b}$. Trace a path from $F_\dlpt{t}$ to $F_p$ in $G^*$ which does not cross any edge of $G^{\neg \sigma}$ (existence guaranteed by locality), and extend the endpoints of the drawing within $F_\dlpt{t}, F_p$ to reach resp. $\dlpt{t}$ and $p$. The curve $\gamma_{p, \dlpt{t}}$ traced this way is a link for $\sigma$. We trace $\gamma_{p, \dlpt{b}}$ between $p$ and $\dlpt{b}$ exactly the same way. Both paths share an endpoint at $p$, hence $\gamma_{p, \dlpt{t}}, \gamma_{p, \dlpt{b}}$ can be drawn such that their intersection is just the point $p$.

\subsubsection{Tree decomposition and separators}

For each $i \in [L]$, define tree $T_i$ to be the union of union of paths $\left( \cup_{j: j < i} \ppath{p_i\dlpt{b}_j} \right) \cup \left( \cup_{j: j > i} \ppath{p_i\dlpt{t}_j} \right)$.
For $L \geq 3$ the 
trees $T_i$ satisfy the assumptions of \Cref{lem:log-cover-spanner}
if we use $T_2$ in the role of the first tree:
the tree $T_2$ intersects $T_1$ at $t_L$ and every tree $T_i$ for $i\geq 3$
at $b_1$.
Let $H$ be the result of applying \Cref{lem:log-cover-spanner}.
Applying \Cref{thm:tree-decomp} to triangulation $G'$, we obtain $(T, \beta)$, a tree decomposition of $G$. Recall that with the decomposition obtained this way, its every adhesion is a separator formed by a union of two upward paths of $H$, whose endpoints are connected via a closing edge of $E(G') - E(H)$. For the sake of convenience, we use $S$ to refer to the subgraph of $G$ containing the vertices of such separator and with edges restricted only to those two paths. By $\xi_S$ we denote the planar loop obtained by taking the drawing of both paths plus the closing edge.

We will say that a triplet $\sigma = \dltri{} \in \DL$ is \emph{bad} with respect to $S$ if either of the links $\gamma_{p, \dlpt{t}}, \gamma_{p, \dlpt{b}}$ intersects $\xi_S$. There are only two ways for a triplet to be bad: either it is \emph{separated} by $S$, here meaning that any pair of its terminals is separated by $S$, or it contaminates $S$ ($\exists e \in S$ s.t. $e \notin G^{\neg\sigma}$). 
These two cases are indeed exhaustive: this follows from the analogue of
\Cref{clm:types_of_bad}, and its proof in the double ladder setting remains the same, the only difference being that instead of a single link $\gamma_\sigma$ per pair, we look at the union of two links $\gamma_{p, \dlpt{b}}, \gamma_{p, \dlpt{t}}$.

As before, our aim is to bound the number of bad pairs w.r.t to a separator, so that we can take a number of separators, pick a good pair from each component, and in the end obtain a non-crossing instance of reasonable size.

\subsubsection{Bounding bad triplets}

We continue to mirror the proof of~\Cref{lem:to-noncrossing} and now seek to bound the number of each type of bad triplets.
Since our instance is clean, we can argue that $G$ has $O(L^4)$ vertices. The argument follows the proof of \Cref{clm:instance_size_bound} line by line, the only difference being that trees $T_i$ now have $2L$ leaves, instead of $L$, and that there are $3L$ vertices of degree $1$, instead of $2L$. The bound asymptotically remains the same.

We reuse the \Cref{clm:noncrossing-pairs-cont} to bound the number of triples contaminating $S$. We restate the claim here for completeness. Its proof remains exactly the same, as we retain the property that every separator $S$ is covered by a family of $\Oh(\log^2 L)$ terminal to terminal paths.

\begin{claim}\label{clm:DL-noncrossing-triples-cont}
    There are at most $\Oh(\log^2 L)$ triples in $\DL$ that are \emph{contaminating} $S$.
\end{claim}

The bound on the number of separated triplets is trivial given our previous separation lemmas.

 \begin{claim} \label{clm:DL-noncrossing-triples-sep}
     There are $\Oh(\log^8 L \cdot \varepsilon^{-8})$ triples in $\DL$ that are separated by $S$.
 \end{claim}
 \begin{proof}
     Every subpath of $T_i$ is of length at most $2R$. By \Cref{lem:log-cover-spanner}, $S$ can be covered by $\alpha = \Oh(\log^2L)$ subpaths of some $T_i$s. Thus the claim holds as a direct result of~\Cref{lem:sep-bound-DL} with $\alpha$ and $\beta = 2$.
 \end{proof}

\subsubsection{Good triplets with noncrossing links}

We are finally ready to prove \Cref{lem:DL-to-noncrossing}. Again, the intuition here remains the same as in the comatching case and the proof follows the one of \Cref{lem:to-noncrossing} very closely. We retell it here for the sake of completeness, although in a more concise way. Refer to the proof of \Cref{lem:to-noncrossing} for a more detailed explanation.

\DLtoNoncrossing*
\begin{proof}[Proof of \Cref{lem:DL-to-noncrossing}]
    We define the weight function $\chi: V(G) \to \{0, 1\}$ which assigns $1$ to the terminals of all triples of $\DL$, and $0$ otherwise. We fix parameter $\alpha$, roughly of order $L^{1/2}$.

    We use \Cref{lem:tree-decomp-split} to obtain a set $F \subseteq E(T)$ which splits $T$ in a balanced way w.r.t. $\chi$.
    Fix any connected component $C$ of $T - F$ and look at the set $V_C$ of vertices defined as the union of bags of $C$ minus all separators $S_f$ corresponding to adhesions of all $f \in F$.
    
    By~\Cref{clm:DL-noncrossing-triples-cont}, and~\Cref{clm:DL-noncrossing-triples-sep}, there are at most $\alpha \cdot \Oh(\log^8 L \cdot \varepsilon^{-8})$ triples which are not good with respect to some $S_f$. Moreover, by the structure of our separators, at most $\Oh(\alpha)$ terminals belong to some $S_f$. Hence, the total number of terminal vertices in the union of bags of $C$ which are either bad or in some $S_f$ is at most $\alpha \cdot c \cdot \log^8 L \cdot \varepsilon^{-8}$, where $c$ is some universal constant.

    The number of all terminal vertices in such union is at least $L/2\alpha$, due to \Cref{lem:tree-decomp-split}. Putting $\alpha$ such that $4\alpha^2 \cdot c \leq L \log^{-8} L \cdot \varepsilon^{8}$, we get that every such component $C$ contains a triple $\sigma_C$ good w.r.t. all separators $S_f$.

    The links of such triples are pairwise-nonintersecting, hence restricting our instance to those only gives us a noncrossing instance of size $\alpha$, which finishes the proof.
\end{proof}

\subsection{Bounding linked noncrossing double ladder}

In this section we develop the following bound. 

\DLtoNoncrossingBound*

Let $\mathcal{I} = (G,\DL)$ be a linked non-crossing double ladder instance,
where $\DL = (\sigma_i~|~1 \leq i \leq L)$ and $\sigma_i = \dltri{i}$ for
$i \in [L]$. 

We have the following initial observation.
\begin{claim}\label{clm:kcenter-touch-topdown}
Let $a, b, c, d \in [L]$ with $a < b \leq d$ and $a \leq c < d$. 
Then, 
$\ppath{p_a\dlpt{t}_b}$ and $\ppath{\dlpt{b}_c p_d}$ do not intersect.
\end{claim}
\begin{proof}
    If these two paths intersect, then
    \[ 2(1-\varepsilon)R \geq \dist(p_a, \dlpt{t}_b) + \dist(\dlpt{b}_c, p_d) \geq \dist(p_a, \dlpt{b}_c) + \dist(p_d, \dlpt{t}_b) \geq 2R. \]
    This is a contradiction.
\end{proof}
\subsubsection{Equidistant instance}

\begin{figure}[ht]
    \centering
    \begin{subfigure}{0.32\textwidth}
        \centering
        \includegraphics[width=1\linewidth]{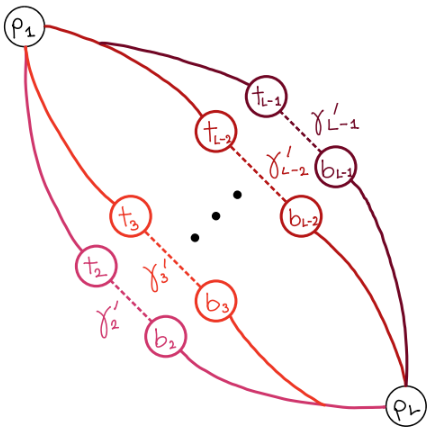}
        \caption{Rib Cage}
        \label{fig:rib_cage}
    \end{subfigure}
    \begin{subfigure}{0.32\textwidth}
        \centering
        \includegraphics[width=1\linewidth]{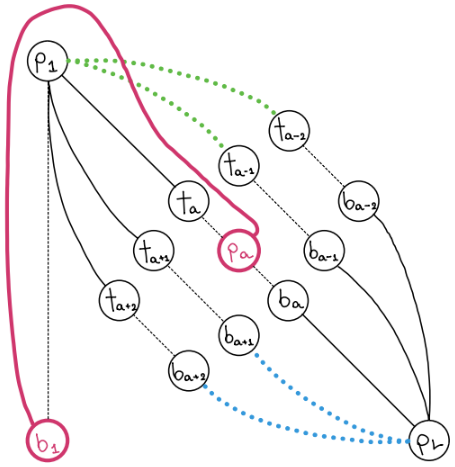}
        \caption{Type T}
        \label{fig:type_T}
    \end{subfigure}
    \begin{subfigure}{0.32\textwidth}
        \centering
        \includegraphics[width=1\linewidth]{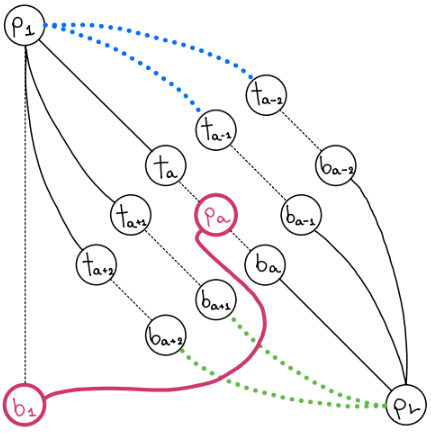}
        \caption{Type B}
        \label{fig:type_B_DL}
    \end{subfigure}
    \caption{The dashed lines represent links. For visual clarity, we assume that the links do not intersect any paths ($\gamma_a' = \Gamma(\sigma_a)$). (a) Each pair $\sigma_a$ forms a rib as the concatenation of a subpath of $\ppath{p_1 \dlpt{t}_a}$ from $p_1$ to the endpoint of $\gamma_a'$, the curve $\gamma_a'$, and the subpath of $\ppath{\dlpt{b}_a p_L}$ from the second endpoint of $\gamma_a'$ to $p_L$. Ribs may share a prefix or suffix, in this picture ribs $\gamma_{L-1}$ and $\gamma_{L-2}$ share a prefix, and $\gamma_1, \gamma_2$ share a suffix. (b) $\sigma_a$ is a pair where Type T holds, that is $\ppath{\dlpt{b}_1 p_a}$ (thick pink line) intersects transversally every path $\ppath{p_1 \dlpt{t}_b}$ for $a < b < L$ (green dotted lines), in the decreasing order of indices $b$, and does not intersect transversally any path $\ppath{p_L \dlpt{b}_b}$ for $1 < b < a$ (blue dotted lines). (here we have only drawn 5 pairs other than $\sigma_1$, $\sigma_L$). (c) Here $\sigma_a$ is instead a pair where Type B holds, Meaning, $\ppath{\dlpt{b}_1 p_a}$ (thick pink line) intersects transversally every path $\ppath{p_1 \dlpt{t}_b}$ for $a < b < L$ (green dotted lines) and does not intersect transversally any path $\ppath{p_L \dlpt{b}_b}$ for $1 < b < a$ (blue dotted lines).}
    \label{fig:ribs-typeT-typeB}
\end{figure}

We now filter $\DL$ to get the following substructure.
\begin{definition}
    A double ladder $\DL' = (\sigma_i')_{i \in [L']}$ 
    where $\sigma_i' = \dltri{i'}$ is called \emph{equidistant} 
    if for every $1 < i < L'$ we have the following:
    
    \[ \Udist(p_1', \dlpt{t}_i') = \Udist(p_1', \dlpt{t}_{L'}') \qquad \mathrm{and} \qquad \Udist(\dlpt{b}_1', p_i') = \Udist(\dlpt{b}_1', p_{L'}'). \]

    \bigskip %

\end{definition}

Note that for every $I \subseteq [L]$, the double ladder $\DL$ restricted
only to triplets $\sigma_i$ for $i \in I$ is a double ladder again (with the same $R$). 
We call this the double ladder induced by $I$ and denote by $\DL[I]$.

\begin{lemma}\label{lem:kcenter-find-equidistant}
    Let $L'$ be the maximum size of a set $I \subseteq [L]$ such that
    $\DL[I]$ is equidistant. 
    $|\DL| = \Oh(L' \varepsilon^{-2}).$
\end{lemma}
\begin{proof}
Tag each $1 < i \leq |\DL|$
with the following tuple: $(\Udist(p_1, \dlpt{t}_i), \Udist(\dlpt{b}_1, p_i))$. 
There are $\Oh(\varepsilon^{-2})$ possible tags.
Let $I_1 \subseteq \{2,3,\ldots,L\}$
be a nonempty subset of pairs with the same tag.
Let $r$ be the maximum element of $I_1$. 
Observe now that $I := \{1,r\} \cup I_1$ induces an equidistant structure.
The lemma follows.
\end{proof}

\Cref{lem:kcenter-find-equidistant} allows us to restrict to the case
when $\DL$ is equidistant (and aim for an $\Oh(\varepsilon^{-2})$ bound). 
Hence, in what follows we assume that $\DL$ is equidistant.

\subsubsection{Forbidden intersections}

A special case of~\Cref{clm:kcenter-touch-topdown} is the following:
\begin{claim}\label{clm:kcenter-ribs}
For every $1 < a \leq L$ and $1 \leq b < L$ we have $\ppath{p_1 \dlpt{t}_a} \cap \ppath{\dlpt{b}_b p_L} = \emptyset$.
\end{claim}

\Cref{clm:kcenter-ribs} motivates the following definition.
For $1 < a < L$, let $\gamma_a'$ be a minimal part of the link of $\Gamma(\sigma_a)$
that connects $\ppath{p_1 \dlpt{t}_a}$ with $\ppath{\dlpt{b}_a p_L}$,
and let $\gamma_a$ be a concatenation of a subpath of $\ppath{p_1 \dlpt{t}_a}$ from $p_1$ to the endpoint
of $\gamma_a'$, the curve $\gamma_a'$, and the subpath of $\ppath{\dlpt{b}_a p_L}$ from 
the second endpoint of $\gamma_a'$ to $p_L$. 
Then, two curves $\gamma_a$ and $\gamma_b$ for $1 < a  < b < L$ can share only
a prefix starting at $p_1$ or a suffix ending at $p_L$, but otherwise are disjoint
and form a ``rib cage'' structure; see~\Cref{fig:rib_cage}.

Before we proceed further with cleaning the rib cage, we infer from the fact that $\DL$ is equidistant
that some more pairs of paths are disjoint.

\begin{claim}\label{clm:kcenter-TT-or-BB}
For every $1 < a \leq b < L$ we have 
$\ppath{p_1 \dlpt{t}_a} \cap \ppath{p_b \dlpt{t}_L} = \emptyset$
and
$\ppath{\dlpt{b}_1 p_a} \cap \ppath{\dlpt{b}_b p_L} = \emptyset$
\end{claim}

\begin{proof}
We prove the first claim; the proof of the second one is analogous.

    By contradiction, assume that $\ppath{p_1 \dlpt{t}_a}$ and $\ppath{p_b \dlpt{t}_L}$ intersect.
    By~\Cref{clm:uncross},
    \[ (1-\varepsilon)R + \dist(p_1, \dlpt{t}_a) \geq \dist(p_b, \dlpt{t}_L) + \dist(p_1, \dlpt{t}_a) \geq \dist(p_b, \dlpt{t}_a) + \dist(p_1, \dlpt{t}_L) \geq R + \dist(p_1, \dlpt{t}_L). \]
    However, from the definition of an equidistant structure,
    $\Udist(p_1, \dlpt{t}_a) = \Udist(p_1, \dlpt{t}_L)$, which implies
    $|\dist(p_1, \dlpt{t}_a) - \dist(p_1, \dlpt{t}_L)| \leq \delta$, a contradiction.
\end{proof}

\subsubsection{Ordering ribs}

Recall that the double ladder $\DL$ is assumed to be equidistant. 
Recall also that all paths $\ppath{p_1 \dlpt{t}_a}$ for $1 < a \leq L$ form a shortest path tree
rooted at $p_1$, with all terminals $\dlpt{t}_a$ having degree $1$ in $G$, and hence this tree
imposes a cyclic (say, clockwise) order on the indices $\{2,3,\ldots,L\}$. 
An analogous statement holds for all paths $\ppath{\dlpt{b}_1 p_a}$ for $1 < a \leq L$. 
The fact that the ribs $\gamma_a$ do not cross transversally and may only share a prefix or a suffix imply that these orders, restricted to $\{2,3,\ldots,L-1\}$, are reversions of each
other. 

The classic Erd\H{o}s-Szekeres theorem allows us to find a large monotonous subsequence in this order.

\begin{theorem}[Erd\H{o}s-Szekeres~\cite{ES35}]
    For $r, s\in \mathbb{Z}_{>0}$, any sequence of distinct real numbers with length at least $(r-1)(s-1)+1$ contains either a monotonically increasing subsequence of length $r$ or a monotonically decreasing subsequence of length $s$. 
\end{theorem}

We say that a subset $I \subseteq \{2,\ldots,L-1\}$ of size $|I| \geq 3$ is \emph{ordered} 
if the indices of $I$ appear either in increasing or in decreasing order in the clockwise 
order of paths $\ppath{p_1 \dlpt{t}_a}$ around $p_1$ (equivalently, paths $\ppath{\dlpt{b}_1 p_a}$ around $\dlpt{b}_1$)
and, furthermore, if $i$ and $j$ are the minimum and maximum elements of $I$, respectively,
then the terminals $\dlpt{t}_L$ and $\dlpt{b}_1$ lie in the area enclosed by the concatenation of $\gamma_i$ and $\gamma_j$ that does not contain $\gamma_a$ for $a \in I \setminus \{i,j\}$. 

\begin{claim}\label{cl:get-ordered}
There exists an ordered set $I \subseteq \{2,\ldots,L-1\}$ of size $\Omega(\sqrt{L})$. 
\end{claim}
\begin{proof}
    By the Erd\H{o}s-Szekeres theorem, there exists a subset $I_0 \subseteq \{2,\ldots,L-1\}$
    of size at least $\sqrt{L-2}$ such that the indices of $I_0$ lie monotonously in the cyclic order of paths $\ppath{p_1 T_a}$ around $p_1$ and in the reversed cyclic order of paths $\ppath{T_1 p_a}$ around $T_1$.
    This cyclic order is split into three parts by (1) the position of $\ppath{p_1 T_L}$ in 
    the order around $p_1$, (2) the position of $\ppath{B_1 p_L}$ in the order around $p_L$,
    and (3) the place between the maximum and the minimum element of $I_0$. By choosing
    the largest of these three parts we get the desired set $I$ of size at least
    $\frac{1}{3}\sqrt{L-2}$. 
\end{proof}

\Cref{cl:get-ordered} allows us to focus on ordered sets. Somewhat abusing the notation,
in what follows we will assume that the entire $\{2,\ldots,L-1\}$ is ordered
(and aim for an $\Oh(\varepsilon^{-1})$ bound on $L$).

\subsubsection{Constructing the final contradiction}

\begin{figure}[ht]
    \centering
    \begin{subfigure}{0.45\textwidth}
        \centering
        \includegraphics[width=1\linewidth]{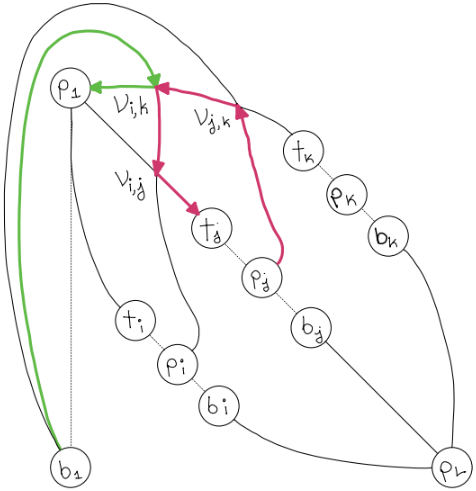}
        \caption{Type T Key Construction}
        \label{fig:typeT-cont}
    \end{subfigure}
    \begin{subfigure}{0.45\textwidth}
        \centering
        \includegraphics[width=1\linewidth]{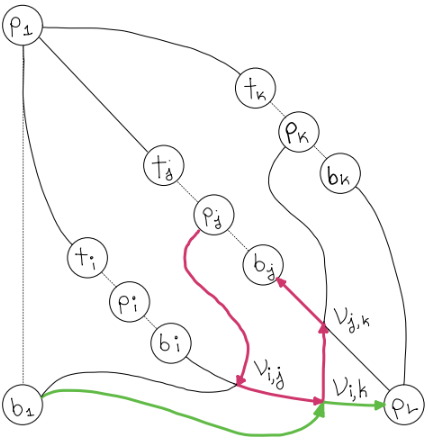}
        \caption{Type B Key Construction}
        \label{fig:DltypeB-cont}
    \end{subfigure}
    \caption{Rearrangements which lead to a bounds on $|I_T|$, $|I_B|$ similarly to the end of the comatching analysis. (a) $i < j < k$ are three consecutive indices of $I_T$. We construct paths $P_j$ (thick pink lines) and $Q_j$ (thick green lines) in~\Cref{cl:bound-typeT} following the direction of the arrows in the picture from $p_j$, $p_1$ to $\dlpt{t}_j$, $\dlpt{b}_1$ respectively, then combining the traversed path segments. Clearly the picture will not be much different if consecutive intersections $v_{a, b}$ are equal. (b) Here $i < j < k$ are three consecutive indices of $I_B$ instead. Symmetrically to the Type T case, for~\Cref{cl:bound-typeB} we construct paths $P_j$ (pink) and $Q_j$ (green) following the direction of the arrows in the picture from $p_j$, $\dlpt{b}_1$ to $\dlpt{b}_j$, $p_L$ respectively, and concatenating each visited path segment.}   
    \label{fig:DL-final-contradiction-drawn}
\end{figure}

The final contradiction will be obtained by looking at how the paths 
$\ppath{\dlpt{b}_1 p_a}$ for $1 < a < L$ behave. 

Fix $1 < a < L$. Observe the following.
\begin{itemize}
    \item For every $1 < b  \leq a$, $\ppath{\dlpt{b}_1 p_a}$ does not intersect
    $\ppath{p_1 \dlpt{t}_b}$ due to~\Cref{clm:kcenter-touch-topdown}.
    \item For every $a \leq b < L$, $\ppath{\dlpt{b}_1 p_a}$ does not intersect
    $\ppath{p_L \dlpt{b}_b}$ due to~\Cref{clm:kcenter-TT-or-BB}.
\end{itemize}
Because the instance is ordered, for every $b_1,b_2$ with $1 < b_1 < a < b_2 < L$, 
the concatenation of $\gamma_{b_1}$ and $\gamma_{b_2}$ is a closed loop that separates
$\dlpt{b}_1$ from $p_a$. Hence, $\ppath{\dlpt{b}_1 p_a}$ intersects either $\ppath{p_L \dlpt{b}_{b_1}}$
or $\ppath{p_1 \dlpt{t}_{b_2}}$. 
Furthermore, because of the ordering, one of the following options hold:
\begin{description}
    \item[Type T:] 
    $\ppath{\dlpt{b}_1 p_a}$ intersects transversally every path $\ppath{p_1 \dlpt{t}_b}$ for
    $a < b < L$, in the decreasing order of indices $b$ if one goes from $\dlpt{b}_1$ to $p_a$, 
    and does not intersect transversally any path $\ppath{p_L \dlpt{b}_b}$ for $1 < b < a$.
    \item[Type B:]
    $\ppath{\dlpt{b}_1 p_a}$ intersects transversally every path $\ppath{p_L \dlpt{b}_b}$ for
    $1 < b < a$, in the increasing order of indices $b$ if one goes from $\dlpt{b}_1$ to $p_a$,
    and does not intersect transversally any path $\ppath{p_1 \dlpt{t}_b}$ for $a < b < L$.
\end{description}
Let $I_T$ and $I_B$ be the set of indices of $\{2,\ldots,L-1\}$ for which option
$T$ or $B$ holds, respectively.
We now separately analyze $I_T$ and $I_B$; the analyses are similar, but not exactly the same.

\paragraph{Analysis of Type T.}
For two indices $a,b \in I_T$, $a < b$, let $v_{a,b}$ be any vertex of the intersection
of $\ppath{\dlpt{b}_1 p_a}$ and $\ppath{p_1 \dlpt{t}_b}$. 
We have already established that for fixed $a \in I_T$, the vertices $v_{a,b}$
for $b \in I_T$, $b > a$, lie on $\ppath{\dlpt{b}_1 p_a}$ in the decreasing order, if one
goes from $\dlpt{b}_1$ to $p_a$ (two consecutive vertices $v_{a,b}$ may be equal). 
We observe that the same holds along $\ppath{p_1 \dlpt{t}_b}$.

\begin{claim}\label{cl:order-typeT}
For every $b \in I_T$, the vertices $v_{a,b}$ for $a \in I_T$, $a < b$, are ordered
in the increasing order on $\ppath{p_1 \dlpt{t}_b}$ if one traverses this path from $p_1$ to $\dlpt{t}_b$
(two consecutive vertices $v_{a,b}$ may be equal). 
\end{claim}
\begin{proof}
    Assume the contrary: for some $a_1 < a_2 < b$, $a_1,a_2 \in I_T$, the vertex
    $v_{a_1,b}$ is closer to $\dlpt{t}_b$ on $\ppath{p_1 \dlpt{t}_b}$ than $v_{a_2,b}$.
    Then, the path $\ppath{\dlpt{b}_1 p_{a_1}}$ crosses at $v_{a_1,b}$ the closed loop $\gamma$
    without self-intersections consisting of 
    the concatenation of (i) part of $\gamma_{b}$ from $p_L$ to $v_{a_2,b}$,
    (ii) subpath of $\ppath{\dlpt{b}_1 p_{a_2}}$ from $v_{a_2,b}$ to $p_{a_2}$, 
    (iii) $\Gamma(\sigma_{a_2})$ from $p_{a_2}$ to $\dlpt{b}_{a_2}$; (iv)
    $\ppath{\dlpt{b}_{a_2} p_L}$. 
    Because of the ordering, $\gamma$ separates vertices $p_a$ for $L > a > a_2$ from
    vertices $p_a$ for $1 < a < a_2$, and $a_1$ is on the other side. This is a contradiction.
\end{proof}

This allows us to obtain a bound on $|I_T|$ similarly as in the end of the comatching 
analysis. (See also~\Cref{fig:typeT-cont}.)

\begin{claim} \label{cl:bound-typeT}
    $|I_T|=\Oh(\varepsilon^{-1})$.
\end{claim}
\begin{proof}
    Let $i < j < k$ be three consecutive indices of $I_T$. 
    We construct a path $P_j$ connecting $p_j$ and $\dlpt{t}_j$ as follows:
    start in $p_j$, continue along $\ppath{p_j \dlpt{b}_1}$ to $v_{j,k}$
    then along $\ppath{p_1 \dlpt{t}_k}$ in the direction of $p_1$ to $v_{i,k}$, then 
    along $\ppath{\dlpt{b}_1 p_i}$ in the direction of $p_i$ to $v_{i,j}$, 
    and finally along $\ppath{p_1 \dlpt{t}_j}$ to $\dlpt{t}_j$.
    We construct a path $Q_j$ connecting $p_1$ and $\dlpt{b}_1$ as follows:
    start in $\dlpt{b}_1$, continue along $\ppath{\dlpt{b}_1 p_i}$ to $v_{i,k}$,
    and then along $\ppath{p_1 \dlpt{t}_k}$ to $p_1$.
    There are $|I_T|-2$ paths $P_j$ and $|I_T|-2$ paths $Q_j$, and each
    path $P_j$ or $Q_j$ is of length at least $R$.
    
    In the rearrangement we only used edges
    of paths $\ppath{p_1 \dlpt{t}_i}$ and $\ppath{\dlpt{b}_1 p_i}$ for $i \in I_T$
    and we do not use the same edge twice. Each of these paths
    is of length at most $(1-\varepsilon)R$. Hence,
    \[ 2|I_T| \cdot (1-\varepsilon)R \geq (2|I_T|-4) \cdot R.\]
    This implies $|I_T| \leq 2\varepsilon^{-1}$, as desired.
\end{proof}

\paragraph{Analysis of Type B.}
For two indices $a,b \in I_B$, $b < a$, let $v_{a,b}$ be any vertex of the intersection
of $\ppath{\dlpt{b}_1 p_a}$ and $\ppath{p_L \dlpt{b}_b}$. 
We have already established that for fixed $a \in I_B$, the vertices $v_{a,b}$
for $b \in I_B$, $b < a$, lie on $\ppath{\dlpt{b}_1 p_a}$ in the increasing order, if one
goes from $\dlpt{b}_1$ to $p_a$ (two consecutive vertices $v_{a,b}$ may be equal). 
We observe that the same holds along $\ppath{p_L \dlpt{b}_b}$.

\begin{claim}\label{cl:order-typeB}
For every $b \in I_B$, the vertices $v_{a,b}$ for $a \in I_B$, $a > b$, are ordered
in the decreasing order on $\ppath{p_L \dlpt{b}_b}$ if one traverses this path from $p_L$ to $\dlpt{b}_b$
(two consecutive vertices $v_{a,b}$ may be equal). 
\end{claim}
\begin{proof}
    Assume the contrary: for some $b < a_1 < a_2$, $a_1,a_2 \in I_B$, the vertex
    $v_{a_1,b}$ is closer to $p_L$ on $\ppath{p_L \dlpt{b}_b}$ than $v_{a_2,b}$.
    Then, the path $\ppath{\dlpt{b}_1 p_{a_2}}$ crosses at $v_{a_2,b}$ the closed loop $\gamma$
    without self-intersections consisting of 
    the concatenation of (i) part of $\gamma_{b}$ from $p_1$ to $v_{a_1,b}$,
    (ii) subpath of $\ppath{\dlpt{b}_1 p_{a_1}}$ from $v_{a_1,b}$ to $p_{a_1}$, 
    (iii) $\Gamma(\sigma_{a_1})$ from $p_{a_1}$ to $\dlpt{t}_{a_1}$; (iv)
    $\ppath{p_1 \dlpt{t}_{a_1}}$. 
    Because of the ordering, $\gamma$ separates vertices $p_a$ for $L > a > a_1$ from
    vertices $p_a$ for $1 < a < a_1$, and $a_2$ is on the other side. This is a contradiction.
\end{proof}

This allows us to obtain a bound on $|I_B|$ similarly as in the end of the comatching 
analysis. (See also~\Cref{fig:DltypeB-cont}.)

\begin{claim} \label{cl:bound-typeB}
    $|I_B|=\Oh(\varepsilon^{-1})$.
\end{claim}
\begin{proof}
    Let $i < j < k$ be three consecutive indices of $I_B$. 
    We construct a path $P_j$ connecting $p_j$ and $\dlpt{b}_j$ as follows:
    start in $p_j$, continue along $\ppath{p_j \dlpt{b}_1}$ to $v_{i,j}$
    then along $\ppath{p_L \dlpt{b}_i}$ in the direction of $p_L$ to $v_{i,k}$, then 
    along $\ppath{\dlpt{b}_1 p_k}$ in the direction of $p_k$ to $v_{j,k}$, 
    and finally along $\ppath{p_L \dlpt{b}_j}$ to $\dlpt{b}_j$.
    We construct a path $Q_j$ connecting $\dlpt{b}_1$ and $p_L$ as follows:
    start in $\dlpt{b}_1$, continue along $\ppath{\dlpt{b}_1 p_i}$ to $v_{i,k}$,
    and then along $\ppath{p_L \dlpt{b}_k}$ to $p_L$.
    There are $|I_B|-2$ paths $P_j$ and $|I_B|-2$ paths $Q_j$.
    Each path $P_j$ is of length at least $R$, while
    each path $Q_j$ is of length at least $\dist(\dlpt{b}_1,p_L)$. 
    Furthermore, as the instance is equidistant, we
    have $|\dist(\dlpt{b}_1,p_L) - \dist(\dlpt{b}_1, p_j)| \leq \delta$ for every $j \in I_B$. 
    
    In the rearrangement we only used edges
    of paths $\ppath{p_L \dlpt{b}_i}$ and $\ppath{\dlpt{b}_1 p_i}$ for $i \in I_B$
    and we do not use the same edge twice. 
    Hence,
    \[ \sum_{i \in I_B} \big( \dist(\dlpt{b}_1, p_i) + \dist(p_L, \dlpt{b}_i) \big) \geq (|I_B|-2)R + (|I_B|-2)\dist(\dlpt{b}_1, p_L). \]
    Each of the distances in the left hand side above is at most $(1-\varepsilon)R$, and the length of the paths $\ppath{\dlpt{b}_1 p_i}$ do not differ
    from $\dist(\dlpt{b}_1,p_L)$ by more than $\delta$.
    Hence,
    \[ (|I_B|+2)(1-\varepsilon)R + (|I_B|-2)\delta \geq (|I_B|-2)R. \]
    This implies $|I_B| = \Oh(\varepsilon^{-1})$, as desired.
\end{proof}

This concludes the proof of~\Cref{lem:DL-noncrossing-bound}.

%% file: OtherApps.tex
\section{Polygon with holes}\label{subsec:other}

Here we show a polynomial coreset bound for furthest neighbors of points in a polygon with holes. This follows from \Cref{lm:poly-holes} below and \Cref{thm:1center}. 

\begin{lemma}\label{lm:poly-holes} Let $f(\eps)$ be the maximum size of $\eps$-coresets for furthest neighbors in planar metrics for any $\eps\in (0,1)$. Then we can construct an $\eps$-coreset of size $f(\eps/8)$ for furthest neighbors of any point set in a polygon with holes.
\end{lemma}
\begin{proof} Let $\mathcal{P}$ be a polygon with holes  on the plane, and $P$ be a set of points in $\mathcal{P}$. Let $\dist_{\mathcal{P}}$ be the geodesic distance function between points in $\mathcal{P}$.  Let $\Delta$ be the diameter of $P$.  Let $P^+$ be the set of points that are within distance $\Delta/\eps$ from $P$.  Observe that for any point $q \in V\setminus P^+$, any point $p\in P$ is an $(1-\eps)$-approximate furthest neighbor of $q$. Thus, as long as the coreset is non-empty, points in $V\setminus P^+$ are taken care of. 

Herein, we focus on points in $P^+$.  Let $Q$ be an $(\eps\Delta/8 )$-net of $P^+$. That is, every point $x\in P^+$ has a point $y\in Q$ such that $\dist_{\mathcal{P}}(x,y)\leq \eps\Delta/8$ and every two points $x,y \in Q$ has $\dist_{\mathcal{P}}(x,y) > \eps \Delta/8$.  We say that a subset $C\subseteq P$ is a furthest $(\eps/8)$-coreset for $P$ \EMPH{restricted to $Q$} if for every $q\in Q$, there exists a point $p\in C$ such that:
\begin{equation*}
     \dist_{\mathcal{P}}(p,q) \geq (1-\eps) \dist_{\mathcal{P}}(p^*,q)
\end{equation*}
where $p^*$ is the furthest neighbor of $q$ in $p$. That is, $C$ only preserves the furthest neighbors of points in $Q$. The following claim shows that it suffices to restrict to $Q$; the proof will be delayed until later.

\begin{claim}\label{clm:coreset-extended} Let $C$ be a furthest neighbor $(\eps/8)$-coreset of $P$ restricted to $Q$. Then $C$ is a furthest neighbor $\eps$-coreset for $P$ restricted to $P^+$. 
\end{claim}

Let $G = (V,E,w)$ be an edge-weighted planar graph such that (i)  the vertex set satisfies $Q\cup P\subseteq V\subseteq \mathcal{P}$, (ii) $\dist_G(x,y) = \dist_{\mathcal{P}}(x,y)$ for all $x,y\in P\cup Q$, and (iii) for every edge $E = (u,v)$, $w(u,v) = \dist_{\mathcal{P}}(u,v)$. Such a graph $G$ exists~\cite{BKKLLPT25}. (The basic idea to construct $G$ is to take the geodesics connecting all pairs of points in $Q$, and place a vertex at the intersection of any two geodesics; the graph is planar by the planarity of the polygon.)
 Next we construct a furthest neighbor $(\eps/8)$-coreset $C\subseteq P$ of size $f(\eps/8)$  for the shortest path metrics of $G$. Since $G$ preserves pairwise distances between points in $P\cup Q$,  $C$ is an  $\eps$-coreset for $P$ restricted to $Q$.  Thus, by \Cref{clm:coreset-extended}, $C$ is a  furthest neighbor $\eps$-coreset for $P$ restricted to $P^+$ and hence the whole domain $\mathcal{P}$.  
\end{proof}

Now we show \Cref{clm:coreset-extended}, which is simply an application of the triangle inequality.

\begin{proof}[Proof of \Cref{clm:coreset-extended}] Let $x$ be a point in $P^+$ and $y$ be a point in $Q$ such that  $\dist_{\mathcal{P}}(x,y)\leq \eps\Delta/8$; $y$ exists by the definition of $Q$. Let $x^*$ and $y^*$ be the furthest neighbors of $x$ and $y$ in $P$. Since the diameter of $P$ is at most $\Delta$, by triangle inequality:
\begin{equation*}
    \dist_{\mathcal{P}}(x,x^*), \dist_{\mathcal{P}}(y,y^*)\geq \Delta/2
\end{equation*}

Since $C$ is an $(\eps/8)$-coreset restricted to $Q$, there exists a point $\hat{y}\in C$ such that:
\begin{equation*}
    \dist_{\mathcal{P}}(y,\hat{y}) \geq (1-\eps/8) \dist_{\mathcal{P}}(y,y^*)
\end{equation*}

Our goal is to show that $\hat{y}$ is a $(1-\eps)$-approximate furthest neighbor of $x$, as follows:
\begin{equation}
    \begin{split}
         \dist_{\mathcal{P}}(x,\hat{y})  &\geq \dist_{\mathcal{P}}(y,\hat{y}) - \dist_{\mathcal{P}}(x,y)\qquad \text{(triangle inequality)}\\
         &\geq (1-\eps/8)\dist_{\mathcal{P}}(y,y^*) - \dist_{\mathcal{P}}(x,y)\qquad \text{(by definition of $\hat{y}$)}\\
         &\geq (1-\eps/8)\dist_{\mathcal{P}}(y,x^*) - \dist_{\mathcal{P}}(x,y)\qquad \text{(since $y^*$ is furthest neighbor of $y$)}\\
         &\geq (1-\eps/8)(\dist_{\mathcal{P}}(x,x^*) - \dist_{\mathcal{P}}(x,y)) - \dist_{\mathcal{P}}(x,y)\qquad \text{(triangle inequality)}\\
         &\geq (1-\eps/8)\dist_{\mathcal{P}}(x,x^*) -  2\dist_{\mathcal{P}}(x,y)\\
         &\geq (1-\eps/8)\dist_{\mathcal{P}}(x,x^*) - \eps\Delta/4  \qquad \text{(since $\dist_{\mathcal{P}}(x,y)\leq \eps\Delta/8$)}\\
          &\geq (1-\eps/8)\dist_{\mathcal{P}}(x,x^*) - (\eps/2)\dist_{\mathcal{P}}(x,x^*)   \qquad \text{(since $ \dist_{\mathcal{P}}(x,x^*)\geq \Delta/2$)}\\
          &\geq (1-\eps)\dist_{\mathcal{P}}(x,x^*)~,
     \end{split}
\end{equation}
as desired.
\end{proof}

%% file: genus.tex
\section{Bounded genus graphs}\label{sec:genus}

The upper bounds for metrics induced by graphs embeddable in surfaces of bounded Euler genus will follow from a reduction to planar metrics via the use of \emph{cut-graphs}. We will reuse below some of the formalism defined for a similar argument in~\cite{DBLP:conf/icalp/KlukPPS25}. 

For a surface $\Sigma$, a graph $H$ embedded in $\Sigma$ is its \EMPH{simple cut-graph} if it has a single face homeomorphic to an open disc. In particular, the surface $\Sigma \cut H$ obtained from $\Sigma$ by cutting along the embedding of the edges of $H$ is a disk.

Let $G$ be a graph embedded in $\Sigma$ and let $H$ be its subgraph that is a simple cut-graph of $\Sigma$.
We define $G \cut H$ as the graph embedded in $\Sigma \cut H$ obtained from $G$ as follows.
First, let $\sigma$ be the (unique) facial walk of $H$. Each edge $e$ of $H$ is contained exactly twice in $\sigma$ and each vertex $v$ of $H$ is contained in $\sigma$ as many times as the degree of $v$ in $H$.
To obtain $G \cut H$, we replace $H$ with a simple cycle $C_\sigma$ whose vertex set is the set of copies of vertices of $H$ and whose edge set is the set of copies of edges of $H$ defined the obvious way. Notice that $\sigma$ also prescribes for every edge $uv$ of $G$ between a vertex $u\in V(G)\setminus V(H)$ and a vertex $v\in V(H)$, to which copy of $v$ in $G \cut H$ the vertex $u$ should be adjacent to in $G \cut H$.

The key component of our reduction is the following well-known result; see, e.g.~\cite{BorradaileDT14,CabelloCL12algo,EricksonW05}.
\begin{lemma}\label{lem:cutgraph}
	For every integer $k \geq 1$ and for every edge-weighted connected graph $G$ embedded on a surface $\Sigma$ of Euler genus at most $g$ and every vertex $u \in V(G)$, there is a subgraph $H$ of $G$ with the following properties:
	\begin{itemize}
		\item $H$ is a simple cut-graph of $\Sigma$, and
		\item $V(H)$ is the union of the vertex sets of $\Oh(g)$ shortest paths in $G$ that have $u$ as a common endpoint.
	\end{itemize}
    In particular, if $G$ has diameter $\Delta$, then the set of copies of the vertices of $H$ in $G \cut H$ can be covered with $\Oh(g)$ paths of length at most $\Delta$ in $G \cut H$.
\end{lemma}

The upper bound proofs for $\eps$-comatchings and $\eps$-double ladders follow a very similar structure; hence, we will define the following notion of an $\eps$-double semi-ladder to capture both structures.

\begin{definition}[$\varepsilon$-double semi-ladder]
Let $\DL = \{\dltri{1}, \dltri{2},\ldots, \dltri{\ell}\}\subseteq V^3$  be a sequence of $\ell$ \emph{ordered triplets} of points in a metric $(V,\dist)$, and $\eps\in (0,1)$ be a parameter.  We say that $\DL$ is an \EMPH{$\eps$-double semi-ladder} if there exists $R > 0$ such that:
    \begin{enumerate}
        \item for every $1 \leq i < j \leq \ell$,
        \[ \dist(p_i, \dlpt{t}_j) \leq (1-\varepsilon)R \quad \mathrm{and} \quad \dist(p_j, \dlpt{b}_i) \leq (1-\varepsilon)R.\]
        \item for every $1 \leq i \leq \ell$,
        \[ \dist(p_i, \dlpt{t}_i) > R \quad \mathrm{and} \quad \dist(p_i, \dlpt{b}_i) > R;\]
    \end{enumerate}
An $\varepsilon$-double semi-ladder instance is a pair $(G, \DL)$ where $G$ is a graph and $\DL$ is an $\eps$-double semi-ladder in the graph metric induced by $G$.
\end{definition}
Note that this notion generalizes both $\eps$-comatchings and $\eps$-double ladders: the former can be obtained by putting $\dlpt{b}_i = \dlpt{t}_i = q_i$ and the latter is a special case where we also require the distances $\dist(p_j, \dlpt{t}_i), \dist(p_i, \dlpt{b}_j)$ for $1 \leq i < j \leq \ell$ to be large.
The reduction itself is realized by the following lemma.

\begin{lemma}\label{lem:genus}
    Fix any integer $g \geq 0$ and $\eps \in (0, 1)$. Let $(G, \DL = \{\dltri{i}\}_{i \in [L]})$ be an $\eps$-double semi-ladder instance such that $G$ is embeddable in a surface of Euler genus at most $g$. Then, there exists a set of indices $j_1, \dots, j_{L'} \in [L]$, a planar graph $G'$ and a map $\eta : V(G) \to V(G')$ such that:
    \begin{itemize}
        \item $\dist_{G'}(\eta(p_{j_x}), \eta(\dlpt{t}_{j_y})) \geq \dist_G(p_{j_x}, \dlpt{t}_{j_y})$ and $\dist_{G'}(\eta(p_{j_x}), \eta(\dlpt{b}_{j_y})) \geq \dist_G(p_{j_x}, \dlpt{b}_{j_y})$ for $x, y \in [L']$,
        \item $\dist_{G'}(\eta(p_{j_x}), \eta(\dlpt{t}_{j_y})) = \dist_G(p_{j_x}, \dlpt{t}_{j_y})$ for $1 \leq x < y \leq L'$,
        \item $\dist_{G'}(\eta(p_{j_y}), \eta(\dlpt{b}_{j_x})) = \dist_G(p_{j_y}, \dlpt{b}_{j_x})$ for $1 \leq x < y \leq L'$,
        \item $L = \Oh(g^4 \eps^{-8}) \cdot L'$.
    \end{itemize}
\end{lemma}

\begin{proof}
	Let
	$$
		R := \min_{i \in [L]} \min(\dist(p_i, \dlpt{b}_i), \dist(p_i, \dlpt{t}_i))
	$$
	and set $\delta := \varepsilon R / 100$. W.l.o.g. we can assume that every edge of $G$ lies on either a path $\ppath{p_i \dlpt{t}_j}$ or $\ppath{p_j, \dlpt{b}_i}$ for some $1 \leq i < j \leq L$, and so by the same argument as in \Cref{diam-lnci-DL}, that the diameter of $G$ is at most $8R$.

	Let $u \in V(G)$ be an arbitrary vertex of $V(G)$. Let $H$ be a subgraph of $G$ obtained by applying \Cref{lem:cutgraph} to $G$ with an endpoint $u$. For $v \in V(G)$, let $\mu(v)$ denote all vertices of $G \cut H$ into which $v$ is split after performing the cut. In particular, $\mu(v) = \{v\}$ if $v \not\in V(H)$. For any $u, v \in V(G)$ and $u' \in \mu(u), v' \in \mu(v)$, we have $\dist_{G \cut H}(u', v') \geq \dist_G(u, v)$, that is, the cutting procedure can only increase the distances between the vertices. Moreover, if $\dist_{G \cut H}(u', v') > \dist_G(u, v)$, then the shortest path from $u$ to $v$ in $G$ must cross $V(H)$. From each $\mu(v)$ we select one vertex arbitrarily and call it $\eta(v)$.
	
	Let $\DL'$ be a projection of $\DL$ onto $G \cut H$, that is, $\DL' = \{(\eta(p_i), \eta(\dlpt{b}_i), \eta(\dlpt{t}_i)) : i \in [L]\}$. Note that all pairs $(\eta(p_i), \eta(\dlpt{b}_i))$ and $(\eta(p_i), \eta(\dlpt{t}_i))$ are still at distance at least $R$, however, some pairs $(\eta(p_i), \eta(\dlpt{t}_j))$ or $(\eta(p_j), \eta(\dlpt{b}_i))$ close in $\DL$ possibly also ended up far apart after cutting.

	Let $X = \bigcup_{v \in V(H)} \mu(v)$. Since $V(H)$ was contained in a union of $\Oh(g)$ shortest paths in $G$, which are of length at most $8R$, the same holds for $X$ in $G \cut H$. In particular, there is a set $Z$ of size $\Oh(g \varepsilon^{-1})$ such that every $x \in X$ is at distance at most $\delta$ in $G \cut H$ from some $z \in Z$. Group triples $\dltri{} \in \DL'$ by the \emph{rounded distance profile} of $p$ to $Z$ in $G \cut H$. Since $G \cut H$ is planar, by \Cref{lem:Uprof-vc-DL} we get that the number of such groups is of $\Oh(g^4 \varepsilon^{-8})$. In particular, there is a subset $\DL_0 \subseteq \DL$ such that vertices $p_i$ over all $\dltri{i} \in \DL_0$ share the rounded distance profile to $X$, and such that $L = \Oh(|\DL_0| \cdot g^4 \varepsilon^{-8})$.

    Let $j_1, \dots, j_{L'}$ be the indices of triples $(\eta(p_{j_i}), \eta(\dlpt{b}_{j_i}), \eta(\dlpt{t}_{j_i}))$ in $\DL_0$. We want to argue, that indices $(j_i)_{i \in [L']}$, graph $G \cut H$ and the map $\eta$ satisfy the lemma statement.
    The condition $\dist_{G \cut H}(\eta(v), \eta(u)) \geq \dist_G(v, u)$ follows immediately for all $v, u \in V(G)$ and $L = \Oh(g^4 \eps^{-8} L')$ follows from the way we chose $\DL_0$. Assume by contradiction that we have two indices $j_x < j_y$ such that $\dist_{G \cut H}(\eta(p_{j_x}), \eta(\dlpt{t}_{j_y}))$ or $\dist_{G \cut H}(\eta(p_{j_y}), \eta(\dlpt{b}_{j_x}))$ is larger than $(1 - \eps)R$. W.l.o.g. assume the former. We have $\dist_G(p_{j_x}, \dlpt{t}_{j_y}) \leq (1 - \eps)R$, hence the path $\ppath{p_{j_x} \dlpt{t}_{j_y}}$ crosses $V(H)$.
	
	Let $u$ be the first vertex of $V(H)$ on a path from $p_{j_x}$ to $\dlpt{t}_{j_y}$. In particular, there is $u^* \in \mu(u) \subseteq X$ such that $\dist_G(p_{j_x}, u) = \dist_{G \cut H} (\eta(p_{j_x}), u^*)$. Let $z \in Z$ be the vertex at distance at most $\delta$ in $G \cut H$ from $u^*$. We have
	$$
		\dist_G(p_{j_y}, u) \leq
		\dist_{G \cut H}(\eta(p_{j_y}), u^*) \leq
		\dist_{G \cut H}(\eta(p_{j_y}), z) + \delta \leq
    $$
    $$ \leq
		\dist_{G \cut H}(p_{j_x}, z) + 2\delta \leq
		\dist_{G \cut H}(p_{j_x}, u^*) + 3\delta =
		\dist_G(p_{j_x}, u) + 3\delta,
	$$
    where all inequalities follow from the triangle inequality plus the fact that $\eta(p_{j_x}), \eta(p_{j_y})$ share the rounded distance to $z$.
    Thus,
    $$
		\dist_G(p_{j_y}, \dlpt{t}_{j_y}) \leq
		\dist_G(p_{j_y}, u) + \dist_G(p_{j_x}, \dlpt{t}_{j_y}) - \dist_G(p_{j_x}, u) \leq
		\dist_G(p_{j_x}, u) + 3\delta + (1 - \varepsilon)R - \dist_G(p_{j_x}, u) < R,
    $$
	which contradicts that $\DL$ is an $\eps$-double ladder in $G$.

    Therefore, our choice of $j_i$ and $\nu$ satisfies the lemma statement, which finishes the proof.
\end{proof}

\begin{proof}[Proof of \Cref{thm:genus}]
    Let $\{(p_i, q_i)\}_{i \in [L]}$ be an $\eps$-comatching in $G$ of Euler genus at most $g$. Let
	$
		R = \min_{i \in [L]} \dist(p_i, q_i)
	$.
    The set $\{(p_i, q_i, q_i)\}_{i \in [L]}$ is an $\eps$-double semi-ladder in $G$. Take a planar graph $G'$, indices $j_1, \dots, j_{L'} \in [L]$ and map $\eta : V(G) \to V(G')$ given by \Cref{lem:genus} applied to $(G, \{(p_i, q_i, q_i)\}_{i \in [L]})$. We have:
    \begin{itemize}
        \item $\dist_{G'}(\eta(p_{j_x}), \eta(q_{j_x})) \geq \dist_G(p_x, q_x) > R$ for $x \in [L']$, and
        \item $\dist_{G'}(\eta(p_{j_x}), \eta(q_{j_y})) = \dist_G(p_x, q_y) \leq (1 - \eps)R$ for $x,y \in [L']$, $x \neq y$,
    \end{itemize}
    hence $\{(\eta(p_{j_i}), \eta(q_{j_i}))\}_{i \in [L']}$ is an $\eps$-comatching in $G'$ and by \Cref{thm:comatching-bound} is of size at most $L_{\mathrm{planar}}(\eps^{-1})$. Thus, $L \leq \Oh(g^4 \varepsilon^{-8}) \cdot L_{\mathrm{planar}}(\eps^{-1})$ which finishes the proof.
\end{proof}

\begin{proof}[Proof of \Cref{thm:genus2}]
    Let $\DL = \{\dltri{i}\}_{i \in [L]}$ be an $\eps$-double ladder in $G$ of Euler genus at most $g$. Let
	$
		R = \min_{i \in [L]} \min(\dist(p_i, \dlpt{b}_i), \dist(p_i, \dlpt{t}_i))
	$.
    Take a planar graph $G'$, indices $j_1, \dots, j_{L'} \in [L]$ and map $\eta : V(G) \to V(G')$ given by \Cref{lem:genus} applied to $(G, \DL)$. We have:
    \begin{itemize}
        \item $\dist_{G'}(\eta(p_{j_x}), \eta(\dlpt{t}_{j_y})) = \dist_G(p_{j_x}, \dlpt{t}_{j_y}) \leq (1 - \eps)R$ for $1 \leq x < y \leq L'$,
        \item $\dist_{G'}(\eta(p_{j_y}), \eta(\dlpt{b}_{j_x})) = \dist_G(p_{j_y}, \dlpt{b}_{j_x}) \leq (1 - \eps)R$ for $1 \leq x < y \leq L'$,
        \item $\dist_{G'}(\eta(p_{j_y}), \eta(\dlpt{t}_{j_x})) \geq \dist_G(p_{j_y}, \dlpt{t}_{j_x}) > R$ for $1 \leq x \leq y \leq L'$,
        \item $\dist_{G'}(\eta(p_{j_x}), \eta(\dlpt{b}_{j_y})) \geq \dist_G(p_{j_x}, \dlpt{b}_{j_y}) > R$ for $1 \leq x \leq y \leq L'$,
    \end{itemize}
    hence $\{(\eta(p_{j_i}), \eta(\dlpt{b}_{j_i}), \eta(\dlpt{t}_{j_i}))\}_{i \in [L']}$ is an $\eps$-double ladder in $G'$, hence by \Cref{thm:doubleladder-bound} is of size at most $L'_{\mathrm{planar}}(\eps^{-1})$. Thus, $L \leq \Oh(g^4 \varepsilon^{-8}) \cdot L'_{\mathrm{planar}}(\eps^{-1})$ which finishes the proof.
\end{proof}

%% file: lower-bound.tex
\section{Lower bounds}\label{sec:lb}

\subsection{Beyond bounded genus graphs: Proof of \Cref{thm:comatching-lb}}

Let us first restate \Cref{thm:comatching-lb}:

\SokoLowerBound*

\begin{proof}
For a nonnegative integer $k$, let us define an undirected integer-weighted graph $G_k$ in the following way. First, let us create two complete binary trees $T_1$ and $T_2$ with roots $r_1$ and $r_2$, both of height of $k$ edges and let $I : V(T_1) \to V(T_2)$ be an arbitrary isomorphisms between them. Let $S : V(T_1) \to V(T_1) \cup \{\perp\}$ be a function that for a vertex $v$ of $T_1$ returns its unique sibling (with the exception of the root $r_1$, for which $S(r_1) = \perp$). Then, in addition to the edges of the binary trees, add all edges of the form $(v, I(B(v)))$ for $v \neq r$ (which we will call \textit{matching edges}). Moreover, all edges of $T_1$ and $T_2$ have unit weight and each edge $(v, I(B(v)))$ has a weight $2 \dist(r_1, v) - 1$.  Refer to \Cref{fig:Soko} for an illustrative example of $G_3$.

We will show that $G_k$ satisfies all of the required properties. 

\begin{enumerate}

\item Let $\cal{M}$ be the set of all pairs $(l, I(l))$, where $l$ is a leaf of the upper tree $T_1$. As $T_1$ has $2^k$ leaves, this set is of size $2^k$.

\begin{claim}
    $\cal{M}$ is a $(2k-1)$-comatching.
\end{claim}
This claim will subsequently prove \Cref{it:com-lb1}.
\begin{proof}
Let us consider a pair of leaves $(l, m)$, where $l \in V(T_1)$ and $m \in V(T_2)$. We aim to prove that $\dist(l, m) \leq 2k-1 \Leftrightarrow I(l) \neq m$. An observant reader may inspect the distances from the leftmost leaf of the upper tree on \Cref{fig:Soko} and notice that their patterns easily generalize to higher values of $k$, but we will opt for a more detailed proof anyway.

Let $p : V(G_k) \to V(G_k) \cup \{\perp\}$ be a function that for a vertex returns its parent in the tree it belongs to (except for roots $r_1$ and $r_2$, where $p(r_1)=p(r_2)=\perp$).  Let us consider sequences $P_l = (l, p(l), p^2(l), \ldots, p^k(l))$ and $P_m = (m, p(m), p^2(m), \ldots, p^k(m))$. As we have that $p^k(l)=r_1$, $p^k(m)=r_2$ and $I(r_1)=r_2$, there exists a smallest index $c$ such that $I(p^c(l))=p^c(m)$. If $m \neq I(l)$, then $c>0$ and $I(p^{c-1}(l))$ is a sibling of $p^{c-1}(m)$, hence there is a matching edge between $p^{c-1}(l)$ and $p^{c-1}(m)$. That edge has length $2(k-(c-1))-1$, hence the path $l, p(l), \ldots, p^{c-1}(l), p^{c-1}(m), p^{c-2}(m), \ldots, p(m), m$ has length $2k-1$ showing that $m \neq I(l) \Rightarrow \dist(l, m) \le 2k-1$.

Let us now consider any path from $l$ to $I(l)$. It has to contain a matching edge --- let it be $(x, I(B(x)))$. 
We have that the weight of $(x, I(B(x)))$ is $2\dist(x, r_1)-1$, that $\dist(l, x) + \dist(x, r_1) \ge \dist(l, r_1)=k$, that $\dist(I(l), I(B(x))) + \dist(I(B(x)), r_2) \ge \dist(I(l), r_2) = k$, and that $\dist(x, r_1) = \dist(I(B(x)), r_2)$. Hence, the length of the considered path is at least $\dist(l, x) + 2 \dist(x, r_1)-1 + \dist(I(B(x)), I(l)) \ge (k-\dist(x, r_1)) + 2 \dist(x, r_1) - 1 + (k-\dist(I(B(x)), r_2) = 2k-1$. However, for the equality to hold, all of the considered inequalities would need to become equalities. Consequently, it would imply that $x$ is an ancestor of $l$ and that $I(B(x))$ is an ancestor of $I(l)$. But if $x$ is an ancestor of $l$, then $B(x)$ is not an ancestor of $l$ and $I(B(x))$ is not an ancestor of $I(l)$, which gives a contradiction proving that $\dist(l, I(l)) > 2k-1$, concluding that $\cal{M}$ is indeed a $(2k-1)$-comatching.
\end{proof}

\item Let us now focus on providing a tree decomposition of $G_k$ that will prove \Cref{it:com-lb2} and aid in proving \Cref{it:com-lb3} and \Cref{it:com-lb4}. 

Let $T$ denote the induced subgraph of $T_1$ on all of its non-leaves. We set our tree decomposition to be $\mathcal{T} = (T, B)$, where $B$ is a bag function such that $B(v)$ is a bag containing six vertices $v, s_1, s_2, I(v), I(s_1), I(s_2)$, where $s_1$ and $s_2$ are the children of $v$ in $T_1$. It is easy to verify that this in fact really is a tree decomposition of $G_k$. It also clearly satisfies that its bags are of size $6$ and that its adhesions are of size $2$, proving \Cref{it:com-lb2}.

\item For the proof of \Cref{it:com-lb3}, let us assume by a contradiction that $K_{3, 4}$ is a minor of $G_k$. In other words, there exist disjoint subsets $L_1, L_2, L_3, R_1, R_2, R_3$ and $R_4$ of $V(G_k)$ such that each of them induces a connected subgraph of $G_k$ and for each pair $(L_i, R_j)$ there exists an edge connecting a vertex of $L_i$ with a vertex of $R_j$. Let $T_{L_1}, \ldots, T_{R_4}$ be subgraphs of $T$, where $T_{X}$ is a~subgraph of $T$ induced on vertices $v$ such that $B(v) \cap X \neq \emptyset$. As $L_1, \ldots, R_4$ induce connected subgraphs of $G_k$, consequently we have that $T_{L_1}, \ldots, T_{R_4}$ are connected as well. By a known property, if it had been the case that all pairs of them intersected, there would have to be a~single bag that intersects all of them. However, that cannot be the case as that bag would have to have size at least seven, while all bags in $\cal T$ have size six. Hence, there is a pair $(X, Y)$ of subsets $L_1, \ldots, R_4$, such that $T_X$ and $T_Y$ do not intersect. This pair cannot be of the form $(L_i, R_j)$, because then there would be no edge connecting $L_i$ and $R_j$ contradicting our assumptions. Let us firstly assume that this is a pair of the form $(L_i, L_j)$. There exists a unique shortest path $P$ in $T$ that connects $T_{L_i}$ and $T_{L_j}$ and it has at least one edge. All of $T_{R_1}, \ldots, T_{R_4}$ have to intersect both $T_{L_i}$ and $T_{L_j}$ and consequently they all have to contain $P$. However, this implies that adhesions along $P$ are all of size at least four, which is a contradiction with the fact that adhesions in $\cal T$ are of size at most two. The case, where the not connected pair is of the form $(R_i, R_j)$ is analogous, with the difference that in the final argument we get adhesions of size at least three, instead of four, which is still a contradiction. That concludes the proof of \Cref{it:com-lb3}.

\item By a contradiction with \Cref{it:com-lb4}, let us assume that $K_5$ is a minor of $G_k$, that is, there exist five disjoint subsets $V_1, V_2, V_3, V_4$ and $V_5$ (called \textit{branch sets}) of $V(G_k)$ such that each of them induces a connected subgraph of $G_k$ and for each pair of them, there exists an edge connecting this pair of subsets. Through a folklore argument, there has to be a bag of $\cal{T}$ that contains at least one vertex for each of the five branch sets let $\hat{B}$ be such a~bag. As $\hat{B}$ induces a $C_6$ there are at most $6$ pairs of branch sets that can be connected by an edge within $\hat{B}$. Let $(V_i, V_j)$ be a pair that is not connected by an edge within $\hat{B}$. There have to be at least $|E(K_5)| -  |E(G_k[\hat{B}])|=10-6=4$ pairs like this. If they are not connected by an edge within $\hat{B}$, there exists another bag containing vertices from both of them, which implies that there exists an adhesion of $\hat{B}$ with one of its neighboring bags that contains vertices from both $V_i$ and $V_j$. However, each adhesion is of size two and the degrees of all vertices in $T$
are at most three, hence there are at most three pairs like that, showing a contradiction with the fact that there needed to be at least four pairs like this. That concludes the proof of \Cref{it:com-lb4}.

\item For the proof of \Cref{it:com-lb5} let us note that $G_3$ contains a topological minor of $K_{3, 3}$, as shown on the \Cref{fig:Soko2}, showing that its Euler genus is at least $1$. As $G_k$ for $k \ge 3$ contains $2^{k-3}$ disjoint copies of $G_3$ (ignoring weights) and each of them has Euler genus at least $1$, it follows that the Euler genus of $G_k$ is at least $2^{k-3}$. On the other hand, if for each $v \in V(T_1) \setminus \{r_1\}$ we add a cross-cap to an embedding as in \Cref{fig:Soko}
to accommodate an intersection of the edges $(v, I(B(v)))$ and $(B(v), I(v))$, 
we obtain an embedding of $G_k$ into a surface of Euler
genus equal to the number of such pairs, which is $2^k-1$. Hence, the Euler genus of $G_k$ is at most $2^k-1$, proving our claim.

\end{enumerate}

\end{proof}

\begin{figure}[ht] 
\centering

\begin{tikzpicture}[
    v/.style={circle,fill,inner sep=1.5pt},
    tree/.style={black,thick},
    match/.style={red,thick}
]

\def\ytop{3}
\def\yone{2}
\def\ytwo{1}
\def\ythree{0}

\def\ybthree{-1}
\def\ybtwo{-2}
\def\ybone{-3}
\def\ybbot{-4}
\def\dx{4}
\def\dxx{1.25}
\node[v, label=above:$3$] (t0)  at (0,\ytop) {};
\node[v, label=above:$2$] (t1)  at (-3,\yone) {};
\node[v, label=above:$4$] (t2)  at (3,\yone) {};

\node[v, label=above:$1$] (t3)  at (-4.5,\ytwo) {};
\node[v, label=above:$3$] (t4)  at (-1.5,\ytwo) {};
\node[v, label=above:$5$] (t5)  at (1.5,\ytwo) {};
\node[v, label=above:$5$] (t6)  at (4.5,\ytwo) {};

\node[v, label=above:$0$] (t7)  at (-5.25,\ythree) {};
\node[v, label=above:$2$] (t8)  at (-3.75,\ythree) {};
\node[v, label=above:$4$] (t9)  at (-2.25,\ythree) {};
\node[v, label=above:$4$] (t10) at (-0.75,\ythree) {};
\node[v, label=above:$6$] (t11) at (0.75,\ythree) {};
\node[v, label=above:$6$] (t12) at (2.25,\ythree) {};
\node[v, label=above:$6$] (t13) at (3.75,\ythree) {};
\node[v, label=above:$6$] (t14) at (5.25,\ythree) {};

\node[v, label=above:$4$] (b0)  at (0,\ybbot) {};
\node[v, label=above:$5$] (b1)  at (-3,\ybone) {};
\node[v, label=above:$3$] (b2)  at (3,\ybone) {};

\node[v, label=above:$6$] (b3)  at (-4.5,\ybtwo) {};
\node[v, label=above:$4$] (b4)  at (-1.5,\ybtwo) {};
\node[v, label=above:$4$] (b5)  at (1.5,\ybtwo) {};
\node[v, label=above:$4$] (b6)  at (4.5,\ybtwo) {};

\node[v, label=above:$7$] (b7)  at (-5.25,\ybthree) {};
\node[v, label=above:$5$] (b8)  at (-3.75,\ybthree) {};
\node[v, label=above:$5$] (b9)  at (-2.25,\ybthree) {};
\node[v, label=above:$5$] (b10) at (-0.75,\ybthree) {};
\node[v, label=above:$5$] (b11) at (0.75,\ybthree) {};
\node[v, label=above:$5$] (b12) at (2.25,\ybthree) {};
\node[v, label=above:$5$] (b13) at (3.75,\ybthree) {};
\node[v, label=above:$5$] (b14) at (5.25,\ybthree) {};

\draw[tree] (t0)--(t1) (t0)--(t2);
\draw[tree] (t1)--(t3) (t1)--(t4);
\draw[tree] (t2)--(t5) (t2)--(t6);
\draw[tree] (t3)--(t7) (t3)--(t8);
\draw[tree] (t4)--(t9) (t4)--(t10);
\draw[tree] (t5)--(t11) (t5)--(t12);
\draw[tree] (t6)--(t13) (t6)--(t14);

\draw[tree] (b0)--(b1) (b0)--(b2);
\draw[tree] (b1)--(b3) (b1)--(b4);
\draw[tree] (b2)--(b5) (b2)--(b6);
\draw[tree] (b3)--(b7) (b3)--(b8);
\draw[tree] (b4)--(b9) (b4)--(b10);
\draw[tree] (b5)--(b11) (b5)--(b12);
\draw[tree] (b6)--(b13) (b6)--(b14);

\draw[match]
(t1) .. controls ($(t1)+(\dx,0)$) and ($(b2)+(-\dx,0)$) .. node[pos=0.25,above] {$1$} (b2);

\draw[match]
(t2) .. controls ($(t2)+(-\dx,0)$) and ($(b1)+(\dx,0)$) .. node[pos=0.25,above] {$1$}(b1);

\draw[match]
(t3) .. controls ($(t3)+(\dxx,0)$) and ($(b4)+(-\dxx,0)$) .. node[pos=0.25,above] {$3$}(b4);

\draw[match]
(t4) .. controls ($(t4)+(-\dxx,0)$) and ($(b3)+(\dxx,0)$) .. node[pos=0.25,above] {$3$}(b3);

\draw[match]
(t5) .. controls ($(t5)+(\dxx,0)$) and ($(b6)+(-\dxx,0)$) .. node[pos=0.25,above] {$3$}(b6);

\draw[match]
(t6) .. controls ($(t6)+(-\dxx,0)$) and ($(b5)+(\dxx,0)$) .. node[pos=0.25,above] {$3$}(b5);

\draw[match, left=30]  (t7) to node[pos=0.25,above] {$5$} (b8);
\draw[match, right=30] (t8) to node[pos=0.25,above] {$5$} (b7);
\draw[match, left=30]  (t9) to node[pos=0.25,above] {$5$} (b10);
\draw[match, right=30] (t10) to node[pos=0.25,above] {$5$} (b9);
\draw[match, left=30]  (t11) to node[pos=0.25,above] {$5$} (b12);
\draw[match, right=30] (t12) to node[pos=0.25,above] {$5$} (b11);
\draw[match, left=30]  (t13) to node[pos=0.25,above] {$5$} (b14);
\draw[match, right=30] (t14) to node[pos=0.25,above] {$5$} (b13);

\end{tikzpicture}
\caption{A drawing of $G_3$. The isomorphism $I$ corresponds to mirroring the upper tree with respect to a horizontal line. Matching edges are drawn in red and their weights are displayed. To each vertex we associate its distance from the leftmost leaf of the upper tree. }
\label{fig:Soko}
\end{figure}
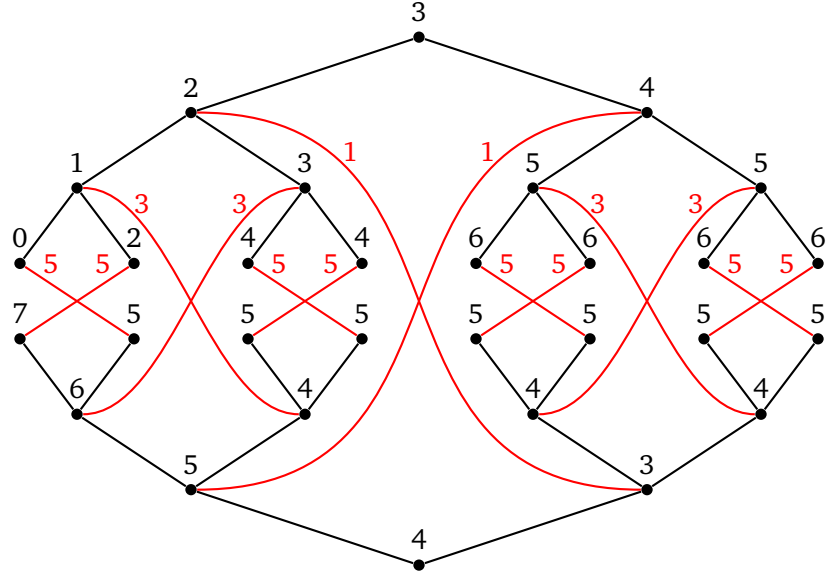

\begin{figure}[ht] 
\centering

\begin{tikzpicture}[
    v/.style={circle,fill,inner sep=1.5pt},
    L/.style={circle, draw=black, inner sep=3pt},
    R/.style={circle, fill=black, inner sep=3pt},
    tree/.style={black,thick, dashed},
    selected/.style={black,very thick}
]

\def\ytop{3}
\def\yone{2}
\def\ytwo{1}
\def\ythree{0}

\def\ybthree{-1}
\def\ybtwo{-2}
\def\ybone{-3}
\def\ybbot{-4}
\def\dx{4}
\def\dxx{1.25}
\node[v] (t0)  at (0,\ytop) {};
\node[R] (t1)  at (-3,\yone) {};
\node[v] (t2)  at (3,\yone) {};

\node[L] (t3)  at (-4.5,\ytwo) {};
\node[L] (t4)  at (-1.5,\ytwo) {};
\node[v] (t5)  at (1.5,\ytwo) {};
\node[v] (t6)  at (4.5,\ytwo) {};

\node[v] (t7)  at (-5.25,\ythree) {};
\node[v] (t8)  at (-3.75,\ythree) {};
\node[v] (t9)  at (-2.25,\ythree) {};
\node[v] (t10) at (-0.75,\ythree) {};
\node[v] (t11) at (0.75,\ythree) {};
\node[v] (t12) at (2.25,\ythree) {};
\node[v] (t13) at (3.75,\ythree) {};
\node[v] (t14) at (5.25,\ythree) {};

\node[v] (b0)  at (0,\ybbot) {};
\node[L] (b1)  at (-3,\ybone) {};
\node[v] (b2)  at (3,\ybone) {};

\node[R] (b3)  at (-4.5,\ybtwo) {};
\node[R] (b4)  at (-1.5,\ybtwo) {};
\node[v] (b5)  at (1.5,\ybtwo) {};
\node[v] (b6)  at (4.5,\ybtwo) {};

\node[v] (b7)  at (-5.25,\ybthree) {};
\node[v] (b8)  at (-3.75,\ybthree) {};
\node[v] (b9)  at (-2.25,\ybthree) {};
\node[v] (b10) at (-0.75,\ybthree) {};
\node[v] (b11) at (0.75,\ybthree) {};
\node[v] (b12) at (2.25,\ybthree) {};
\node[v] (b13) at (3.75,\ybthree) {};
\node[v] (b14) at (5.25,\ybthree) {};

\draw[tree] (t0)--(t1) (t0)--(t2);
\draw[selected] (t1)--(t3) (t1)--(t4);
\draw[tree] (t2)--(t5) (t2)--(t6);
\draw[selected] (t3)--(t7);
\draw[tree] (t3)--(t8);
\draw[selected] (t4)--(t9);
\draw[tree] (t4)--(t10);
\draw[tree] (t5)--(t11) (t5)--(t12);
\draw[tree] (t6)--(t13) (t6)--(t14);

\draw[selected] (b0)--(b1) (b0)--(b2);
\draw[selected] (b1)--(b3) (b1)--(b4);
\draw[tree] (b2)--(b5) (b2)--(b6);
\draw[tree] (b3)--(b7);
\draw[selected] (b3)--(b8);
\draw[tree] (b4)--(b9);
\draw[selected] (b4)--(b10);
\draw[tree] (b5)--(b11) (b5)--(b12);
\draw[tree] (b6)--(b13) (b6)--(b14);

\draw[selected]
(t1) .. controls ($(t1)+(\dx,0)$) and ($(b2)+(-\dx,0)$) .. (b2);

\draw[tree]
(t2) .. controls ($(t2)+(-\dx,0)$) and ($(b1)+(\dx,0)$) .. (b1);

\draw[selected]
(t3) .. controls ($(t3)+(\dxx,0)$) and ($(b4)+(-\dxx,0)$) .. (b4);

\draw[selected]
(t4) .. controls ($(t4)+(-\dxx,0)$) and ($(b3)+(\dxx,0)$) .. (b3);

\draw[tree]
(t5) .. controls ($(t5)+(\dxx,0)$) and ($(b6)+(-\dxx,0)$) .. (b6);

\draw[tree]
(t6) .. controls ($(t6)+(-\dxx,0)$) and ($(b5)+(\dxx,0)$) .. (b5);

\draw[selected, left=30]  (t7) to (b8);
\draw[tree, right=30] (t8) to (b7);
\draw[selected, left=30]  (t9) to (b10);
\draw[tree, right=30] (t10) to (b9);
\draw[tree, left=30]  (t11) to (b12);
\draw[tree, right=30] (t12) to (b11);
\draw[tree, left=30]  (t13) to (b14);
\draw[tree, right=30] (t14) to (b13);

\end{tikzpicture}
\caption{A drawing of $G_3$ with a topological minor model of $K_{3, 3}$ highlighted. }
\label{fig:Soko2}
\end{figure}
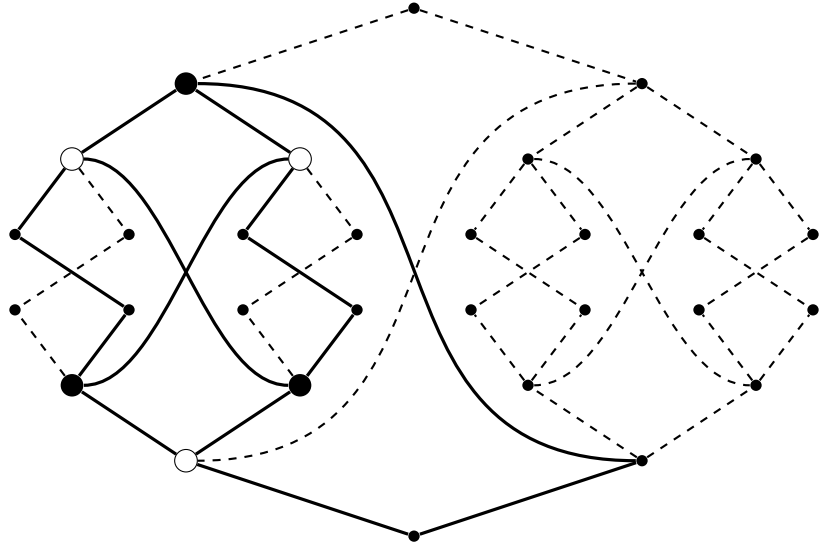

\subsection{Large comatchings in trees: Proof of \Cref{thm:kcenter-lb}}

We restate \Cref{thm:kcenter-lb}:

\KKComatching*

\begin{proof}
Let $k \geq 1$ be an integer.
The graph $G_k$ is constructed as follows. 
First, start with a full binary tree $T_k$ of depth $k$.
Let $r$ be the root of $T_k$. 
Then, for every vertex $v \neq r$ of $T_k$, if $h$ is the distance from $v$ to $r$
in $T_k$, then we attach to $v$ a path $P_v$ of length $h$ to $v$ and
denote by $\hat{v}$ the other endpoint of $P_v$. 
This completes the description of the graph $G_k$.

Let $L$ be the set of leaves of $T_k$. For every $v \in L$, let $Q_v$
be the path from $v$ to $r$ in $T_k$. 
Let $Y_v = N_{T_k}(V(Q_v))$, that is, the set of neighbors in $T_k$ of the vertices
on $Q_v$. (In other words, to construct $Y_v$, we go from $v$ towards $r$
and insert into $Y_v$ the sibling of every visited vertex.
Finally, let $X_v = \{\hat{u}~|~u \in Y_v\}$. 
We claim that $(v,X_v)_{v \in L}$ satisfies the desired properties.
Clearly, $|L| = 2^k$.

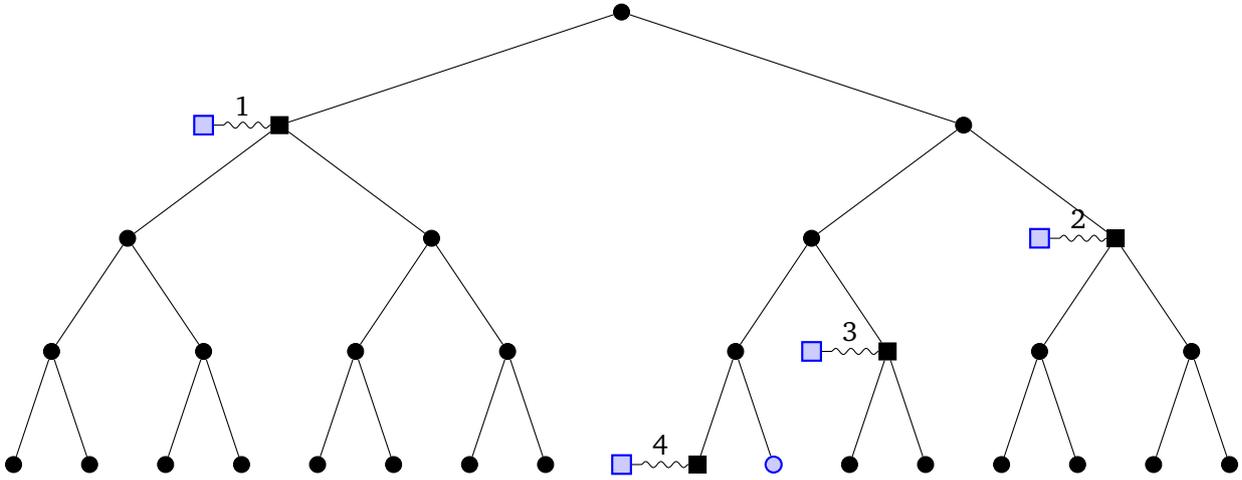
\begin{figure}[tb]
\centering
\begin{tikzpicture}
    \tikzstyle{vertex} = [circle, draw=black, fill=black, inner sep=0pt, minimum size=6pt]
    \tikzstyle{pvtx}  = [circle, draw=blue, thick, fill=blue!20, minimum size=6pt]
    \tikzstyle{xvtx}  = [rectangle, draw=black, thick, fill=black, minimum size=6pt]
    \tikzstyle{yvtx}  = [rectangle, draw=blue, thick, fill=blue!20, minimum size=6pt]
    \tikzstyle{every node} = [vertex]
    \tikzstyle{level 1} = [sibling distance=9cm]
    \tikzstyle{level 2} = [sibling distance=4cm]
    \tikzstyle{level 3} = [sibling distance=2cm]
    \tikzstyle{level 4} = [sibling distance=1cm]
    
    \node {}
        child { node[xvtx] (x1) {}
            child { node {}
                child { node {}
                    child { node {} }
                    child { node {} }
                }
                child { node {}
                    child { node {} }
                    child { node {} }
                }
            }
            child { node {}
                child { node {}
                    child { node {} }
                    child { node {} }
                }
                child { node {}
                    child { node {} }
                    child { node {} }
                }
            }
        }
        child { node {}
            child { node {}
                child { node {}
                    child { node[xvtx] (x4) {} }
                    child { node[pvtx] (v) {} }
                }
                child { node[xvtx] (x3) {}
                    child { node {} }
                    child { node {} }
                }
            }
            child { node[xvtx] (x2) {}
                child { node {}
                    child { node {} }
                    child { node {} }
                }
                child { node {}
                    child { node {} }
                    child { node {} }
                }
            }
        };
    \tikzstyle{every node} = []
    \foreach \n in {1,2,3,4} {
        \node[yvtx] (y\n) at ($ (x\n) + (-1, 0) $) {};
        \draw[decorate,decoration={snake,amplitude=.4mm,segment length=2mm,post length=1mm}] (x\n) -- node[above] {\n} (y\n);
    }
\end{tikzpicture}
\caption{Part of the graph $G_k$ of \Cref{thm:kcenter-lb} for $k=4$:
the tree $T_k$ and selected paths $P_u$.
Blue circle denotes one $v \in L$, blue squares denote the elements $\hat{u}$
of its corresponding $X_v$ and black squares denote the corresponding vertices $u$.}
\end{figure}

The analysis of the distances between the vertices of $L$ and the vertices $\hat{u}$
is encapsulated in the following lemma.
\begin{lemma}\label{lem:kcenter-lb-dist}
Let $v \in L$ and $u \in V(T_k)$. Then,
\[ \dist(v, \hat{u}) = k + 2\dist(u, V(Q_v)).\]
\end{lemma}
\begin{proof}
    Let $h = \dist(u, r)$.
    Clearly, $u$ lies on the unique path in $G_k$ from $v$ to $\hat{u}$. 
    Let $w$ be the lowest common ancestor of $u$ and $v$ in $T_k$; note 
    that $w$ is also the closest to $u$ vertex of $Q_v$ and
    $\dist(w, r) = h - \dist(u,w)$.
    Hence,
    \begin{align*}
        \dist(v, \hat{u}) &= \dist(v, w) + \dist(w, u) + \dist(u, \hat{u}) \\
        &= (k-h+\dist(u,w)) + \dist(u,w) + h\\
        &= k + 2\dist(u, V(Q_v)).
    \end{align*}
    This finishes the proof. 
\end{proof}
Lemma~\ref{lem:kcenter-lb-dist} implies that for every $v \in L$
and $\hat{u} \in X_v$ we have $\dist(v, \hat{u}) = k+2$
as $u$ is within distance exactly $1$ from $V(Q_v)$ by definition. 

Furthermore, consider two distinct $v_1,v_2 \in L$.
Let $w$ be the lowest common ancestor of $v_1$ and $v_2$ in $T_k$
and let $u_1$ (resp. $u_2$) be the child of $w$ that is an ancestor
of $v_1$ (resp. $v_2$). Then, $u_1 \in V(Q_{v_1})$ and $\hat{u}_1 \in X_{v_2}$.
This implies that $\dist(v_1, X_{v_2}) \leq \dist(v_1, \hat{u}_1) = k$
by \Cref{lem:kcenter-lb-dist}.
\Cref{thm:kcenter-lb} follows.
\end{proof}

\subsection{Even larger comatchings in planar graphs: Proof of \Cref{thm:kcenter-lb2}}\label{ss:lb2}

Firstly, we restate \Cref{thm:kcenter-lb2}:

\SoullessTheorem*

\begin{proof}
Let $k,d \geq 1$ be integers.
Let $h \leq \min(k,d)$ and let $k = k_1+k_2+\ldots+k_h$, $d = d_1+d_2 + \ldots d_h$
such that each $k_i$ and $d_i$ is a positive integer. 
For these integers, we construct a graph $G$ as follows.

\paragraph{The cycle gadget.}
Let $i \in [h]$ and consider a cycle $C^i$ of length $1 + k_i(2d_i+1)$. 
For every vertex $v \in V(C^i)$, let $X_v^i$ be constructed as follows: split
the path $C^i \setminus \{v\}$ into $k_i$ paths of length $2d_i+1$ each
and insert into $X_v$ the middle vertex of each of the path. 
Clearly, $|X_v^i| = k_i$ and 
\begin{align}
 \dist_{C^i}(v, X_v^i) &= d_i+1 & \forall_{v \in V(C^i)}, \nonumber\\ 
 \dist_{C^i}(v, X_u^i) &\leq d_i & \forall{u,v \in V(C^i), u \neq v}.\label{eq:gadget:cycle}
\end{align}

\paragraph{The local gadget.}
Let $i \in [h]$. A \emph{local gadget} $H_i$ consists of:
\begin{itemize}
    \item a cycle $C^i$;
    \item a root vertex $r$ connected to every vertex of the cycle $C^i$ with a path of length $d_i$;
    \item for every $v \in C^i$, a pendant vertex $\hat{v}$, connected to $v$ with a
    path of length $\sum_{j=1}^{i-1} d_j$ (note that for $i=1$ we have just $\hat{v} = v$).
\end{itemize}
Note that the connections to the root $r$ do not distort the distances between vertices
$v \in V(C^i)$ in cycle gadgets and sets $X_u$, so~\eqref{eq:gadget:cycle} still holds
in a local gadget.

\paragraph{Arranging in a tree-like fashion.}
We now arrange local gadgets in a tree-like fashion.
A subtree of level $h$ is just a local gadget $H_h$ and the root of the subtree is the root of $H_h$.
For $1 \leq i < h$, a subtree of level $i$ consists of a local gadget $H_i$
plus, for every $v$ on the cycle $C^i$ in $H_i$, a copy of the subtree of level $i+1$
with its root identified with $v$. The root of the subtree of level $i$ is the root
of the local gadget $H_i$ used in its construction. 

A subtree of level $1$ is the final graph $G$.
Note that $G$ is planar, as the blocks of $G$ are local gadgets $H_i$
that are trivially planar.

\paragraph{Comatching definition.}
Let $L$ be the union of all vertices on cycles $C^h$ in all level-$h$ local gadgets
in the construction. Note that
\begin{equation}\label{eq:lb:sizeL}
|L| = \prod_{i=1}^h \left(1+ k_i(2d_i+1)\right).
\end{equation}

For a vertex $v \in L$, we denote $v^h = v$ and for $1 \leq i < h$
by $v^i$ we denote the root of the local gadget $H_{i+1}$ on the path upwards from 
$v$ in the construction.
As the shortest path from $v$ to $v^i$ consists of the paths from 
$v^j$ to $v^{j-1}$ (which is of length $d_j$) for $j=h,h-1,\ldots,i+1$,
we have
\begin{equation}\label{eq:lb:vi-dist}
 \dist(v, v^i) = \sum_{j=i+1}^h d_j.
\end{equation}

For every $v \in L$, construct a set $X_v$ as follows.
For every $i \in [h]$, look at the local gadget $H_i$ on the path upwards from $v$ in the construction. Let $C^i$ be the cycle in $H_i$. For every $u \in X^i_{v^i}$, insert $\hat{u}$ into $X_v$.

Observe that
\begin{equation}\label{eq:lb:sizeX}
|X_v| = \sum_{i=1}^h k_i = k.
\end{equation}

We claim that $(v,X_v)_{v \in L}$ is a $(k,d)$-comatching. We verify this in the next two lemmata.
\begin{lemma}\label{lem:lb:dist}
For every $v \in L$ , $\dist(v, X_v) \ge d+1$.
\end{lemma}
\begin{proof}
    Let $\hat{u} \in X_v$.
    Assume that $\hat{u}$ is the pendant vertex attached to a vertex $u$
    of the cycle $C^i$ in a local gadget $H_i$ on the upward path from $v$ in the construction.
    Note that both $v^i$ and $u$ separate $v$ from $\hat{u}$ in $G$, thus
    \[ \dist(v, \hat{u}) = \dist(v, v^i) + \dist(v^i, u) + \dist(u, \hat{u}). \]
    The first term is computed in~\eqref{eq:lb:vi-dist}.
    For the second term, by~\eqref{eq:gadget:cycle}, as $u \in X^i_{v^i}$,
    \[ \dist(v^i, u) \ge d_i + 1. \]
    Finally, the path from $u$ to $\hat{u}$ is of length
    \[ \dist(u, \hat{u}) = \sum_{j=1}^{i-1} d_j. \]
    The lemma follows as $d = \sum_{j=1}^h d_j$.
\end{proof}

\begin{lemma}\label{lem:lb:dist2}
For every two distinct $v,w \in L$, we have $\dist(v, X_w) \leq d$.
\end{lemma}
\begin{proof}
    Let $H_i$ be the local gadget that is the lowest common ancestor of $v$ and $w$
    in the construction. 
    As $v \neq w$, we have $v^i \neq w^i$.
    By~\eqref{eq:gadget:cycle}, there exists $u \in X^i_{w^i}$
    with $\dist(v^i, u) \leq d_i$. We claim that $\dist(v, \hat{u}) \leq d$;
    this proves the lemma as $\hat{u} \in X_w$.
    Clearly,
    \[ \dist(v, \hat{u}) \leq \dist(v, v^i) + \dist(v^i, u) + \dist(u, \hat{u}).\]
    The first term is computed in~\eqref{eq:lb:vi-dist}.
    The second term is at most $d_i$ by the choice of $u$. 
    Finally, the path from $u$ to $\hat{u}$ is of length
    \[ \dist(u, \hat{u}) = \sum_{j=1}^{i-1} d_j. \]
    The lemma follows as $d = \sum_{j=1}^h d_j$.
\end{proof}

\paragraph{Choosing $t$, $k_i$ and $d_i$.}
For the final bound of Theorem~\ref{thm:kcenter-lb2},
we choose $h = \min(k,d)$ and split $k$ into $k_i$s and $d$ into $d_i$s as evenly as possible.
Then, if $k \leq d$, we have $h = k$, $k_i=1$ for each $i \in [h]$, and
\[ 1 + k_i(2d_i+1) = 2d_i+2 \geq 2\left\lfloor\frac{d}{k}\right\rfloor+2. \]
Otherwise, if $d \leq k$, we have $h=d$, $d_i=1$ for each $i \in [h]$, and
\[ 1+ k_i(2d_i+1) = 1+3k_i \geq 3\left\lfloor\frac{k}{d}\right\rfloor+1.\]
The size of the constructed comatching follows from~\eqref{eq:lb:sizeL}.
\end{proof}